\pdfoutput=1
\documentclass[acmsmall,screen,nonacm]{acmart}
\usepackage[utf8]{inputenc}
\usepackage[T1]{fontenc}
\usepackage[english]{babel}
\usepackage{mathtools}
\usepackage{xcolor}

\usepackage{halloweenmath}
\usepackage{xspace}
\usepackage{paralist}
\usepackage{xparse}
\usepackage{subcaption}
\usepackage{tabularx}

\newcommand{\smartparagraph}[1]{\subsubsection*{#1}}

\definecolor{colorCodeComment}{RGB}{80,170,0}
\definecolor{colorHighlight}{RGB}{40,40,220}
\definecolor{colorAnnotation}{RGB}{0,127,127}
\definecolor{colorContext}{RGB}{40,40,220}
\definecolor{backgroundFootprint}{RGB}{255,252,222}
\definecolor{backgroundContext}{RGB}{236,236,236}
\definecolor{colorFlow}{RGB}{196, 230, 235}
\definecolor{colorFlowChange}{RGB}{245,217,226}
\definecolor{colorGhost}{RGB}{132,69,32}

\newcommand{\colorbg}[2]{\smash{\setlength\fboxsep{0pt}\colorbox{#1}{\strut\hspace{1.5pt}#2\hspace{1.5pt}}}}

\usepackage{mathpartir}
\newcounter{inferencerule}
\newcommand{\rulelabel}[2]{%
  \def\theinferencerule{#2}%
  \refstepcounter{inferencerule}%
  \label[rule]{#1}%
}
\newcommand{\mkrulelabel}[1]{\textsc{(#1)}\rulelabel{rule:#1}{\textsc{(#1)}}}
\newcommand{\mkrulelabellab}[2]{\textsc{(#2)}\rulelabel{rule:#1}{\textsc{(#2)}}}
\newcommand{\inferH}[3]{\inferrule[\mkrulelabel{#1}]{#2}{#3}}
\newcommand{\inferHlab}[4]{\inferrule[\mkrulelabellab{#1}{#2}]{#3}{#4}}

\newcommand{\ruleref}[1]{\ref{rule:#1}}

\usepackage{stackengine}
\stackMath
\newcommand\xxrightarrow[1]{\raisebox{-.85pt}{\ensuremath{\smash{\mathrel{%
  \setbox2=\hbox{\stackon{\scriptstyle#1}{\scriptstyle#1}}%
  \stackon[-4.0pt]{%
    \xrightarrow{\makebox[\dimexpr\wd2\relax]{}}%
  }{%
   \scriptstyle#1\,%
  }%
}}}}}

\usepackage{listings}
\usepackage{upquote} 
\lstdefinelanguage{SPL}{
  morekeywords={method,struct,if,else,returns,procedure,def,requires,ensures,:=,var,let,
    new,old,elif,free,assume,assert,choose,havoc,
    predicate,function,invariant,while,return,atomic, split, type, field, result,
    define, datatype, domain, axiom, val, lock, unlock, not, restart, private, public, match, with, case},
  deletekeywords={union,int},
  numbers=left,
  xleftmargin=2em,
  escapeinside={@}{@},
  numberstyle=\tiny,
  basicstyle=\footnotesize\ttfamily,
  columns=flexible,
  morecomment=[s][\color{colorCodeComment}]{/*}{*/},
  morecomment=[l][\color{colorCodeComment}]{//},
  moredelim=[is][\color{colorGhost}]{`}{'},
  moredelim=[is][\underline]{^}{^},
  mathescape=true,
  aboveskip=0pt,
  belowskip=5pt,
  firstnumber=last,
}
\lstdefinestyle{codeDisplay}{
  language=SPL,
}
\usepackage{relsize}
\lstdefinestyle{codeInline}{
  basicstyle=\relscale{.9}\ttfamily,
  aboveskip=\smallskipamount,
  keywords={},
}

\newcommand{\code}[2][]{\lstinline[style=codeInline,#1]!#2!}
\newcommand{\mymathtt}[1]{\text{\relscale{.9}\ttfamily#1}}
\newcommand{\mcode}[1]{\ensuremath{\mymathtt{#1}}}

\newcommand{\annot}[1]{{\color{colorAnnotation}\bigl\{\,#1\,\bigr\}}}
\newcommand{\annotml}[1]{{\addtolength{\jot}{-2.5pt}\color{colorAnnotation}\left\{\,\begin{aligned}#1\end{aligned}\,\right\}}}
\newcommand{\ctx}[1]{{\color{colorContext}\bigl\{\,#1\,\bigr\}}}
\newcommand{\ctxml}[1]{{\addtolength{\jot}{-2.5pt}\color{colorContext}\left\{\,\begin{aligned}#1\end{aligned}\,\right\}}}

\usepackage[capitalise]{cleveref}
\crefformat{section}{\S#2#1#3} 
\crefmultiformat{section}{\S#2#1#3}{ and~\S#2#1#3}{, \S#2#1#3}{, and~\S#2#1#3} 
\crefrangeformat{section}{\S#3#1#4 to~\S#5#2#6} 

\crefname{rule}{rule}{rules}
\crefname{prop}{property}{properties}
\crefname{subfigure}{Fig.}{Figs.}
\Crefname{subfigure}{Figure}{Figures}

\makeatletter
\newcommand{\customlabel}[2]{%
   \protected@write \@auxout {}{\string \newlabel {#1}{{#2}{\thepage}{#2}{#1}{}} }%
   \hypertarget{#1}{}
}
\makeatother

\usepackage{tikz}
\usetikzlibrary{arrows,positioning,backgrounds,patterns,snakes,shapes,tikzmark,calc}

\tikzset{%
  treeelem/.style = {
    inner sep=0pt, text centered,
    thick, draw=black,
    scale=.8,
    anchor=north,
  },
  treenode/.style = {
    treeelem,
    circle, minimum width=6mm,
  },
  subtree/.style = {
    treeelem,
    isosceles triangle,  isosceles triangle stretches,
    shape border rotate=90, minimum width=8mm, minimum height=9mm,
  },
  treeptr/.style = {
    ->, >=stealth,
    line width=1.2pt,
    child anchor=north,
  },
  inflow/.style = {
    ->, >=stealth,
    line width=1.2pt,
    child anchor=north,
    dotted,draw=colorFlow!350,
  },
  flow/.style = {
    fill=colorFlow,
    draw=none,
    font={\footnotesize},
    inner sep=1.5pt,
  },
  flowch/.style = {
    flow,fill=colorFlowChange,
  },
  footprint/.style = {
    draw=none,
    fill=backgroundFootprint,
  },
  context/.style = {
    draw=none,
    fill=backgroundContext,
  },
  stackptr/.style = {
    ->, >=stealth,
    line width=0.5pt,
    child anchor=north,
    draw=black,
  },
  stackVar/.style={circle, fill=none, inner sep=0pt, minimum size=6mm, font=\normalsize, outer sep=-4pt},
  interval/.style = {above,rotate=42,anchor=south west,inner sep=0pt,font=\footnotesize},
  hinterval/.style = {font=\footnotesize},
}
\makeatletter
\newcommand{\gettikzxy}[3]{%
  \tikz@scan@one@point\pgfutil@firstofone#1\relax
  \edef#2{\the\pgf@x}%
  \edef#3{\the\pgf@y}%
}
\makeatother

\definecolor{colorMyGreen}{RGB}{80,170,0}
\definecolor{colorMyRed}{RGB}{180,20,20}
\definecolor{colorMyBlue}{RGB}{40,40,220}
\definecolor{colorMyPink}{RGB}{220,40,220}

\DeclareFontFamily{U}{MnSymbolC}{}
\DeclareSymbolFont{MnSyC}{U}{MnSymbolC}{m}{n}
\DeclareMathSymbol{\Diamonddot}{\mathbin}{MnSyC}{"7E}
\DeclareFontShape{U}{MnSymbolC}{m}{n}{
    <-6>  MnSymbolC5
   <6-7>  MnSymbolC6
   <7-8>  MnSymbolC7
   <8-9>  MnSymbolC8
   <9-10> MnSymbolC9
  <10-12> MnSymbolC10
  <12->   MnSymbolC12}{}

\def\logicspacing{\mskip 5.0mu plus 5.0mu}
\def\prallspacing{\mskip 2mu plus 2mu minus 3mu}
\newcommand{\prall}[1]{{\prallspacing#1\prallspacing}}

\newcommand{\mkmathrel}[2]{\newcommand#1{\mathrel{#2}}} 
\newcommand{\defineCallableOperator}[2]{
  \expandafter\newcommand\csname#1\endcsname{\operatorname{#2}}
  \expandafter\newcommand\csname#1of\endcsname[1]{\csname#1\endcsname(##1)}
  \expandafter\newcommand\csname#1Of\endcsname[1]{\csname#1\endcsname\left(##1\right)}
  \expandafter\newcommand\csname#1OF\endcsname[1]{\csname#1\endcsname\bigl(##1\bigr)}
}
\newcommand{\defineCallableOperatorBinary}[2]{
  \expandafter\newcommand\csname#1\endcsname{\operatorname{#2}}
  \expandafter\newcommand\csname#1of\endcsname[2]{\csname#1\endcsname(##1,##2)}
  \expandafter\newcommand\csname#1Of\endcsname[2]{\csname#1\endcsname\left(##1,##2\right)}
  \expandafter\newcommand\csname#1OF\endcsname[2]{\csname#1\endcsname\bigl(##1,##2\bigr)}
}

\newcommand{\nat}{\mathbb{N}}

\newcommand{\emp}{\mathsf{emp}}
\newcommand{\nullptr}{\mymathtt{null}}

\newcommand{\dontcare}{\bullet} 
\newcommand{\true}{\mathit{true}}
\newcommand{\false}{\mathit{false}}
\renewcommand{\emptyset}{\varnothing}

\newcommand{\set}[1]{\{\,#1\,\}}
\newcommand{\setcompact}[1]{\{#1\}}
\newcommand{\setc}[1]{\setcompact{#1}}
\newcommand{\setcond}[2]{\set{#1\:\mid\:#2}}

\mkmathrel{\defeq}{\triangleq} 
\mkmathrel{\eqdef}{\triangleq} 
\mkmathrel{\defebnf}{\Coloneqq} 
\mkmathrel{\defiff}{\mathrel{\vcentcolon\Leftrightarrow}} 
\mkmathrel{\defifff}{\mathrel{\vcentcolon\!\iff}} 
\mkmathrel{\bnf}{\ \mid\ } 
\newcommand{\ite}[3]{#1\:?\;#2\::\:#3}

\newcommand{\makeSpacedRel}[2]{
  \newcommand#1{{}\logicspacing{}#2{}\logicspacing{}}
}
\makeSpacedRel{\DEFEQ}{\defeq}
\makeSpacedRel{\EQDEF}{\eqdef}
\makeSpacedRel{\DEFEBNF}{\defebnf}
\makeSpacedRel{\DEFIFF}{\defiff}
\makeSpacedRel{\BNF}{\bnf}
\makeSpacedRel{\MSTAR}{\mstar}
\makeSpacedRel{\AND}{\wedge}
\makeSpacedRel{\OR}{\vee}

\newcommand{\powerset}[1]{\mathbb{P}(#1)}

 \newcommand{\domof}[1]{\mathit{dom}(#1)}
\newcommand{\project}[2]{#1|_{#2}}

\mkmathrel{\mstar}{\mathrel{*}}
\mkmathrel{\sepimp}{{\relbar}\mkern-8.25mu{\relbar}\mkern-2.5mu{\scalebox{.507}{\raisebox{2.5pt}{\ensuremath{\cdot}}}}\mkern-4.68mu{\mstar}} 
\mkmathrel{\septract}{{\relbar}\mkern-2.5mu{\circledast}}
\usepackage{scalerel}
\DeclareMathOperator*{\bigmstar}{\scalerel*{\ast}{\sum}}

\newcommand{\setstates}{\Sigma}
\newcommand{\munit}{1}
\mkmathrel{\mstardef}{\#}

\newcommand{\astate}{\mathsf{s}}
\newcommand{\astatep}{\mathsf{t}}

\newcommand{\apred}{\mathit{a}}
\newcommand{\apredp}{\mathit{b}}
\newcommand{\apredpp}{\mathit{c}}
\newcommand{\apredppp}{\mathit{d}}

\newcommand{\acom}{\mymathtt{com}}
\newcommand{\astmt}{\mymathtt{st}}
\newcommand{\cskip}{\mymathtt{skip}}
\newcommand{\seqof}[2]{#1\mathop{;}#2}
\newcommand{\choiceof}[2]{#1+#2}
\newcommand{\loopof}[1]{{#1}^{*}}

\newcommand{\semCom}[1]{\llbracket{#1}\rrbracket}

\newcommand{\sem}[1]{\semCom{#1}}
\newcommand{\semOf}[2]{\sem{#1}(#2)}
\newcommand{\semof}[2]{\sem{#1}(#2)}

\defineCallableOperator{now}{\_}
\defineCallableOperator{past}{\Diamonddot}
\defineCallableOperator{weakpast}{{\_\!\!\Diamonddot}}

\DeclareDocumentCommand\anobl{ g }{\mathsf{OBL}_{\IfValueT{#1}{#1}}}
\DeclareDocumentCommand\aful{ m g }{\mathsf{FUL}_{#1\IfValueT{#2}{,#2}}}

\newcommand{\hoareof}[3]{\set{#1}\:#2\:\set{#3}}
\newcommand{\highlight}[1]{#1}
\newcommand{\choareof}[4]{#1\; {\mid}\; \hoareof{#2}{#3}{#4}}
\newcommand{\chhoareof}[4]{#1\; \highlight{\mid}\; \hoareof{#2}{#3}{#4}}
\newcommand{\acontext}{\mathit{c}}

\newcommand{\theInterference}{\mathbb{I}}
\newcommand{\thePredicates}{\mathbb{P}}

\newcommand{\isInterferenceFreeOf}[2][\theInterference]{\boxast_{#1}\,#2}

\newcommand{\semcalc}{\Vdash}

\defineCallableOperator{acss}{\mathsf{CSS}}
\defineCallableOperator{ajob}{\mathsf{JOB}}
\defineCallableOperatorBinary{acssp}{\overline{\mathsf{CSS}}}

\newcommand{\abscontent}{\mathcal{C}}

\newcommand{\restrictto}[2]{#1|_{#2}}

\newcommand{\contfunof}[1]{\mathit{ContFun}(#1)}

\newcommand{\achain}{K}
\newcommand{\achainp}{L}

\newcommand{\atfun}{f}

\newcommand{\aninflow}{\inflow}

\newcommand{\inflowat}[1]{\inflow(#1)}

\newcommand{\amonoid}{\mathbb{M}}
\newcommand{\monadd}{+}

\newcommand{\monunit}{0}
\newcommand{\amonval}{\mathit{m}}
\newcommand{\amonvalp}{\mathit{n}}
\newcommand{\amonvalpp}{\mathit{o}}
\newcommand{\statemultdef}{\mathop{\#}}
\newcommand{\statemult}{\mathop{*}}

\newcommand{\setnodes}{\mathit{X}}
\newcommand{\setnodesp}{\mathit{Y}}

\newcommand{\anode}{\mathit{x}}
\newcommand{\anodep}{\mathit{y}}
\newcommand{\anodepp}{\mathit{z}}
\newcommand{\anodeppp}{\mathit{u}}

\newcommand{\edges}{\mathit{E}}

\newcommand{\edgesat}[2]{\edges_{(#1, #2)}}
\newcommand{\edgesatof}[3]{\edges_{(#1, #2)}(#3)}

\newcommand{\aflowconstraint}{\mathit{h}}
\newcommand{\setflowconstraints}{\mathit{FG}}
\newcommand{\fval}{\mathit{flow}}
\newcommand{\fvalof}[1]{\fval(#1)}
\newcommand{\inflow}{\mathit{in}}
\newcommand{\inflowof}[2][]{\inflow_{#1}(#2)}
\newcommand{\outflow}{\mathit{out}}
\newcommand{\outflowof}[2][]{\outflow_{#1}(#2)}

\newcommand{\transformerof}[1]{\mathit{tf}(#1)}
\newcommand{\transformerofof}[2]{\transformerof{#1}(#2)}

\newcommand{\lfpof}[1]{\mathit{lfp}.\;#1}
\newcommand{\pairingof}[2]{\langle #1, #2\rangle}

\newcommand{\atoolname}[1]{\code{#1}\xspace}
\newcommand{\nekton}{\atoolname{nekton}}

\usepackage{pifont}
\newcommand{\symbolYes}{\smash{\ding{51}}\xspace}

\newcommand{\theLogic}
{\textsf{Co(Co)SL}\xspace}
\newcommand{\theLogicSeq}
{\textsf{CoSL}\xspace}
\newcommand{\theLogicOG}
{\textsf{CoCoSL}\xspace}

\begin{document}
	\setdefaultenum{(i)}{(a)}{(a)}{(a)}
	\newtheorem{remark}{Remark}

	\title{Context-Aware Separation Logic}

	\author{Roland Meyer}
	\orcid{0000-0001-8495-671X}
	\affiliation{%
	  \institution{TU Braunschweig}
	  \country{Germany}
	}
	\email{roland.meyer@tu-bs.de}

	\author{Thomas Wies}
	\orcid{0000-0003-4051-5968}
	\affiliation{%
	  \institution{New York University}
	  \country{USA}
	}
	\email{wies@cs.nyu.edu}

	\author{Sebastian Wolff}
	\orcid{0000-0002-3974-7713}
	\affiliation{%
	  \institution{New York University}
	  \country{USA}
	}
	\email{sebastian.wolff@cs.nyu.edu}


\begin{abstract}
Separation logic is often praised for its ability to closely mimic the locality of state updates when reasoning about them at the level of assertions. The prover only needs to concern themselves with the footprint of the computation at hand, i.e., the part of the state that is actually being accessed and manipulated. Modern concurrent separation logics lift this local reasoning principle from the physical state to abstract ghost state. For instance, these logics allow one to abstract the state of a fine-grained concurrent data structure by a predicate that provides a client the illusion of atomic access to the underlying state. However, these abstractions inadvertently increase the footprint of a computation: when reasoning about a local low-level state update, one needs to account for its effect on the abstraction, which encompasses a possibly unbounded portion of the low-level state. Often this gives the reasoning a global character.

We present context-aware (concurrent) separation logic (\theLogic) to provide new opportunities for local reasoning in the presence of rich ghost state abstractions. \theLogic introduces the notion of a context of a computation, the part of the concrete state that is only affected on the abstract level. Contexts give rise to a new proof rule that allows one to reduce the footprint by the context, provided the computation preserves the context as an invariant. The context rule complements the frame rule of separation logic by enabling more local reasoning in cases where the predicate to be framed is known in advance. We instantiate our developed theory for the flow framework, enabling contextual reasoning about programs manipulating general heap graphs, and describe two other applications of the logic. We have implemented the flow instantiation of the logic in a concurrency proof outline checker and used it to verify two highly-concurrent binary search tree implementations with their maintenance operations.
\end{abstract}

	\maketitle


\newcommand{\invpred}{\mathit{inv}}

\newcommand{\join}{\sqcup}
\newcommand{\bigjoin}{\bigsqcup}
\newcommand{\aseq}[2]{\left(#2\right)_{#1\in\nat}}

\newcommand{\fiter}{\mathit{iter}}

\newcommand{\fprel}{\preceq}
\newcommand{\fpreldot}{\mathop{\dot{\preceq}}}
\newcommand{\fpreleq}{\mathop{\precapprox}}
\newcommand{\ctxfprel}{\preceq_\mathit{ctx}}
\newcommand{\makectxrel}[1]{#1_\mathit{ctx}}
\newcommand{\fpclosureof}[3][\fprel]{\mathrm{closure}_{#1}^{#2}(#3)}
\newcommand{\fpcompatible}[2][\fprel]{\mathrm{compatible}_{#1}(#2)}
\newcommand{\fpdecreasing}[2][\fprel]{\mathrm{decreasing}_{#1}(#2)}
\newcommand{\fpdisjoint}[2][\fprel]{\mathrm{disjoint}_{#1}(#2)}
\newcommand{\fpthresholdof}[3][\fprel]{\mathrm{threshold}_{#1}^{#2}(#3)}

\newcommand{\cval}{\mathit{cval}}
\newcommand{\myid}{\mathit{id}}
\newcommand{\myidof}[1]{\myid_{#1}}

\newcommand{\explain}[1]{\text{\small{(\;#1\;)}}\quad}
\newcommand{\explains}[2]{\stackrel{\text{\small{#1}}}{#2}}

\newcommand{\Bekic}{Beki\'c's Lemma\xspace}

\newcommand{\simplerel}{\Subset} 
\newcommand{\complexrel}{\mathrel{\Subset\!\!\!\!-}} 

\newcommand{\absreach}{\mathcal{R}}
\newcommand{\keyspace}{\mathbb{K}}
\newcommand{\contents}{\mathit{C}}
\newcommand{\ctxupclosed}[1]{\Psi(#1)}
\newcommand{\ctnof}[1]{\mathsf{C}(#1)}
\newcommand{\inset}{\mathsf{IS}}
\newcommand{\isof}[1]{\inset(#1)}
\newcommand{\keyset}{\mathsf{KS}}
\newcommand{\ksof}[1]{\keyset(#1)}
\newcommand{\oslof}[1]{\mathsf{OS}_\mathit{left}(#1)}
\newcommand{\osrof}[1]{\mathsf{OS}_\mathit{right}(#1)}
\newcommand{\ispof}[1]{\mathsf{\overline{IS}}(#1)}
\newcommand{\kspof}[1]{\mathsf{\overline{KS}}(#1)}
\newcommand{\loopinv}[1]{\mathsf{LoopInv}(#1)}
\newcommand{\invraw}{\mathsf{Inv}}
\newcommand{\inv}[1]{\invraw(#1)}
\newcommand{\invp}[1]{\mathsf{\overline{Inv}}(#1)}
\newcommand{\ninvraw}{\mathsf{NInv}}
\newcommand{\ninv}[1]{\mathsf{NInv}(#1)}
\newcommand{\ninvp}[1]{\mathsf{NInv}^{\mkern-2mu+\mkern-4mu}(#1)}
\newcommand{\lbof}[1]{\mathsf{LB}(#1)}
\newcommand{\ubof}[1]{\mathsf{UP}(#1)}
\newcommand{\pcof}[1]{???}
\newcommand{\rof}[1]{???}
\newcommand{\Root}{\mathit{Root}}
\newcommand{\key}{\mathit{key}}
\newcommand{\pred}{\mathit{pred}}
\newcommand{\curr}{\mathit{curr}}
\newcommand{\p}{\mathit{p}}
\newcommand{\x}{\mathit{x}}
\newcommand{\y}{\mathit{y}}
\newcommand{\z}{\mathit{z}}
\newcommand{\xkey}{\mathit{xk}}
\newcommand{\ykey}{\mathit{yk}}
\newcommand{\dup}{\mathit{c}}
\newcommand{\res}{\mathit{res}}
\newcommand{\aks}{\mathcal{K}}
\newcommand{\ais}{\mathcal{M}}
\newcommand{\aisp}{\mathcal{M'}}
\newcommand{\keyof}[1]{\mathit{key}(#1)}
\newcommand{\dataof}[1]{\keyof{#1}} 
\newcommand{\leftof}[1]{\mathit{left}(#1)}
\newcommand{\rightof}[1]{\mathit{right}(#1)}
\newcommand{\delof}[1]{\mathit{del}(#1)}
\newcommand{\inof}[1]{\mathit{in}(#1)}
\newcommand{\dupof}[1]{\mathit{dup}(#1)}
\newcommand{\dupvalnone}{\textsc{\sc no}}
\newcommand{\dupvalleft}{\textsc{\sc left}}
\newcommand{\dupvalright}{\textsc{\sc right}}
\newcommand{\remof}[1]{\mathit{rem}(#1)}
\newcommand{\rotrof}[1]{\mathit{rotR}(#1)}
\newcommand{\rotlof}[1]{\mathit{rotL}(#1)}
\newcommand{\lockof}[1]{\mathit{lock}(#1)}
\newcommand{\acqof}[1]{\mathsf{ACQ}(#1)}
\newcommand{\holds}[1]{\acqof{\lockof{#1}}}
\newcommand{\free}[1]{\mathsf{REL}(\lockof{#1})}

\newcommand{\pto}{\mapsto}
\newcommand{\upto}{\leadsto}
\newcommand{\INTER}[3]{\begin{aligned}[t]\set{&#1}\\\:&#2\:\\\set{&#3}\end{aligned}}
\newcommand{\data}{\mathit{val}}
\newcommand{\lchild}{\mathit{left}}
\newcommand{\rchild}{\mathit{right}}
\newcommand{\del}{\mathit{del}}
\newcommand{\rem}{\mathit{rem}}
\newcommand{\rotr}{\mathit{rotR}}
\newcommand{\rotl}{\mathit{rotL}}
\newcommand{\lock}{\mathit{lock}}
\newcommand{\is}{\mathsf{IS}}
\newcommand{\isp}{\mathsf{\overline{IS}}}
\newcommand{\lb}{\mathsf{LB}}
\newcommand{\ub}{\mathsf{UP}}

\newcommand{\bst}[1]{\mathsf{BST}(#1)}
\newcommand{\akey}{\mathit{k}}
\newcommand{\nodeof}[1]{\mathsf{N}(#1)}
\newcommand{\selof}[2]{#1\mcode{.#2}}


\newcommand{\astatepp}{\mathsf{u}}
\newcommand{\astateppp}{\mathsf{v}}
\newcommand{\arel}{\mathit{R}}
\newcommand{\arelof}[1]{\arel(#1)}
\newcommand{\arelp}{\mathit{S}}
\newcommand{\arelpof}[1]{\arelp(#1)}
\newcommand{\arelpp}{\mathit{T}}
\newcommand{\arelppof}[1]{\arelpp(#1)}
\newcommand{\acontextrelof}[2]{\arel_{#1}(#2)}
\newcommand{\abort}{\top}
\newcommand{\setpreds}{\mathsf{Preds}(\setstates)}
\newcommand{\setlocacts}{\mathsf{LActs}(\setstates)}
\newcommand{\predleq}{\sqsubseteq}
\newcommand{\predjoin}{\sqcup}
\newcommand{\predmeet}{\sqcap}
\newcommand{\bigpredjoin}{\bigsqcup}
\newcommand{\bigpredmeet}{\bigsqcap}
\newcommand{\predtransleq}{\mathrel{\dot\predleq}}
\newcommand{\predtransjoin}{\mathop{\dot\sqcup}}
\newcommand{\bigpredtransjoin}{\dot\bigsqcup}

\renewcommand{\highlight}[1]{\textcolor{blue}{#1}}
\newcommand{\makeColorLogic}[1]{\textcolor{blue}{#1}}
\newcommand{\makeColorLogicDep}[1]{\textcolor{red}{#1}}

\newcommand{\hoareOf}[3]{\hoareof{#1}{#2}{#3}}
\newcommand{\choareOf}[4]{\choareof{#1}{#2}{#3}{#4}}
\newcommand{\choareHighOf}[4]{\chhoareof{#1}{#2}{#3}{#4}}

\newcommand{\aprecond}{\mathit{p}}
\newcommand{\apostcond}{\mathit{q}}

\newcommand{\setcom}{\mathtt{COM}}

\newcommand{\setpredtrans}{\mathsf{PT}(\setstates)}
\newcommand{\setcapredtrans}{\mathsf{CAPT}(\setstates)}
\newcommand{\capredtransleq}{\ddot\sqsubseteq}
\newcommand{\capredtransjoin}{\ddot\sqcup}
\newcommand{\bigcapredtransjoin}{\ddot\bigsqcup}

\newcommand{\casem}[2]{\sem{#1}_{#2}}
\newcommand{\casemof}[3]{\casem{#1}{#2}(#3)}
\newcommand{\argument}[1]{\set{\text{#1}}}
\newcommand{\icasem}[2]{\sem{#1}_{#2}^{\mathrm{ind}}}
\newcommand{\icasemof}[3]{\icasem{#1}{#2}(#3)}

\newcommand{\overapprox}{\textsf{over}}
\newcommand{\overapproxof}[3]{\overapprox(\icasemof{#1}{#2}{#3})}

\newcommand{\asetpreds}{\mathit{P}}
\newcommand{\sizeof}[1]{|#1|}


\newcommand{\internal}{\scalebox{0.9}{$\mathghost$}}
\newcommand{\cor}{\mathop{\sim}}
\newcommand{\imult}{\mathop{{\internal}}}
\newcommand{\imultdef}{\mathop{\#_{\scalebox{.7}{\internal}}}}
\newcommand{\paraleq}[1]{\mathop{\preceq_{#1}}}
\newcommand{\paraclof}[2]{#1\!\uparrow_{#2}}
\newcommand{\ileq}{\mathop{<_{\internal}}}

\newcommand{\funs}{\contfunof{\amonoid}}
\newcommand{\discup}{\uplus}
\newcommand{\setflowgraphs}{\mathsf{FG}}

\newcommand{\astatex}{\mathsf{x}}
\newcommand{\astatey}{\mathsf{y}}

\newcommand{\abssem}[1]{\sem{#1}^{\#}}
\newcommand{\abssemof}[2]{\abssem{#1}(#2)}

\newcommand{\phyabs}{\mathop {\rightarrow^{\footnotesize\sharp}}}

\newcommand{\ghostconc}[1]{[\imult #1]}
\newcommand{\ghostconcof}[2]{\ghostconc{#1}(#2)}
\newcommand{\ghostabs}[1]{\ghostconc{#1}^{\sharp}}
\newcommand{\ghostabsof}[2]{\ghostabs{#1}(#2)}

\newcommand{\absimult}{\mathop{{{\internal}}^{\sharp}}}

\newcommand{\vertin}{\rotatebox{270}{$\!\!\!\in$}}
\newcommand{\verteq}{\rotatebox{90}{$=$}}

\newcommand{\up}[1]{\mathop{[#1]}}
\newcommand{\upof}[2]{\up{#1}(#2)}
\newcommand{\absup}[1]{\mathop{[#1]^{\sharp}}}
\newcommand{\absupof}[2]{\absup{#1}(#2)}


\newcommand{\localindex}{\mathsf{L}}
\newcommand{\sharedindex}{\mathsf{G}}
\newcommand{\sharedmult}{\mathop{{\mstar}_{\sharedindex}}}
\newcommand{\localmult}{\mathop{{\mstar}_{\localindex}}}
\newcommand{\sharedemp}{\emp_{\sharedindex}}
\newcommand{\localemp}{\emp_{\localindex}}
\newcommand{\setshared}{\Sigma_{\sharedindex}}
\newcommand{\setlocal}{\Sigma_{\localindex}}
\newcommand{\alocal}{\mathsf{l}}
\newcommand{\ashared}{\mathsf{g}}
\newcommand{\setstmt}{\mathtt{ST}}
\newcommand{\aconfig}{\mathsf{cf}}
\newcommand{\apc}{\mathsf{pc}}
\newcommand{\setconfig}{\mathsf{CF}}
\newcommand{\initset}[2]{\mathsf{Init}_{#1, #2}}
\newcommand{\acceptset}[1]{\mathsf{Acc}_{#1}}
\newcommand{\reachset}[1]{\mathsf{Reach}(#1)}
\newcommand{\reachsetof}[2]{\mathsf{Reach}_{#1}(#2)}
\newcommand{\pcStepOf}[3]{#1\,\xxrightarrow{\vphantom{pt}#2}\,#3}
\newcommand{\progStepRel}{\rightarrow}
\newcommand{\subModels}{\models}
\DeclareRobustCommand{\mmodels}{\mathrel{|\mkern-2mu|}\joinrel \Relbar}
\newcommand{\semCalc}{\semcalc}
\newcommand{\interop}{\rhd}


\section{Introduction}

Separation logic~\cite{DBLP:conf/csl/OHearnRY01,DBLP:conf/lics/Reynolds02} has had a formative influence on many modern program logics. Its success has been linked to its ability to reason locally about mutable state~\cite{DBLP:journals/cacm/OHearn19}. Assertions in separation logic denote physical resources such as memory locations and their contents. These resources can be composed using \emph{separating conjunction} to express disjointness constraints. Correctness judgments in the logic guarantee that a program does not access any resources that are not explicitly specified in the program's \emph{footprint}. Together, these characteristics give rise to the \emph{frame rule}, which allows one to conclude for free that any resource disjoint from the footprint is not affected by the program's execution. Thus, one can reason locally about only those parts of the program state that are relevant for the computation at hand.

Modern separation logics provide rich formalisms for layering abstractions on top of the physical resources manipulated by the program~\cite{DBLP:journals/jfp/JungKJBBD18, DBLP:conf/ecoop/Dinsdale-YoungDGPV10, DBLP:conf/ecoop/PintoDG14, DBLP:conf/pldi/GuSKWKS0CR18, DBLP:books/daglib/0034962}. These abstractions take the form of ghost resources that come equipped with their own fictional notion of separation, lifting the locality principle from the low-level program state all the way up to the level of abstract specifications of functional correctness properties.

Ghost resources induce a stronger notion of separation than mere disjointness on the abstracted physical resources~\cite{DBLP:journals/pacmpl/FarkaN0DF21}. As a consequence, the abstract footprint can comprise more physical resources than only those that the program directly manipulates. In fact, the footprint can become unbounded. For example, consider a ghost resource that abstracts a linked data structure by its contents. When the program inserts a new value into the structure, then the abstract effect will involve reasoning about the entire data structure state (because the ghost resource abstracts the whole structure), even though the insertion may only update a single memory location. Such unbounded ghost footprints also arise for other forms of ghost resources, e.g., when reasoning about future-dependent linearization points~\cite{DBLP:journals/pacmpl/JungLPRTDJ20,DBLP:journals/pacmpl/PatelKSW21} and space complexity bounds~\cite{DBLP:journals/pacmpl/MoineCP23}.

Thus, reasoning about ghost resources often deteriorates back to global reasoning about an unbounded set of physical resources. This global reasoning may, e.g., involve induction proofs for lemmas that are used to manipulate recursive predicates in the proof. While there has been much progress on automating such reasoning~\cite{DBLP:journals/jacm/CalcagnoDOY11, DBLP:conf/cade/BrotherstonDP11,DBLP:conf/pldi/PekQM14,DBLP:conf/sas/ToubhansCR14,DBLP:conf/nfm/EneaLSV17,DBLP:conf/cav/DardinierPWMS22,DBLP:journals/tocl/MathejaPZ23}, it remains a challenge for rich functional specifications and ghost resources that cannot be expressed in decidable theories.

\smartparagraph{Contributions.}
This paper aims to create new opportunities for local reasoning when dealing with computations that have unbounded footprints. We introduce \emph{context-aware (concurrent) separation logic (\theLogic)}. The key insight of \theLogic is that one can relax the locality requirement on the semantics of programs if the resources to be framed are known in advance. That is, in \theLogic one can frame a given \emph{context} $\acontext$ across a computation, provided that any changes affected on the resources in the context preserve $\acontext$. Intuitively, $\acontext$ can be subtracted from the footprint of the computation even though these resources may be subject to modification. We then present an abstract interpretation principle for computing appropriate contexts $\acontext$ to aid proof automation.

We describe several applications of context-aware reasoning. Our main application is a full instantiation of our approach to the flow framework~\cite{DBLP:journals/pacmpl/KrishnaSW18,DBLP:conf/esop/KrishnaSW20,DBLP:conf/tacas/MeyerWW23} to enable contextual reasoning about heap-manipulating programs and properties that are defined inductively over general heap graphs. In particular, this allows us to handle fine-grained concurrent search tree implementations featuring unbounded footprints due to intricate maintenance operations (e.g. removal of interior nodes) without the need for induction proofs.

To demonstrate the practical usefulness, we have implemented our approach in the proof outline checker \nekton~\cite{DBLP:conf/cav/MeyerOWW23}. We use the tool to verify the FEMRS tree~\cite{DBLP:conf/wdag/FeldmanE0RS18}, the contention-friendly binary search tree~\cite{DBLP:conf/europar/CrainGR13}, and the practical concurrent binary search tree~\cite{DBLP:conf/ppopp/BronsonCCO10}. Our proofs are the first formal proofs of these tree implementations. Beyond our verification effort, contextual reasoning applies in a similar fashion to a wide variety of concurrent search trees \cite{DBLP:conf/spaa/HowleyJ12,DBLP:conf/podc/EllenFRB10,DBLP:conf/ppopp/BrownER14,DBLP:conf/podc/ArbelA14,DBLP:conf/icdcn/RamachandranM15,DBLP:conf/ppopp/RamachandranM15,DBLP:conf/ppopp/Drachsler-Cohen18,DBLP:journals/topc/NatarajanRM20}. Overall, this makes contextual reasoning an indispensable technique for (semi-)automatic proofs.



\section{Motivation and Overview}
\label{sec:motivation}

We motivate our work by demonstrating how contextual reasoning can simplify linearizability proofs for concurrent data structures.
Such proofs often require ghost state to synchronize the linearization status of all threads, particularly when dealing with future-dependent linearization points (\Cref{sec:motivation:lin}), and to relate the logical contents of the structure to its physical representation (\Cref{sec:motivation:flow}).
Reasoning about these ghost state updates is challenging because they are frequently non-local to the actual physical updates performed by the program code and may involve an unbounded number of ghost resources.
We show how to decompose these complex ghost updates into a finite \emph{core} ghost update and the remaining \emph{context}.
To prove the core ghost update, we proceed as if the context was framed.
For the context, we employ a much simpler proof argument, namely that (the assertion describing) the context is \emph{invariant} under the update.


\newcommand{\tombstone}{\square}
\newcommand{\mcsstate}{\mathsf{DS}}
\newcommand{\mcslstate}{\mcsstate}
\newcommand{\Hist}{\mathit{Hist}}
\newcommand{\Status}{\mathit{Status}}
\newcommand{\obl}{\mathsf{OBL}}
\newcommand{\ful}{\mathsf{FUL}}
\newcommand{\slt}{\mathsf{SLT}}
\renewcommand{\anobl}[1]{\obl(#1)}
\renewcommand{\aful}[1]{\ful(#1)}
\newcommand{\aslt}[1]{\slt(#1)}
\newcommand{\astatus}[1]{\mathit{status}(#1)}
\newcommand{\Valid}{\mathit{Valid}}
\newcommand{\tid}{\mathit{tid}}
\newcommand{\latest}{\mathit{latest}}

\subsection{Linearizability with Helping}
\label{sec:motivation:lin}

We illustrate contextual reasoning for the purpose of linearizability proofs of concurrent data structure operations whose linearization points are future-dependent and potentially located in other threads.
Specifically, we focus on proofs that use prophecy variables and involve \emph{helping protocols} that govern the transfer of linearizability obligations between threads~\cite{DBLP:journals/pacmpl/JungLPRTDJ20,DBLP:journals/pacmpl/PatelKSW21}.

Concretely, we consider concurrent data structures that implement a (total) map $M$ from keys $K$ to values $V$.
For simplicity, assume a dedicated \emph{tombstone} value $\tombstone \in V$ that indicates the absence of a mapped value.
There are two types of operations: \code{search(k)} retrieves the value associated with $k$ in $M$ and \code{upsert(k,v)} updates the value of $k$ in $M$ to the new value $v$. 
We represent the data structure's physical state using an abstract predicate $\mcslstate(M)$.
So the goal is to prove that the operations are linearizable subject to the expected sequential specification:
\begin{alignat*}{2}
  & \hoareof{\mcslstate(M)}{&~\mcode{search}(k)\,~~&}{\mcode{$v$}.\; \mcslstate(M) * M(k)=v}
  \\
  & \hoareof{\mcslstate(M)}{&~\mcode{upsert}(k,v)~&}{\mcslstate(M[k \mapsto v])}\ .
\end{alignat*}
To do so, we can use the history $h \in (K \times V)^*$ of key/value pairs that have been upserted thus far as an intermediate abstraction of the physical state.
To be precise, $h$ induces the abstract state $M(h)$ that evaluates every key to the latest upserted value for that key or the tombstone if there is no upsert for the key.
The core aspect of the linearizability proof is carried out at this level of abstraction \cite{DBLP:journals/pacmpl/PatelKSW21}.
(Refer to \Cref{sec:motivation:flow} to see how relating the physical state to such an abstraction can also benefit from contextual reasoning.)

We focus on the linearizability argument for search threads.
A thread executing \code{search($k$)} may return value $v$ if either $M(h)(k)=v$ holds for the history $h$ when the search started or some \code{upsert($k$,$v$)} operation linearized during the execution of the search.
In the first case, the linearization point of \code{search($k$)} is right at the start of the operation.
In the second case, its linearization point coincides with the linearization point of the interfering \code{upsert($k$,$v$)} thread.

To enable thread modular reasoning about linearizability, the proof maintains a shared ghost state component that consists of a registry $R$.
The registry is a partial map from the thread IDs of all active search threads to their \emph{linearizability status}, $\anobl{k,v}$ or $\aful{k,v}$.
Status $\anobl{k,v}$ indicates that the thread (i) is searching for key $k$, (ii) it will return value $v$, and (iii) it still has the obligation to linearize.
Status $\aful{k,v}$ is similar but indicates that the thread has fulfilled its obligation to linearize.
The choice of the return value $v$ is implemented using a prophecy variable; we elide the details here.

Overall, the ghost state for the proof is a pair $(h, R)$ consisting of the current history $h$ and the registry $R$.
The actual code induces two kinds of updates to that ghost state: spawning a new search and linearizing an upsert.
When spawning a new \code{search($k$)}, an entry $\tid \mapsto s(k,v)$ for a fresh thread ID $\tid$ is added to the registry $R$, where $v$ is the thread's prophesied return value.
If $M(h)(k)=v$ then $s$ is chosen to be $\ful$ (the thread immediately linearizes) and otherwise $s=\obl$.
That is, the resulting ghost state is $(h, R\uplus\setcompact{\tid \mapsto s(k,v)})$.

When linearizing an \code{upsert($k$, $v$)}, a new pair $(k,v)$ is appended to the history $h$.
More importantly, the registry is updated to linearize other threads that are searching for key $k$ and expect value $v$.
To be precise, the new ghost state is $\bigl((k,v) \cdot h, R'\bigr)$ where $\cdot$ is the concatenation of histories and, for all threads $\tid$, $R'(\tid)=\aful{k,v}$ if $R(\tid)=\anobl{k,v}$ and $R'(\tid)=R(\tid)$ otherwise.
Note that this means we have $M((k,v) \cdot h)(k)=M(h)[k \mapsto v]=v$, so the sequential specification of \code{search($k$)} is satisfied and the search can indeed linearize.

Now, the actual proof of \code{upsert} operations has to deal with both the physical representation and the ghost state, i.e., with assertions $\mcslstate(M(h)) \mstar (h, R)$.
Consequently, for a command $\acom$ executing the linearization point of $\mcode{upsert}(k,v)$, the proof goal will be: \[
  \hoareof{\mcslstate(M(h)) \mstar (h, R)}{\acom}{\mcslstate(M(h)[k \mapsto v]) \mstar ((k,v) \cdot h, R')}
  \enspace .
\]
Notably, the entire proof has to deal with the registry $R$ although its updates are not relevant when updating the physical representation $\mcslstate(M)$.
Nevertheless, we cannot frame $R$ because separation logic does not allow the frame to be changed by $\acom$.
\begin{quote}\itshape
In short, due to the update of the ghost state, separation logic fails to localize the reasoning about the physical update.
\end{quote}

To alleviate this shortcoming of the frame rule, we approximate the exact registry $R$.
Towards this, we define the separating conjunction for the ghost state as $(h_1, R_1) \mstar (h_2, R_2) \defeq (h_1, R_1 \uplus R_1)$ if $h_1=h_2$ and $R_1,R_2$ are disjoint, leaving it undefined in all other case.
Then, rewrite $(h, R)$ into $(h, \emptyset) \mstar (h, R)$.
We use the former conjunct to keep track of the history.
The latter conjunct we approximate by a predicate $\acontext$ that corresponds to the smallest set of ghost states containing $(h, R)$ as well as all $(h'', R'')$ that result from $(h, R)$ by applying some sequence of search and upsert ghost updates, as discussed above.
By construction, $\acontext$ is stable under $\acom$: it denotes the ghost state $(h, R)$ from the precondition as well as the ghost state $((k,v) \cdot h, R')$ from the postcondition.
Note that despite this approximation, we can recover the desired registry $R'$ from computing $(h, \emptyset) \mstar \acontext$.
This leaves us with the following new proof goal: \[
  \hoareof{\mcslstate(M(h)) \mstar (h, \emptyset) \mstar \acontext}{\acom}{\mcslstate(M(h)[k \mapsto v]) \mstar ((k,v) \cdot h, \emptyset) \mstar \acontext}
  \enspace .
\]
Now, we treat $\acontext$ like a frame and ``remove'' it from the proof.
Technically, we do not use the frame rule.
Instead, we use a new context rule.
Like the frame rule, it allows us to ignore $\acontext$ and focus on the remaining parts of the proof.
Unlike the frame rule, we allow commands to modify the resources in $\acontext$.
To that end, we keep $\acontext$ syntactically in the proof tree and ensure its stability under updates of commands like $\acom$.
The result is a \highlight{context-aware} Hoare triple: \[
  \choareHighOf{\highlight{\acontext}}
  {\mcslstate(M(h)) \mstar (h, \emptyset)}{\acom}{\mcslstate(M(h)[k \mapsto v]) \mstar ((k,v) \cdot h, \emptyset)}
  \enspace .
\]
Applying this argument to the full proof of \code{upsert} allows us to focus on the updates of the physical representation.
While moving the registry to the context does not come for free, i.e., without any proof obligation, we observe that stability arguments are typically quite simple and may even be discharged upfront by reasoning over the semantics of commands rather than specific commands, just as we did when introducing the ghost state updates.
Hence, the context-aware Hoare triple that we are left with removes the need for reasoning about the registry altogether.



\newcommand{\treepred}{\mathsf{tree}}
\newcommand{\htreepred}{\mathsf{htree}}
\newcommand{\findSucc}{\mymathtt{findSucc}}
\newcommand{\remove}{\mymathtt{remove}}
\newcommand{\pointsto}{\mapsto}
\newcommand{\pnull}{\mathsf{null}}
\newcommand{\aframe}{\mathit{c}}

\subsection{Flow}
\label{sec:motivation:flow}

Contextual reasoning is also useful when relating the physical representation of a data structure to the ghost state that captures its logical contents.
To illustrate this, we use a binary search tree (BST) that implements a mathematical set.
In practical implementations, a BST will have distinct \code{insert} and \code{delete} operations, rather than the single \code{upsert} operation used in our high-level linearizability argument above (\cref{sec:motivation}).
We focus on \code{delete}, specifically the in-place removal of a key stored in an inner node of the tree.
This is the most interesting case of the operation.


\begin{figure}[t]
  \centering
\begin{minipage}{.35\linewidth}
  \begin{lstlisting}[language=SPL,gobble=4, basicstyle=\relscale{.9}\ttfamily]
    def remove($\anode$) {
      $p$, $\anodep$ := findSucc($\anode$);
      assume $p \neq \anode$;
      $\anode$.$\key$ := $\anodep$.$\key$;
      $p$.$\lchild$ := $\anodep$.$\rchild$;
    }
  \end{lstlisting}
\end{minipage}%
  \begin{minipage}{.4\linewidth}
  \begin{tikzpicture}[level/.style={sibling distance = 2.4cm/#1, level distance = 0.3cm}]

  \begin{scope}[local bounding box=g1]
  \node (x) [treenode] {$k$}
    child [treeptr] {
      node (A) [subtree] {$\contents_l$}
    }
    child [treeptr] {
      node (B) [subtree] {$\contents_r'$}
        child [sibling distance = 2cm, level distance = 0.3cm] {
          node (p) [treenode] {$i$} edge from parent[draw=none]
            child [treeptr,sibling distance = 3cm] {
              node (y) [treenode] {$j$} 
                child [treeptr,-|] {}
                child [treeptr,sibling distance = 1cm] {
                  node (D) [subtree] {$\contents_\anodeppp$}
                }
            }
            child [treeptr,sibling distance = 1.5cm] {
              node (C) [subtree] {$\contents_\anodepp$}
            }
        }
        child { node {} edge from parent[draw=none] }
    }
  ;

  \node[stackVar] (xv) at ($(x) + (-.7cm, .2cm)$) {\footnotesize$\anode$};
  \node[stackVar] (pv) at ($(p) + (-.7cm, .2cm)$) {\footnotesize$p$};
  \node[stackVar] (yv) at ($(y) + (-.7cm, .2cm)$) {\footnotesize$\anodep$};

  \path[treeptr,draw=black] (B.209) -- (p) ;
  \path[stackptr] (xv) -- (x.north west) ;
  \path[stackptr] (pv) -- (p.north west) ;
  \path[stackptr] (yv) -- (y.north west) ;
  \end{scope}
  \node at ($(g1.east)+(0.4,0)$) {~~~\scalebox{1.75}{$\rightsquigarrow$}};
\end{tikzpicture}
\end{minipage}%
\begin{minipage}{.3\linewidth}
  \begin{tikzpicture}[level/.style={sibling distance = 2.4cm/#1, level distance = 0.3cm}]

  \begin{scope}[local bounding box=g1]
  \node (x) [treenode] {$j$}
    child [treeptr] {
      node (A) [subtree] {$\contents_l$}
    }
    child [treeptr] {
      node (B) [subtree] {$\contents_r'$}
        child [sibling distance = 2cm, level distance = 0.3cm] {
          node (p) [treenode] {$i$} edge from parent[draw=none]
            child [treeptr,sibling distance = 1.5cm] {
              node (D) [subtree] {$\contents_\anodeppp$}
            }
            child [treeptr,sibling distance = 1.5cm] {
              node (C) [subtree] {$\contents_\anodepp$}
            }
        }
        child { node {} edge from parent[draw=none] }
    }
  ;

  \node[stackVar] (xv) at ($(x) + (-.7cm, .2cm)$) {\footnotesize$\anode$};
  \node[stackVar] (pv) at ($(p) + (-.7cm, .2cm)$) {\footnotesize$p$};

  \path[treeptr,draw=black] (B.209) -- (p) ;
  \path[stackptr] (xv) -- (x.north west) ;
  \path[stackptr] (pv) -- (p.north west) ;
  \end{scope}
  \node[stackVar] (foo) at ($(g1.south)+(0,-.3)$) {};
\end{tikzpicture}
\end{minipage}%
\caption{In-place removal of a key from an inner node $\anode$ in a binary search tree.\label{fig:bst-remove}}
\end{figure}


\Cref{fig:bst-remove} shows the code of the operation and illustrates how it changes the tree.
Each node in the tree is labeled with its key. The key $k$ to be removed is stored in node $\anode$.
The operation proceeds in four steps.
First, it uses the helper function $\findSucc$ to identify the left-most node $\anodep$ in the right subtree of $\anode$, as well as its parent $p$.
That is, $j$ is the next larger key stored in the tree after $k$.
We omit the definition of $\findSucc$.
The \code{assume} statement models a branching condition.
We focus on the case where $p \neq \anode$.
The operation copies the key $j$ from $\anodep$ to $\anode$, effectively removing $k$ from the structure.
Next, it unlinks $\anodep$ from the tree by setting $p$'s left pointer to the right child of $\anodep$.
This is to maintain the invariant that each key occurs at most once in the tree.
Finally, $\anodep$ is garbage collected.

Our goal is to demonstrate the functional correctness of the operation, meaning the operation updates the tree's contents from $\contents$ to $\contents \setminus \set{k}$.
A conventional proof in separation logic would use a recursive predicate to tie the data structure's physical representation to its contents $\contents$.
However, this approach has several disadvantages, especially for proof automation.
First, the prover needs to infer auxiliary inductive predicates to decompose the proof state into the footprint and the frame.
Next, the prover needs to derive auxiliary data-structure and property-specific lemmas to enable reasoning about the involved (auxiliary) predicates.
Finally, and perhaps most importantly, the proof does not easily generalize.
In a concurrent setting, threads may temporarily break the tree structure by introducing sharing, resulting in DAGs rather than trees.
Consequently, proofs can no longer rely on simple recursive predicates but require more complex machinery such as overlapping conjunctions~\cite{DBLP:conf/aplas/DockinsHA09,DBLP:conf/popl/GardnerMS12} and ramifications~\cite{DBLP:conf/popl/HoborV13}.

\smartparagraph{Node-local reasoning.}
An alternative to recursive predicates is to use indexed separating conjunction to describe unbounded heap regions~\cite{Yang01ShorrWaite, DBLP:conf/cav/0001SS16}.
These are predicates of the form $\bigmstar_{\anode \in \setnodes} \apred(x)$ and express that $\apred(\anode)$ must hold disjointly for all nodes $\anode \in \setnodes$.
The predicate $\apred(\anode)$ specifies a node-local property (e.g., constraining the values of a single points-to predicate for $\anode$).
Indexed separating conjunctions can be easily composed and decomposed along arbitrary partitions of $\setnodes$.
This greatly simplifies framing.
They can also be used to describe general graphs.
The recently proposed flow framework~\cite{DBLP:journals/pacmpl/KrishnaSW18,DBLP:conf/esop/KrishnaSW20,DBLP:conf/tacas/MeyerWW23} extends this approach so that $\apred(\anode)$ can capture global properties of the heap graph spanned by the nodes in $\setnodes$.
The approach works by augmenting every node with additional ghost information, its \emph{flow}.
Flows are computed inductively over the graph structure using a data-flow equation.
The equation can be thought of as collecting information about all possible traversals of the graph.
The definition is such that it still yields generic reasoning principles for decomposing and composing predicates similar to those for indexed separating conjunctions.

A suitable flow for verifying the functional correctness of our $\remove$ operation assigns to each node its \emph{inset}.
Intuitively, the inset of a node $\anode$ consists of the set of keys $k$ such that an operation on $k$ may traverse $\anode$ to find $k$.
\Cref{fig:bst-remove-flows} shows two search trees, before and after execution of the $\remove$ operation, with the inset of each node annotated in {\color{blue}blue}.
For example, the inset of $\anodep$ in the pre-state is the interval $\color{blue}(4,8)$ because the largest (highest up) key on the path from the $\Root$ to $\anodep$ when moving right is $4$ and the smallest key when moving left is $8$.

If we subtract from a node's inset all the insets of its children, we derive its \emph{keyset}.
For example, in the pre-state, the keyset of $p$ is $\{8\}$ and the keyset of $\anode$'s left child is $(-\infty,1]$. 
Assuming searches follow deterministic paths through the graph (as they do for binary search trees), then the keysets are pairwise disjoint~\cite{DBLP:journals/tods/ShashaG88}.
This means the keyset of a node $\anode$ consists of exactly those keys that can only be found in $\anode$ if they are stored anywhere in the structure.


\begin{figure}[t]
  \vspace*{-.2em}
  \centering
    \begin{minipage}{.55\linewidth}
  \begin{tikzpicture}[level/.style={sibling distance = 4cm/(1+(0.25*#1)), level distance = 0.3cm}]
    \begin{scope}[local bounding box=g1]
      
    \node (root) [treenode] {$\infty$}
    child [treeptr] {
      node (x) [footprint,treenode] {$4$}
        child [treeptr,sibling distance = 4.5cm] {
          node (A) [context,treenode] {$1$}
          child [sibling distance = 0cm] { node {} edge from parent[draw=none] }
          child [treeptr, level distance = 0.1cm] { node [context] (E) [treenode] {$3$} }
        }
        child [treeptr, sibling distance = 4.5cm] {
          node (B) [context,treenode] {$15$}
          child [sibling distance = 2cm, level distance = 0.3cm] {
            node (p) [footprint,treenode] {$8$}
            child [treeptr] {
              node (y) [footprint,treenode] {$6$} 
              child [treeptr,-|] {}
              child [treeptr] {
                node (D) [treenode] {$7$}
              }
            }
            child [treeptr,sibling distance = 1.5cm] {
              node (C) [treenode] {$9$}
            }
          }
          child { node (F) [treenode] {$18$} }
        }
    }
    child { node {} edge from parent[draw=none] }
    ;

    \node[stackVar,minimum width=8mm] (rootv) at ($(root) + (-.8cm, .17cm)$) {\footnotesize$\Root$};
    \node[stackVar] (xv) at ($(x) + (-.7cm, .2cm)$) {\footnotesize$\anode$};
    \node[stackVar] (pv) at ($(p) + (-.7cm, .2cm)$) {\footnotesize$p$};
    \node[stackVar] (yv) at ($(y) + (-.7cm, .2cm)$) {\footnotesize$\anodep$};

    \node[stackVar,minimum width=8mm] at ($(root) + (0, -.38cm)$) {\scriptsize\color{blue}$(-\infty,\infty]$};
    \node[stackVar,minimum width=8mm] at ($(x) + (0, -.38cm)$) {\scriptsize\color{blue}$(-\infty,\infty)$};
    \node[stackVar,minimum width=8mm] at ($(A) + (0, -.38cm)$) {\scriptsize\color{blue}$(-\infty,4)$};
    \node[stackVar,minimum width=8mm] at ($(E) + (0, -.38cm)$) {\scriptsize\color{blue}$(1,4)$};
    \node[stackVar,minimum width=8mm] at ($(B) + (0, -.38cm)$) {\scriptsize\color{blue}$(4,\infty)$};
    \node[stackVar,minimum width=8mm] at ($(p) + (0, -.38cm)$) {\scriptsize\color{blue}$(4,15)$};
    \node[stackVar,minimum width=8mm] at ($(y) + (0, -.38cm)$) {\scriptsize\color{blue}$(4,8)$};
    \node[stackVar,minimum width=8mm] at ($(D) + (0, -.38cm)$) {\scriptsize\color{blue}$(6,8)$};
    \node[stackVar,minimum width=8mm] at ($(C) + (0, -.38cm)$) {\scriptsize\color{blue}$(8,15)$};
    \node[stackVar,minimum width=8mm] at ($(F) + (0, -.38cm)$) {\scriptsize\color{blue}$(15,\infty)$};

    \path[stackptr] (rootv) -- (root.north west) ;
    \path[stackptr] (xv) -- (x.north west) ;
    \path[stackptr] (pv) -- (p.north west) ;
    \path[stackptr] (yv) -- (y.north west) ;
  \end{scope}
  \node at ($(g1.east)+(0.4,0)$) {~~~\scalebox{1.75}{$\rightsquigarrow$}};
\end{tikzpicture}
\end{minipage}%
    \begin{minipage}{.45\linewidth}
  \begin{tikzpicture}[level/.style={sibling distance = 4cm/(1+(0.25*#1)), level distance = 0.3cm}]
    \begin{scope}[local bounding box=g1]
      
    \node (root) [treenode] {$\infty$}
    child [treeptr] {
      node (x) [footprint,treenode] {$6$}
        child [treeptr,sibling distance = 4cm] {
          node (A) [context,treenode] {$1$}
          child [sibling distance = 0cm]{ node {} edge from parent[draw=none] }
          child [treeptr, level distance = 0.1cm] { node [context] (E) [treenode] {$3$} }
        }
        child [treeptr, sibling distance = 3.5cm] {
          node (B) [context,treenode] {$15$}
          child [sibling distance = 2cm, level distance = 0.3cm] {
            node (p) [footprint,treenode] {$8$}
            child [treeptr] {
              node (D) [treenode] {$7$}
            }
            child [treeptr,sibling distance = 1.5cm] {
              node (C) [treenode] {$9$}
            }
          }
          child { node (F) [treenode] {$18$} }
        }
    }
    child { node {} edge from parent[draw=none] }
    ;

    \node[stackVar,minimum width=8mm] (rootv) at ($(root) + (-.8cm, .17cm)$) {\footnotesize$\Root$};
    \node[stackVar] (xv) at ($(x) + (-.7cm, .2cm)$) {\footnotesize$\anode$};
    \node[stackVar] (pv) at ($(p) + (-.7cm, .2cm)$) {\footnotesize$p$};

    \node[stackVar,minimum width=8mm] at ($(root) + (0, -.38cm)$) {\scriptsize\color{blue}$(-\infty,\infty]$};
    \node[stackVar,minimum width=8mm] at ($(x) + (0, -.38cm)$) {\scriptsize\color{blue}$(-\infty,\infty)$};
    \node[stackVar,minimum width=8mm] at ($(A) + (0, -.38cm)$) {\scriptsize\color{blue}$(-\infty,{\color{colorMyGreen}6})$};
    \node[stackVar,minimum width=8mm] at ($(E) + (0, -.38cm)$) {\scriptsize\color{blue}$(1,{\color{colorMyGreen}6})$};
    \node[stackVar,minimum width=8mm] at ($(B) + (0, -.38cm)$) {\scriptsize\color{blue}$({\color{red}6},\infty)$};
    \node[stackVar,minimum width=8mm] at ($(p) + (0, -.38cm)$) {\scriptsize\color{blue}$({\color{red}6},15)$};
    \node[stackVar,minimum width=8mm] at ($(D) + (0, -.38cm)$) {\scriptsize\color{blue}$(6,8)$};
    \node[stackVar,minimum width=8mm] at ($(C) + (0, -.38cm)$) {\scriptsize\color{blue}$(8,15)$};
    \node[stackVar,minimum width=8mm] at ($(F) + (0, -.38cm)$) {\scriptsize\color{blue}$(15,\infty)$};

    \path[stackptr] (rootv) -- (root.north west) ;
    \path[stackptr] (xv) -- (x.north west) ;
    \path[stackptr] (pv) -- (p.north west) ;
  \end{scope}
  \node[stackVar] (foo) at ($(g1.south)+(0,-.3)$) {};
\end{tikzpicture}
\end{minipage}%
  \vspace{-1em}
  \caption{In-place removal of the key in $\anode$ on a tree augmented with {\color{blue}insets}.
  \label{fig:bst-remove-flows}}
\end{figure}


To reason about the functional correctness of the operations on the tree, we simply maintain the following \emph{keyset invariant}:  the key stored in each node is contained in the node's keyset.
The overall contents $\contents$ of the tree is the union of all keys stored in its nodes.
The keyset invariant together with the disjointness of the keysets imply that the node-local contents are also disjoint.
Hence, any change made to the contents of a node, such as replacing its key, is reflected by a corresponding change of the global contents $\contents$.
That is, we can now reason node-locally about the overall functional correctness of the search tree operations!

\smartparagraph{Unbounded footprints.}

To enable compositional reasoning about inductive properties, the flow framework adds an additional constraint on separating conjunction: two graphs augmented with flows compose only if their flow values are consistent with the flow obtained in the composite graph~\cite{DBLP:journals/pacmpl/KrishnaSW18,DBLP:conf/esop/KrishnaSW20,DBLP:conf/tacas/MeyerWW23}.
As a consequence, the footprint of an update on the graph can be larger than the \emph{physical} footprint that encompasses the changes to the graph structure (i.e., when ignoring the auxiliary ghost state).
In fact, the full footprint can be unbounded even if the physical footprint is not.

For the $\remove$ operation, the physical footprint consists of the three nodes $\anode$, $p$, and $\anodep$ (shaded yellow in \cref{fig:bst-remove-flows}).
However, observe that moving $\anodep$'s key to $\anode$ changes the insets of all the nodes shaded in gray.
These are the nodes on the path from $\anode$ to $\anodep$ as well as all nodes on the path from $\anode$ to the right-most leaf in its left subtree.
As these paths can be arbitrarily long, the footprint of the update is unbounded.

If we attempt to reason solely about the bounded physical footprint and put everything else into the frame, the proof will fail: after the update, the physical footprint no longer composes with the frame, as the insets of the two regions are inconsistent.
In a sense, the stronger notion of graph composition forces us to reconcile with the global effect of the update immediately at the point when the update occurs. Thus, reasoning about an update with an unbounded footprint appears to entail some form of quantifier instantiation or inductive argument, which adversely affects proof automation.

\begin{quote}\itshape
New reasoning techniques are needed to effectively handle unbounded footprints.
\end{quote}

Existing works on the flow framework have either considered only updates with bounded footprint~\cite{DBLP:conf/esop/KrishnaSW20,DBLP:conf/tacas/MeyerWW23} or cases where the unbounded footprint is traversed by the program prior to the update~\cite{DBLP:journals/pacmpl/MeyerWW22}.
However, not all updates fall into these categories as our example demonstrates.
In this paper, we provide a general solution.

Finally, we note that the issue of having to reason about large footprints is not unique to the flow framework or registry-like constructs.
It has been observed in the literature that this issue arises naturally whenever rich ghost state abstractions are layered on top of the physical state, thereby inducing a stronger notion of separation~\cite{DBLP:journals/pacmpl/Nanevski0DF19, DBLP:journals/pacmpl/FarkaN0DF21}.
This is why we formulate our solution in the setting of abstract separation logic~\cite{DBLP:conf/lics/CalcagnoOY07}, so that it can apply broadly.


\subsection{Contributions and Overview}

Our first contribution is \emph{context-aware (concurrent) separation logic (\theLogic)}, which we describe in \cref{Section:CAReasoning}.
\theLogic is a conservative extension of separation logic that enables local reasoning about computations with large footprints.
Hoare judgments in \theLogic take the form $\choareof{\acontext}{\apred}{\astmt}{\apredp}$.
The judgment decomposes the footprint of $\astmt$ into two parts: a core footprint $\apred$ and a \emph{context} $\acontext$.
In our registry example, the core footprint is $\mcslstate(M(h)) \mstar (h, \emptyset)$, meaning we focus on the physical state $\mcslstate(M(h))$ and maintain $(h, \emptyset)$ as minimimalistic information about the ghost state. 
The context $\acontext$ is the approximation of the registry that takes into account potential updates. 
In our flow example, the core footprint is the physical footprint of the update.
The context is a predicate describing the nodes shaded in gray.

Akin to the frame rule, if $\choareof{\acontext}{\apred}{\astmt}{\apredp}$ is valid, then $\astmt$ transforms $\apred \mstar \acontext$ to $\apredp \mstar \acontext$.
However, the frame rule has to work for all possible frames $\aframe$ and must therefore require that no state in $\aframe$ is affected by $\astmt$.
In contrast, when $\hoareof{\apred}{\astmt}{\apredp}$ is viewed in the context of $\acontext$, the logic can take advantage of the fact that $\acontext$ is known.
This enables new opportunities for local reasoning in the cases where the full footprint of $\astmt$ is large.
Intuitively, $\acontext$ is the part of the state whose ghost component may be affected by the update, but the ghost component changes in a way such that $\acontext$ is maintained.
In the registry example, the context is defined by a closure of the current registry under potential updates, and is therefore invariant under updates by construction.  
In the flow example, the important property being maintained is the keyset invariant.

Our second contribution addresses the question of how to derive appropriate context predicates $\acontext$.
More precisely, given a predicate $\apred \mstar \apredppp$ that describes the pre-states of a computation $\astmt$, the \emph{contextualization} problem is to identify predicates $\apredp$ and $\acontext$ such that $\choareof{\acontext}{\apred}{\astmt}{\apredp}$ and $\apredppp \Rightarrow \acontext$ are valid.
We propose a principled approach based on abstract interpretation that solves contextualization 
(\cref{sec:contextualization}).
The crux of the approach is to derive $\apredpp$ using an abstract semantics of $\astmt$.
By tailoring the abstraction to the specific ghost state and property of interest, one can derive simple reasoning principles for showing that $\astmt$ preserves $\apredpp$.
This style of reasoning enables better proof automation compared to proving $\hoareof{\acontext \mstar \apred}{\astmt}{\acontext \mstar \apredp}$ directly in standard separation logic.

We then instantiate this abstract solution for the concrete setting of the flow framework (\cref{sec:instantiation}).
The technical challenge here is that one needs to approximate a fixed point that is computed over the graphs in the image of $\apred \mstar \apredppp$ under $\astmt$, without precise information about what these graphs look like.
Our instantiation is motivated by the observation that, in practice, the change of the flow that emanates from the core footprint simply propagates through the context.
For instance, in the example shown in \cref{fig:bst-remove-flows}, the insets in the left subtree of $\anode$ uniformly increase by $[4,6)$ and in the right subtree they uniformly decrease by $[4,6)$.
In both cases, the change preserves the keyset invariant (which is the desired $\acontext$).
We identify general conditions under which the induced flow changes can be uniformly approximated.
In effect, this allows us to replace complex inductive reasoning to infer $\acontext$ from $\apredppp$ and $\astmt$ with simple local monotonicity reasoning about how the flow changes in the core footprint.

We have implemented our approach in the concurrency proof outline checker \nekton \cite{DBLP:conf/cav/MeyerOWW23} and used it to verify three highly concurrent binary search tree implementations.
It would be difficult to achieve the same degree of proof automation (using flows or recursive predicates) without contextualization due to the aforementioned challenges (unbounded footprints, DAG structures).


\section{Semantics}
\label{sec:semantics}

We briefly recall the setup of abstract separation logic~\cite{DBLP:conf/lics/CalcagnoOY07} which we adapt slightly as a basis for our formal development.

A separation algebra $(\setstates, \mstar, \emp)$ is a cancellative and
commutative monoid in which the multiplication $\mstar$ is only partially defined and we have a set of units $\emp$.
By cancellativity, we mean that if $\astate_1\mstar\astatep$ and $\astate_2\mstar\astatep$ are both defined and $\astate_1\mstar\astatep=\astate_2\mstar\astatep$, then $\astate_1=\astate_2$ follows.
For every state $\astate\in\setstates$, we require that there is a unit $\munit\in\emp$ with $\astate\mstar\munit=\astate$.
Moreover, for every pair of units $\munit\neq \munit'$ in $\emp$, we expect that the multiplication is undefined.
We use $\astate\statemultdef\astatep$ to denote definedness of the multiplication.

Predicates in the set $\setpreds=\powerset{\setstates}\cup\set{\abort}$ are sets of states or a dedicated symbol~$\abort$ that indicates a failure of a computation.
We extend the multiplication to predicates, then called separating conjunction, by defining $\apred\mstar \apredp=\setcond{\astate\mstar\astatep}{\astate\in\apred\wedge \astatep\in\apredp\wedge \astate\statemultdef\astatep}$ for $\apred, \apredp\subseteq\setstates$ and $\apred\mstar\abort=\abort\mstar\apred=\abort$ for $\apred\in\setpreds$.
We endow predicates with an ordering that coincides with inclusion on sets of states and has $\abort$ as the top element: $\apred\predleq\apredp$ if $\apred, \apredp\subseteq\setstates$ and~$\apred\subseteq\apredp$, and $\apred\predleq\abort$ for all $\apred\in\setpreds$.
Then $(\setpreds, \predleq)$ is a complete lattice and we use~$\bigpredjoin\asetpreds$ to denote the least upper bound, or simply join, of a set of predicates $\asetpreds\subseteq\setpreds$.

We say $\apred \subseteq \setstates$ is \emph{precise} if it identifies unique substates: for every $\astate \in \setstates$ there exists at most one $\astatep \in \apred$ such that $\astate \in \{\astatep\} \mstar \setstates$.

We define our programming language parametric in a set of commands $\setcom$.
The set is expected to come with a semantics
\begin{align*}
	\sem{-}: \setcom\rightarrow \setpredtrans
\end{align*}
that assigns to each $\acom\in\setcom$ a predicate transformer $\sem{\acom}\in\setpredtrans$.
The predicate transformers used in separation logic are functions $\sem{\acom}:\setpreds\rightarrow\setpreds$ that satisfy $\semof{\acom}{\abort}=\abort$ and $\semof{\acom}{\bigpredjoin\asetpreds}=\bigpredjoin \semof{\acom}{\asetpreds}=\bigpredjoin\setcond{\semof{\acom}{\apred}}{\apred\in\asetpreds}$ for all $\asetpreds\subseteq\setpreds$.
They are strict in $\abort$ and distribute over arbitrary joins.
With a pointwise lifting of the ordering on predicates, predicate transformers form a complete lattice $(\setpredtrans, \predtransleq)$ as well.

We consider sequential while-programs over $\setcom$ of the form
\begin{align*}
	\astmt\quad \defebnf\quad \acom\bnf\choiceof{\astmt_1}{\astmt_2}\bnf\seqof{\astmt_1}{\astmt_2}\bnf{\loopof{\astmt}} \enspace.
\end{align*}
Programs also have a semantics in terms of predicate transformers that is derived from the semantics of commands.
The non-deterministic choice is the join, $\sem{\choiceof{\astmt_1}{\astmt_2}}=\sem{\astmt_1}\predtransjoin\sem{\astmt_2}$, the composition is function composition, $\sem{\seqof{\astmt_1}{\astmt_2}}=\sem{\astmt_2}\circ\sem{\astmt_1}$, and the semantics of iteration is $\sem{\loopof{\astmt}}=\bigpredtransjoin_{i\in\nat}\sem{\astmt^i}$ with $\sem{\astmt^0}$ being the identity and $\sem{\astmt^{i+1}}=\sem{\seqof{\astmt}{\astmt^i}}$.
The set of predicate transformers is closed under these constructions, and so $\sem{\astmt}\in\setpredtrans$.

We specify program correctness with Hoare triples of the form $\hoareof{\apred}{\astmt}{\apredp}$.
The triple is valid, denoted by $\models \hoareof{\apred}{\astmt}{\apredp}$, if $\semof{\astmt}{\apred}\predleq\apredp$.
We reason about validity using the separation logic (SL) induced by $\sem{-}$, which is given in \cref{Figure:PL} (ignore the \makeColorLogic{blue} parts for now).
If there is a derivation for a Hoare triple, we write $\vdash\hoareof{\apred}{\astmt}{\apredp}$.
For soundness of the frame rule, it is well-known that the predicate transformers $\sem{\acom}$ need to satisfy an extra property called \emph{locality}: for all $\apred, \apredp\in\setpreds$ we need
\begin{align}
	\semof{\acom}{\apred\mstar\apredp}\predleq\semof{\acom}{\apred}\mstar\apredp
	\enspace .
	\tag{Locality}
	\label{Equation:Locality}
\end{align}
We say that a rule is sound if validity of the premise entails validity of the conclusion.

\begin{theorem}
	\label{Lemma:SoundnessHoare}
	\label{Lemma:SoundnessSL}
	Separation logic is sound: ${}\vdash\hoareof{\apred}{\astmt}{\apredp}$ entails ${}\models \hoareof{\apred}{\astmt}{\apredp}$.
\end{theorem}


\newcommand{\semcalcticol}{\vdash}
\begin{figure*}[t]
	\def\MathparLineskip{\lineskip=2mm}
	\small
	\begin{mathpar}
		\inferH{com}{
			\casemof{\acom}{\highlight{\acontext}}{\apred}\predleq \apredp
		}{
		  \choareHighOf{\highlight{\acontext}}{\apred}{\acom}{\apredp}
		}
		\and
		\inferH{consequence}{
			\apred\predleq \apred'
			\\
			\apredp'\predleq\apredp
			\\\\
			\choareHighOf{\highlight{\acontext}}{\apred'}{\astmt}{\apredp'}
		}{
			\choareHighOf{\highlight{\acontext}}{\apred}{\astmt}{\apredp}
		}
		\and
		\inferH{seq}{
			\choareHighOf{\highlight{\acontext}}{\apred}{\astmt_1}{\apredp}
			\\\\
			\choareHighOf{\highlight{\acontext}}{\apredp}{\astmt_2}{\apredppp}
			}{
			\choareHighOf{\highlight{\acontext}}{\apred}{\seqof{\astmt_1}{\astmt_2}}{\apredppp}
		}
		\and
		\inferH{choice}{
			\choareHighOf{\highlight{\acontext}}{\apred}{\astmt_1}{\apredp}
			\\\\
			\choareHighOf{\highlight{\acontext}}{\apred}{\astmt_2}{\apredp}
			}{
			\choareHighOf{\highlight{\acontext}}{\apred}{\choiceof{\astmt_1}{\astmt_2}}{\apredp}
		}
		\and
		\inferH{loop}{
			\choareHighOf{\highlight{\acontext}}{\apred}{\astmt}{\apred}
		}{
			\choareHighOf{\highlight{\acontext}}{\apred}{\loopof{\astmt}}{\apred}
		}	
		\and
		\inferH{frame}{
			\choareHighOf{\highlight{\acontext}}{\apred}{\astmt}{\apredp}
		}{
			\choareHighOf{\highlight{\acontext}}{\apred\mstar\apredppp}{\astmt}{\apredp\mstar\apredppp}
		}
		\and
		\makeColorLogic{
			\inferH{context}{
				\choareHighOf{\acontext\mstar \apredppp}{\apred}{\astmt}{\apredp}
			}{
				\choareHighOf{\acontext}{\apred\mstar \apredppp}{\astmt}{\apredp\mstar \apredppp}
			}
		}
		\and
		\makeColorLogic{
			\inferH{widen}{
				\choareHighOf{\acontext}{\apred\mstar \apredppp}{\astmt}{\apredp\mstar \apredppp}
			}{
				\choareHighOf{\acontext\mstar \apredppp}{\apred}{\astmt}{\apredp}
			}
		}
	\end{mathpar}%
	\caption{%
		Proof rules of context-aware separation logic (\theLogicSeq).\label{Figure:PL}%
	}
\end{figure*}

A state $\astate$ is a footprint of $\sem{\acom}$ if $\semof{\acom}{\set{\astate}}\neq \abort$, the transformer does not abort on the state.


\section{Context-Aware Reasoning for Smaller Footprints}
\label{Section:CAReasoning}

The frame rule is key to local reasoning: it allows one to focus all attention only on a smallest footprint $\apred$ of the computation $\astmt$ and current state at hand, obtaining for free that the remainder of the state, captured by the \emph{frame} $\apredppp$, is preserved by $\astmt$.
We are concerned with situations where the smallest footprint remains inherently large, thwarting any attempt at local reasoning.

What causes large footprints is the locality requirement for commands.
If we cannot guarantee $\semof{\acom}{\astate\mstar\astatep}\predleq\semof{\acom}{\astate}\mstar\set{\astatep}$ for all states $\astatep$, then we have to define $\semof{\acom}{\astate}=\abort$.
That is, $\acom$ aborts on $\astate$ and any attempt at reasoning locally about the effect of $\acom$ on $\astate$ will fail.
The locality requirement, in turn, is a consequence of the fact that the frame rule is meant to hold for all possible frames.
It says that, no matter which frame $\apredppp$ is added to the proof, the program has to leave it unchanged.
In short, since the frame rule is context-agnostic, we need locality, and due to locality programs that affect a large part of the state inherently have large footprints.

This work starts from the idea of introducing a context-aware variant of the frame rule that justifies smaller footprints when reasoning about programs whose effect on the frame is benign.
The rationale is that if the predicate $\apredppp$ to be added by framing is known, then we can relax the locality requirement
and hence enable more local reasoning.
We develop this idea in a conservative extension of separation logic.

\subsection{Context-Aware Separation Logic}
We propose \emph{context-aware separation logic (\theLogicSeq)} in which correctness statements $\choareof{\acontext}{\apred}{\astmt}{\apredp}$ are Hoare triples enriched by a \emph{context} $\acontext$.
The context is a predicate that is meant to be framed to the Hoare triple $\hoareof{\apred}{\astmt}{\apredp}$.
This intuition is captured by the rule \ruleref{context} and becomes more evident as we define the validity of such correctness statements.

\begin{definition}[Validity of \theLogicSeq statements]
	\label{def:casl-validity}
	\(
		{}\models \choareof{\acontext}{\apred}{\astmt}{\apredp}
		~~\defifff~
		{}\models \hoareof{\apred\mstar\acontext}{\astmt}{\apredp\mstar\acontext}
		\enspace .
	\)
\end{definition}

We reason about the validity of \theLogicSeq statements using the program logic from \cref{Figure:PL} (including the \makeColorLogic{blue} parts).
We write $\vdash\choareof{\acontext}{\apred}{\astmt}{\apredp}$ if a correctness statement can be derived using this logic.
The benefit of knowing the predicate that should be framed is that we can relax the locality requirement on the semantics of commands relative to that context.
To develop this relaxation, observe that pushing and pulling predicates $\acontext$ into and from the context as captured by the new rules \ruleref{context} and \ruleref{widen}, respectively, is sound immediately by our definition of validity above.
Instead, we have to focus on the modified rule \ruleref{com}: it uses a new context-aware semantics that takes the role of the standard semantics.

A context-aware semantics is a function that assigns to each command $\acom\in\setcom$ a context-aware predicate transformer $\casem{\acom}{\bullet}:\setpreds\rightarrow\setpredtrans$.
A context-aware predicate transformer expects a context $\acontext$ as input and returns a suitable predicate transformer $\casem{\acom}{\acontext}$.
Context-aware semantics extend naturally to programs.

The soundness of rule \ruleref{com} then relies on the requirement that the context-aware semantics over-approximates the standard semantics for the different choices of the context.

\begin{definition}
	Let $\acom \in \setcom$ and $\acontext \in \setpreds$. We say that $\casem{\acom}{\acontext}$ satisfies \emph{mediation} if
	\begin{align}
		\forall \apred\in\setpreds.
		\quad
		\semof{\acom}{\apred\mstar\acontext} ~\predleq~ \casemof{\acom}{\acontext}{\apred}\mstar\acontext
		\enspace.
		\tag{Mediation}
		\label{Equation:Mediation}
	\end{align}
\end{definition}

Although we need \eqref{Equation:Mediation} for the soundness of rule \ruleref{com}, it plays a similar role for \ruleref{context} as locality does for the \ruleref{frame} rule: it allows us to push a predicate $\acontext$ into the context and focus on the remainder $\apred$ if we can guarantee that $\acontext$ is invariant under the actions of the program.

Soundness of \theLogicSeq now follows because \ruleref{com} is sound by \eqref{Equation:Mediation}, \ruleref{context} and \ruleref{widen} are sound as they exploit our validity from \cref{def:casl-validity}, and the remaining rules are sound because separation logic is sound by \cref{Lemma:SoundnessHoare}.

\begin{theorem}[Soundness of \theLogicSeq]
	\label{thm:casl-soundness}
	Assume that $\sem{\acom}$ satisfies \eqref{Equation:Locality} and $\casem{\acom}{\acontext}$ satisfies \eqref{Equation:Mediation} for all $\acom\in\setcom$, $\acontext\in\setpreds$.
	Then ${}\vdash\choareof{\acontext}{\apred}{\astmt}{\apredp}$ implies ${}\models\choareof{\acontext}{\apred}{\astmt}{\apredp}$.
\end{theorem}

Actually, to prove \Cref{thm:casl-soundness}, we only need mediation for  $\casem{-}{\acontext'}$, if $\acontext'$ is a context that occurs in an applications of rule \ruleref{com} which is used to derive $\vdash\choareof{\acontext}{\apred}{\astmt}{\apredp}$.
We pose the stricter requirement that mediation has to hold for all contexts to avoid a side condition in the rule.
However, \eqref{Equation:Mediation} can be weakened so that it is only required to hold for the contexts that are of interest for a particular proof.

\paragraph{Conservative extensions}
It is worth pointing out that the above soundness result does not rely on any correspondence, besides \eqref{Equation:Mediation}, among the standard semantics $\sem{-}$ and the new context-aware semantics $\casem{-}{\bullet}$.
While we exploit this potential for approximation for practical purposes in \Cref{sec:contextualization}, we typically start from {\theLogicSeq}s that conservatively extend separation logics.
That is, we study {\theLogicSeq}s that are both sound and complete relative to the separation logic induced by a given standard semantics:
\begin{align*}
	\forall \astmt, \apred, \apredp, \acontext. \quad & & \vdash \choareOf{\acontext}{\apred}{\astmt}{\apredp} \implies & \vdash \hoareOf{\apred \mstar \acontext}{\astmt}{\apredp \mstar \acontext} &
	\tag{Relative Soundness}
	\\
	\forall \astmt, \apred, \apredp. \quad & & \vdash \hoareOf{\apred}{\astmt}{\apredp} \implies & \; \exists \acontext. \; \vdash \choareOf{\acontext}{\apred}{\astmt}{\apredp}\ . &
	\tag{Relative Completeness}
\end{align*}

One can always obtain such a conservative extension from a separation logic induced by any given standard semantics $\sem{-}$.
The canonical way to do so is to let $\casem{-}{\emp}$ and $\sem{-}$ coincide.

\begin{theorem}[Conservative extension]
	\label{lem:conservative-extension}
	If $\casem{\acom}{\acontext}$ satisfies \eqref{Equation:Mediation} for all $\acom$ and $\acontext$, and $\casem{-}{\emp} = \sem{-}$, then the \theLogicSeq induced by $\casem{-}{\bullet}$ conservatively extends the SL induced by $\sem{-}$.
\end{theorem}
In the remainder of the section, we develop machinery for deriving suitable context-aware semantics.


\subsection{Contextualization}
\label{sec:contextualization}

The purpose of rule \ruleref{context} is to frame out predicates that are invariant under the command of interest.
Our goal is to obtain \theLogicSeq derivations that look something like this:
\begin{align*}
	\inferrule*[left={\ruleref{consequence}}]{
		~~
		\inferrule*[left={\ruleref{context}}]{
			~~
			\inferrule*[left={\ruleref{com}}]{
				\casemof{\acom}{\apredpp}{\apred}\predleq\apredp
			}{
				~~\choareOf{\acontext}{\apred}{\acom}{\apredp}~~
			}
			~~
		}{
			\choareOf{\emp}{\apred\mstar\acontext}{\acom}{\apredp\mstar\acontext}
		}
		\\
		\apredppp\predleq\acontext
		~~
	}{
		\choareOf{\emp}{\apred\mstar\apredppp}{\acom}{\apredp\mstar\acontext}
	}
\end{align*}
But how does one determine predicates that are guaranteed to be invariant under the command?
We first tackle this problem for the original semantics $\sem{-}$ and from this derive a schema for obtaining context-aware semantics $\casem{-}{\bullet}$.
To be precise, this is the problem we address next:
\begin{quote}
	\underline{Contextualization} \\
	Given: $\semof{\acom}{\apred\mstar\apredppp}$. \\
	Determine: Predicates $\apredp$ and $\acontext$ with $\apredppp\predleq\acontext$ so that $\semof{\acom}{\apred\mstar\acontext}\predleq \apredp\mstar\acontext$.
\end{quote}
Of course the predicates $\apredp$ and $\acontext$ should be as precise as possible.
We solve this problem in a setting that is specific enough to provide helpful assumptions, yet general enough to cover frameworks like flows~\cite{DBLP:conf/esop/KrishnaSW20,DBLP:conf/tacas/MeyerWW23}
and ghost state induced by morphisms~\cite{DBLP:journals/pacmpl/Nanevski0DF19}.
In analogy to the term framing, we say we \emph{contextualize} $\acontext$.

\subsubsection{The Semantics of updates}
\label{sec:contextualization:semantics}

The motivation for contextualization stems from the fact that the states in $\apredppp$ can be large.
In our examples from \cref{sec:motivation}, these states are full registries and full subtrees.
It is worth having a closer look at what forces us to maintain these rich states.

\begin{example}
	\label{ex:decompose-flow}
	The crucial moment in the BST proof from \cref{sec:motivation:flow} is this Hoare triple:
	\begin{lstlisting}[language=SPL,numbers=none,style=codeInline,keywords={free}]
	  $\annot{
	    \anode \pointsto (l,r,k) \MSTAR p \pointsto (\anodep,\anodepp,i) \MSTAR \anodep \pointsto (\pnull, \anodeppp, j)\MSTAR \apredppp
	  }$
	  $\anode$.$\key$ = $\anodep$.$\key$; $p$.$\lchild$ = $\anodep$.$\rchild$;
	  $\annot{
	    \anode \pointsto (l,r,j) \MSTAR p \pointsto (\anodeppp,\anodepp,i) \MSTAR \apredppp'
	  }$
	\end{lstlisting}
	The update modifies a pointer of $p$ and the key of $\anode$. 
	This is the change on the states in the predicate $\apred$ introduced in the notion of contextualization.
	However, the update also has an effect on the subtree rooted at $r$ (without $p$ as the node belongs to $\apred$).
	We are interested in the contents of this subtree, the set of keys of all nodes reachable from the root.
	While the subtree does not change physically, the update changes the contents of $\apredppp$.
	In short, while the physical modification involves only few nodes, it influences the ghost state associated with a whole set of nodes.
	The phenomenon is independent of the formalism we use to describe states, be it recursive predicates, flow graphs, or morphisms.
	\qed
\end{example}

\begin{example}
	\label{ex:decompose-lin}
	Recall the linearizability proof goal from \cref{sec:motivation:lin}: \[
		\hoareof{\mcslstate(M(h)) \mstar (h, R)}{\acom}{\mcslstate(M(h)[k \mapsto v]) \mstar ((k,v) \cdot h, R')}
		\enspace .
	\]
	The linearization point $\acom$ of \code{upsert($k$, $v$)} modifies the physical representation $\mcslstate$ of the structure and appends the new key-value pair $(k,v)$ to the history $h$.
	Moreover, $\acom$ also affects the registry $R$: it linearizes all threads that are awaiting $(k,v)$ to be upserted, as dictated by their prophecy variables, resulting in the (potentially entirely) new registry $R'$.
	Here, we are interested in contextualizing the registry $R$ as part of $\apredppp$ and keeping both the physical representation as well as the history in $\apred$.
	The reason for this is that the registry update is induced by the change of the history.
	We wish to focus the proof on the part that matters, the history and its update.
	\qed
\end{example}

To capture the fact that an update involves a modification of the physical state and a modification of the ghost state, we wish to assume that the semantics of commands can be decomposed into the physical update and a separate operation that adjusts the ghost state according purely to the new physical state. 
However, distinguishing between physical and ghost state is unnecessarily strict and impractical in some cases, as seen in \cref{ex:decompose-lin}. 
Instead, we only assume that the semantics of commands decomposes according to the following equalities:
\[
	\semof{\acom}{\apred\mstar\apredppp} \ =\ \upof{\acom}{\apred\mstar\apredppp}\quad\text{and}\qquad \upof{\acom}{\apred\mstar\apredppp} \ =\ \upof{\acom}{\apred}\imult\apredppp\quad \text{if }\upof{\acom}{\apred}\neq\abort
	\enspace .
\]
Here, $\up{\acom}$ is a predicate transformer that implements the \emph{core update}.
The core update satisfies a condition similar to \eqref{Equation:Locality}, except that the ordinary multiplication $\mstar$ is replaced by a ghost multiplication $\imult$ applying the \emph{induced update} on the remaining state.
Going forward, one can think of the core and induced updates as updates to the physical and ghost state, respectively, but our results do not rely on this understanding.
The ghost multiplication is commutative and associative.
(There is no need to assume the existence of units.)
We lift the ghost multiplication to predicates in the expected way: $\apred\imult\apredp=\setcond{\astate\imult\astatep}{\astate\in\apred\wedge \astatep\in\apredp\wedge \astate\imultdef\astatep}$ and $\abort\imult\apred=\abort=\apred\imult\abort$.
We make the assumption that the result of a ghost multiplication decomposes uniquely as follows.
For $\apred_1, \apred_2\subseteq\setstates$ with $\apred_1\imult\apred_2=\apredp$, there are unique smallest predicates $\apredp_1, \apredp_2$ with $\apredp_1\mstar\apredp_2=\apredp$ so that $\apredp_1$ corresponds to $\apred_1$ and $\apredp_2$ corresponds to $\apred_2$.
This correspondence is formalized as an equivalence on states, which we have suppressed as we do not need it beyond this unique decomposition requirement.

\begin{example}
	\label{ex:imult-lin}
	For our registry example from \cref{sec:motivation:lin}, we define the core update $\up{\acom}$ for the linearization point $\acom$ to extend the history: $\upof{\acom}{(h,\emptyset) \mstar \apredppp} = ((k,v) \cdot h, \emptyset) \imult \apredppp$.
	The induced update $\imult$ takes care of linearizing threads according to new entries of the history.
	Formally, \[
		(h, R_1) \imult\, (h, R_2) \,=\, (h, R_1 \uplus R_2)
		\qquad\text{and}\qquad
		((k,v) \cdot h, R_1) \imult\, (h, R_2) \,=\, ((k,v) \cdot h, R_1 \uplus R_2')\ .
	\]
	Here, $R_2'$ is obtained from $R_2$ by changing all entries $R_2(\tid)=\anobl{k,v}$ to $\aful{k,v}$ and leaving all other entries unchanged.
	In all remaining cases, $\imult$ is undefined.
	With this, $\imult$ indeed captures our intuition of an induced update that adjusts the registry given the effect that the core update has on the history.

	For an induced update $(h_1, R_1) \imult\, (h_2, R_2) = (h_3, R_3)$ that is defined, its unique decomposition splits the resulting registry $R_3$ along the domains of $R_1$ and $R_2$ which are disjoint by definition of $\imult$.
	The decomposition is $(h_3, \project{R_3}{\domof{R_1}}) \mstar (h_3, \project{R_3}{\domof{R_2}})$ where $\project{R_3}{\domof{R_1}}$ is the projection of $R_3$ to the domain of $R_1$, and similarly for $R_2$.

	So far, we have ignored the physical representation $\mcslstate$ because it is orthogonal to the ghost state.
	The separation algebra for the overall proof will be a product of two independent separation algebras, one capturing the physical state and one the ghost state.
	The induced update $\imult$ on the ghost state separation algebra extends naturally to the product separation algebra: the core update keeps the entire physical state in $\upof{\acom}{\apred}$, the physical part of $\apredppp$ is $\emp$, and the ghost multiplication is the separating conjunction.
	\qed
\end{example}

The use of a ghost multiplication is inspired by the morphism framework in \cite{DBLP:journals/pacmpl/FarkaN0DF21} where the separation algebra of states $\Sigma$ is mapped to a separation algebra of ghost states  $\Gamma$ that has its own multiplication.
We stay within one separation algebra, which can be thought of as $\Sigma\times\Gamma$, and assume to inherit the second multiplication.

\subsubsection{Solution}

We approach contextualization by abstract interpretation: we give an approximate semantics for the commands from which we can construct the desired predicates.
A particularity of our approach is that we do not want to devise an abstract domain but wish to stay in the realm of separation logic where the algebraic framework is well-developed.
Another particularity is that the semantics of commands consist of a core and an induced update, both of which we have to approximate.

We mimic the core update by an \emph{approximate core update} $\absup{\acom}$.
Like the core update, it should be a predicate transformer that satisfies $\absupof{\acom}{\apred\mstar\apredppp}=\absupof{\acom}{\apred}\imult\apredppp$ if $\absupof{\acom}{\apred}\neq\abort$.
We also expect soundness, $\upof{\acom}{\apred}\predleq\absupof{\acom}{\apred}$.

To mimic the induced update, observe that the ghost multiplication induces a family of predicate transformers $\ghostconc{\apred}$ that capture the effect of the ghost multiplication on the first operand when the second operand is fixed to be~$\apred$.
For $\apred, \apredppp\subseteq\setstates$, we define $\ghostconcof{\apred}{\apredppp}=\apredppp'$, if $\apredppp\imult\apred=\apredppp'\mstar\apred'$ is the unique decomposition.
This can be understood as currying, then a partial instantiation, and finally a masking of the result.
For $\apred=\abort$ or $\apredppp=\abort$, we define $\ghostconcof{\apred}{\apredppp}=\abort$.
It is worth noting that these functions capture the ghost multiplication without loss of information: $\apredppp\imult\apred=\ghostconcof{\apred}{\apredppp}\mstar \ghostconcof{\apredppp}{\apred}$.
To define the predicates we are after, it will be beneficial to approximate this family rather than the multiplication operator.

\begin{example}
	Consider the ghost states $\apred=((k,v) \cdot h, R_1)$ and $\apredppp=(h, R_2)$.
	What is $\apredppp'=\ghostconc{\apred}(\apredppp)$?
	To find it, first compute $((k,v) \cdot h, R_1) \imult\, (h, R_2)$ along the lines of \cref{ex:imult-lin}.
	If the ghost multiplication is undefined, we have $\apredppp'=\bot$.
	Otherwise, it yields $((k,v) \cdot h, R_1 \uplus R_2')$ with $R_2'$ being the appropriately updated registry as before.
	The $\apredppp$-portion of the unique decomposition for this ghost state gives $\apredppp'=((k,v) \cdot h, R_2')$.
	As expected, $\ghostconc{\apred}$ updates $\apredppp$ by extending its history by the new event $(k,v)$ and updating the registry to $R_2'$ accordingly.

	Similarly, $\ghostconc{\apredppp}(\apred) = ((k,v) \cdot h, R_1)$ if the ghost multiplication from above is defined (recall that $\imult$ is commutative).
	We confirm $\apred\imult\apredppp = ((k,v) \cdot h, R_1) \mstar ((k,v) \cdot h, R_2') = \ghostconc{\apredppp}(\apred) \mstar \ghostconc{\apred}(\apredppp)$.
	\qed
\end{example}

An \emph{approximate ghost multiplication} is a family of predicate transformers $\ghostabs{\apred}$.
We now proceed the other way around and use the family to define $\apred\absimult\apredppp=\ghostabsof{\apredppp}{\apred}\mstar\ghostabsof{\apred}{\apredppp}$.
Again, we expect soundness, $\apred\imult\apredppp\predleq\apred\absimult\apredppp$.

With the approximate core and induced updates in place, we can now state our solution to the contextualization problem.
Recall that we are given $\semof{\acom}{\apred\mstar\apredppp}$, and we want to determine predicates $\apredp$ and $\acontext$ with $\apredppp\predleq\acontext$ so that $\semof{\acom}{\apred\mstar\acontext}\predleq \apredp\mstar\acontext$.
We define:
\begin{alignat*}{5}
	\apredpp\ &=\ \rho^*(\apredppp)\qquad&\text{with}\qquad\rho\ &=\ \ghostabs{\apred'}\qquad \text{and}\qquad \apred'=\absupof{\acom}{\apred}
	\\
	\apredp\ &=\ \sigma(\apred')\qquad&\text{with}\qquad\sigma\ &= \ghostabs{\apredpp}.
	\end{alignat*}
By assumption, $\ghostabs{\apred'}$ is a predicate transformer (i.e., strict and a complete join morphism), and so the reflexive transitive closure $\rho^* = \lfpof{(\lambda f.\, \mathit{id} \predtransjoin \rho \circ f)}$ is well-defined.
The construction  captures our intuition about the context being stable under the (ghost) updates inflicted by the command, and it solves contextualization as promised.

\begin{theorem}
	\label{thm:contextualization}
	Consider $\semof{\acom}{\apred\mstar\apredppp}$.
	Then $\semof{\acom}{\apred\mstar\apredpp}\predleq\apredp\mstar\apredpp$ and $\apredppp\predleq\apredpp$.
\end{theorem}

It is worth noting that we only lose precision in the approximations and in the transitive closure.
The transitive closure seems to be unavoidable to make $\apredpp$ invariant under the command.
The physical update is often deterministic and does not need approximation.
Hence, the only parameter worth tuning is the precision of the approximate ghost multiplication.
We illustrate the construction of $\apredp$ and $\apredpp$ in \cref{thm:contextualization} on the registry example  from \cref{sec:motivation:lin}. 
It is worth noting that, in this example, the transitive closure does not lose information because the ghost multiplication is idempotent.

\begin{example}
	Consider $\semof{\acom}{\apred\mstar\apredppp}$ with $\apred=(h,\emptyset)$, $\apredppp=(h,R)$, and $\acom$ the linearization point of an \code{upsert($k$, $v$)}.
	For simplicity, we choose not to perform any approximation here, i.e., choose $\absup{\dontcare}=\up{\dontcare}$ and $\ghostabs{\dontcare}=\ghostconc{\dontcare}$.
	However, we note that the use of approximations can enhance proof automation by improving the convergence of solving the contextualization problem.
	For an example use of approximations, refer to \cref{sec:bst-example}.

	We now compute $\apredp$ and $\apredpp$ to solve contextualization.
	First, we have $\apred'=\absupof{\acom}{\apred}=((k,v) \cdot h, \emptyset)$.
	Then, $\ghostabs{((k,v) \cdot h, \emptyset)}((h,R))=((k,v) \cdot h, R')$ with $R'$ the updated variant of $R$.
	Because $\ghostabs{\dontcare}$ is idempotent, we obtain $\apredpp=((k,v) \cdot h, R')$.
	Finally, $\apredp=\ghostabs{(h, R)}(((k,v) \cdot h, \emptyset))=((k,v) \cdot h, \emptyset)$.
	That is, $((k,v) \cdot h, \emptyset) \mstar ((k,v) \cdot h, R')$ approximates the post image of $\acom$ under $\apred\mstar\apredppp$.
	\qed
\end{example}

\subsubsection{An induced context-aware semantics}

The above solution to the contextualization problem also gives rise to a context-aware semantics based on the over-approximation principle.
The derived context-aware semantics computes the physical update $\apred'$ and applies to it the approximate ghost multiplication for the given context $\acontext$, if $\acontext$ is invariant under the update, that is, if $\acontext$ is the fixed point solution to $\rho^*(\apredppp)$.
We define the \emph{induced context-aware predicate transformer} $\icasem{\acom}{\acontext}$ for a non-empty context $\acontext\neq\emp$ by \[
	\icasemof{\acom}{\acontext}{\apred} = \begin{cases}
		\ghostabsof{\acontext}{\apred'}
		&\text{if~\:}
		\acontext = \ghostabsof{\apred'}{\acontext}
		\text{\:~where~\:}
		\apred'=\absupof{\acom}{\apred}
		\\
		\top
		&\text{otherwise}
		\enspace .
	\end{cases}
\]
For an empty context, there is no need for approximation, we simply use the original semantics, $\icasem{\acom}{\emp}=\sem{\acom}$.
Using \cref{thm:contextualization} it is easy to see that $\icasem{\acom}{\bullet}$ satisfies \eqref{Equation:Mediation}.
That is, we can instantiate \theLogicSeq with $\icasem{\acom}{\bullet}$ and obtain by \Cref{lem:conservative-extension} a conservative extension of separation logic that supports contextualization for reasoning more locally about large footprints.

\begin{theorem}
	\label{thm:approx-induced-casl-is-conservative}
	The \theLogicSeq induced by $\icasem{\acom}{\bullet}$ conservatively extends SL.
\end{theorem}



\subsection{A Concurrent Extension}
\label{sec:og}

To reason about concurrent programs in separation logic we employ the Owicki-Gries proof principle \cite{DBLP:journals/acta/OwickiG76}.
That is, we reason in two steps.
First, we verify the program code as if it was run by a single thread in isolation.
Second, we check interference freedom to ensure that the proof remains valid in the presence of other threads.
If so, the concurrent Hoare triple $\hoareOf{\apred}{\astmt}{\apredp}$ is valid, denoted by ${}\mmodels\hoareof{\apred}{\astmt}{\apredp}$, meaning that any number of threads each executing $\astmt$ and starting in $\apred$ will reach $\apredp$.

The judgments for verifying the isolated thread take the form $\thePredicates,\theInterference\semcalc\hoareOf{\apred}{\astmt}{\apredp}$.
The proof rules for these judgments (\Cref{app:og-casl}) collect the predicates that were used during the proof in the set $\thePredicates$ and the interferences in the set~$\theInterference$~\cite[Section 7.3]{DBLP:conf/popl/Dinsdale-YoungBGPY13}.
The interferences can be thought of as pairs $(\acom, \apred)$ for which rule \ruleref{com} was applied.
Recording these pairs allows to later \emph{replay} the effect of the command on other threads.

The interference freedom check ensures that, given a set of interferences $\theInterference$ and a set of predicates $\thePredicates$, no interference $(\acom, \apred)$ from $\theInterference$ can invalidate a predicate $\apredp$ from $\thePredicates$.
Intuitively, this means that replaying $\acom$ under $\apred \cap \apredp$ results in a state covered by $\apredp$.
To support per-thread local state, one has to assume that the underlying separation algebra is a product of two separation algebras defining the global and local state.
Then, the effect of the interfering command is its update to the global state, leaving the local state unchanged.
More concretely, if $\apred=(\ashared_\apred,\alocal_\apred)$ and $\apredp=(\ashared_\apredp,\alocal_\apredp)$ then we compute $\semof{\acom}{\ashared_\apred\cap\ashared_\apredp,\alocal_\apred}=(\ashared',\alocal')$ and check if $(\ashared',\alocal_\apredp)\predleq\apredp$.
If this is the case, we write $\isInterferenceFreeOf[\theInterference]{\thePredicates}$ and say that $\thePredicates$ is interference-free wrt. $\theInterference$.

The resulting Owicki-Gries proof system is sound \cite{DBLP:journals/pacmpl/MeyerWW22}.

\begin{theorem}
	\label{thm:soundness-OG}
	$\thePredicates, \theInterference\semcalc\hoareof{\apred}{\astmt}{\apredp}$
	and
	${}\isInterferenceFreeOf[\theInterference]{\thePredicates}$
	and
	$\apred\in\thePredicates$
	imply
	${}\mmodels\hoareof{\apred}{\astmt}{\apredp}$.
\end{theorem}

We develop \emph{context-aware concurrent separation logic (\theLogicOG)} whose
judgements take the form $\thePredicates, \theInterference\semcalc\choareof{\acontext}{\apred}{\astmt}{\apredp}$.
As for \theLogicSeq, $\acontext$ is meant to be framed to the pre- and postcondition.
That is, validity $\mmodels\choareof{\acontext}{\apred}{\astmt}{\apredp}$ holds iff $\mmodels\hoareof{\apred\mstar\acontext}{\astmt}{\apredp\mstar\acontext}$.
The extended program logic is as expected, we elide it here for brevity.
Refer to \Cref{app:og-casl} for more details.
This extension is sound and it is easy to obtain a conservative extension of the standard Owicki-Gries approach.

\begin{theorem}
	\label{thm:soundness-OGCASL}
	$\thePredicates, \theInterference\semcalc\choareof{\acontext}{\apred}{\astmt}{\apredp}$
	and
	${}\isInterferenceFreeOf[\theInterference]{\thePredicates}$
	and
	$\apred\in\thePredicates$
	imply
	${}\mmodels\choareof{\acontext}{\apred}{\astmt}{\apredp}$.
\end{theorem}

Since contextualization addresses atomic commands, it is equally applicable to both the sequential $\theLogicSeq$ and the concurrent $\theLogicOG$. To avoid notational clutter, we stay within $\theLogicSeq$ throughout the remainder of the paper. However, we stress that we have evaluated our approach against concurrent benchmarks, see \cref{sec:instantiation:automation}.


\newcommand{\auxstatemult}{\uplus}
\newcommand{\auxstatemultdef}{\statemultdef_\auxstatemult}
\newcommand{\myfg}{\astate}
\newcommand{\afg}{\astate}
\newcommand{\afgp}{\astatep}
\newcommand{\afgpp}{\astatepp}
\newcommand{\afgppp}{\astateppp}
\newcommand{\afgc}{\astatepp}
\newcommand{\emptyfg}{\astate_\emptyset}
\newcommand{\newinflow}{\inflow_\mathit{new}}
\newcommand{\fprelup}[1]{\fprel^{#1}}

\section{Contextualization for Flow Graphs}
\label{sec:instantiation}

We instantiate the contextualization principle from the previous section for the flow framework~\cite{DBLP:journals/pacmpl/KrishnaSW18,DBLP:conf/esop/KrishnaSW20,DBLP:conf/tacas/MeyerWW23}.
As alluded to in \cref{sec:motivation}, this combination of techniques allows us to handle complicated updates in a local way with relative ease even if the updates' footprints are unbounded. 
The remainder of this section formalizes the intuition from \cref{sec:motivation} about flow graphs (\cref{sec:instantiation:flow-graphs}), studies how updates interact with flow graph composition (\cref{sec:instantiation:updates-approach} and \cref{sec:instantiation:updates}), instantiates the contextualization principle (\cref{sec:instantiation:approximates}), and addresses proof automation (\cref{sec:instantiation:automation}).

\subsection{Flow Graphs}
\label{sec:instantiation:flow-graphs}

We introduce the separation algebra of flow graphs, following \cite{DBLP:conf/tacas/MeyerWW23}. 

\smartparagraph{Flow Monoids}
Flow graphs are parametric in the ghost state they carry.
These so-called flow values are drawn from a \emph{flow monoid}, a commutative monoid $(\amonoid, \monadd, \monunit)$.
The monoid carries the natural order $\amonvalp \leq \amonval$ defined by $\amonval = \amonvalp+\amonvalpp$ for some $\amonvalpp\in\amonoid$.
We require that $(\amonoid,\leq)$ is an $\omega$-cpo, a partial order in which every ascending chain $\achain=\amonval_0\leq\amonval_1\leq\ldots$ has a join $\bigjoin\achain$.
A function $\atfun:\amonoid\to\amonoid$ is \emph{continuous}~\cite{Scott70} if it commutes with joins over ascending chains, $\atfun(\bigjoin\achain)=\bigjoin\atfun(\achain)$.
We write $\contfunof{\amonoid\to\amonoid}$ for the set of all continuous functions.
We expect the monoid operation to be continuous, $\amonvalp \monadd \bigjoin\achain = \bigjoin(\amonvalp \monadd \achain)$.

\smartparagraph{Flow Graphs}
\emph{Flow graphs} $\myfg=(\setnodes, \edges, \inflow)$ consist of a set of nodes $\setnodes\subseteq\nat$,  a set of edges  that are labeled by continuous \emph{edge functions} $\edges:\setnodes\times\nat\to\contfunof{\amonoid\prall{\to}\amonoid}$, and an \emph{inflow} $\inflow: (\nat\setminus\setnodes)\times\setnodes\to\amonoid$. 
The inflow can be thought of as the flow values that $\myfg$ receives from nodes outside the flow graph, from a frame or a context.
We use $\setflowconstraints$ for the set of all flow graphs and define the empty graph $\emptyfg= (\emptyset, \emptyset, \emptyset)$.
We may refer to the nodes, edges, and inflow by $\myfg.\setnodes$, $\myfg.\edges$, $\myfg.\inflow$, respectively.

To understand the ghost state that flow graphs $\myfg=(\setnodes, \edges, \inflow)$ encode, we use the derived quantities \emph{flow} and \emph{outflow}.
The flow dictates how flow values propagate within $\myfg$.
It is the least function $\fval:\setnodes\to\amonoid$ satisfying the flow equation: for all nodes $\anode\in\setnodes$, we have \[
	\fvalof{\anode} ~~=~~ \sum_{\anodep\in\nat\setminus\setnodes} \inflow(\anodep,\anode) ~~\monadd~~ \sum_{\anodep\in\setnodes} \edgesatof{\anodep}{\anode}{\fval(\anodep)}\ . 
	\tag{FlowEq}
	\label{def:flow-equation}
\]
The outflow $\outflow:\setnodes\times(\nat\setminus\setnodes)\to\amonoid$ is then obtained from the flow, $\outflowof{\anode, \anodep}=\edges_{(\anode, \anodep)}(\fvalof{\anode})$. 
It is worth pointing out that the flow can be computed using standard Kleene iteration.

\begin{example}
	\label{ex:inset-flow}
	We revisit the BST example from \cref{sec:motivation}. Let $\keyspace$ be the totally ordered set of keys with minimal and maximal elements $-\infty$ and $\infty$, respectively.
	Recall that the inset of a node $\anode$ in a tree is the set of keys for which the BST search will traverse $\anode$.
	To define insets in terms of a flow, we choose the flow monoid $\bigl( \powerset{\keyspace}\uplus\set{\bot,\top},\, \oplus,\, \bot \bigr)$ with $\amonval\oplus\bot=\amonval$ and $\amonval\oplus\amonvalp=\top$ in all other cases.
	The flow values propagated by this flow are sets of keys $\amonval,\amonvalp\subseteq\keyspace$ (to represent the insets), or dedicated sentinel values $\bot,\top$.
	We will use the sentinel values to capture some rudimentary shape information in the data structure invariant. Value $\bot$ denotes that a node is unreachable from $\Root$.
	Note that $\bot$ and $\emptyset$ differ: $\emptyset$ means that the node is still reachable from $\Root$, but \code{find} will not traverse it.
	Value $\top$ denotes that a node has multiple parents, i.e., the heap graph is not a tree.
	To establish this intuition, $\bot$ is neutral with respect to $\oplus$ and in all other cases $\oplus$ yields $\top$.

	The edge functions encode the BST search principle.
	They are derived from the physical representation of nodes as follows (where we use logical variables like $\leftof{\anode}$ to refer to the value of the corresponding field):
	\begin{align*}
		\edgesatof{\anode}{\anodep}{\amonval} ~=~ \begin{cases}
			\amonval\cap[-\infty,\keyof{\anode})  &\text{if }~ \anodep=\leftof{\anode} \\
			\amonval\cap(\keyof{\anode},\infty]  &\text{if }~ \anodep=\rightof{\anode} \\
			\bot  &\text{otherwise} \enspace.
		\end{cases}
	\end{align*}
	Here, we assume $\leftof{\anode} \neq \rightof{\anode}$, $\bot\cap\amonval=\bot$, and $\top\cap\amonval=\top$.
	The first case handles edges from a node $\anode$ to its left child $\anodep$.
	The edge forwards the portion of the given flow value $\amonval$ that is smaller than $\anode$'s key.
	Similarly, the second case forwards the portion of $\amonval$ that is larger than $\anode$'s key to its right child.
	In all other cases, the edge function produces $\bot$.

	Consider a binary search tree with nodes $\setnodes$ and root node $\Root \in \setnodes$.
	Define the flow graph $\myfg=(\setnodes, \edges, \inflow)$ where $\inflow$ is some inflow that satisfies $(-\infty,\infty] = \sum_{\anodep \notin \setnodes} \inflowof{\anodep,\Root}$ and $\bot = \sum_{\anodep \notin \setnodes} \inflowof{\anodep,\anode}$ for all $\anode \in \setnodes \setminus \set{\Root}$.
	Intuitively, the inflow $\inflow$ encodes that all searches start at $\Root$.
	Then $\myfg.\fvalof{\anode}$ is the inset of a node $\anode \in \setnodes$.
	See also \cref{fig:bst-remove-flows} for a concrete example.

	When $\myfg$ is understood, we write $\isof{\anode}$ for $\myfg.\fvalof{\anode}$.
	We refer to the \emph{left outset} of a node $\anode$ as the quantity produced by the edge function $\edgesat{\anode}{\leftof{\anode}}$ for the inflow of $\anode$.
	Formally, this is $\oslof{\anode} = \edgesatof{\anode}{\leftof{\anode}}{\isof{\anode}}$ if $\leftof{\anode}\neq\nullptr$ and $\oslof{\anode}=\emptyset$ otherwise.
	The \emph{right outset} $\osrof{\anode}$ is defined correspondingly.
	Subtracting $\anode$'s outsets from its inset yields the keys for which \code{find} terminates in $\anode$.
	That is, these are the keys that could be in $\anode$ while still satisfying the BST order property for the remaining graph.
	This quantity is the keyset of $\anode$: \[
		\ksof{\anode} ~=~ \begin{cases}
			\emptyset  &\text{if }~ \isof{\anode}\in\set{\bot,\top} \\
			\isof{\anode} \setminus \bigl( \oslof{\anode} \cup \osrof{\anode} \bigr)  &\text{otherwise}
			\ .
		\end{cases}
	\]
	The definition of the edge functions and the global inflow $\inflow$ guarantees that for $\myfg$ as defined above, the keysets of all nodes are pairwise disjoint.
	Hence, one can draw a conclusion locally about the entire state of the tree.
	\qed
\end{example}

\smartparagraph{Multiplication}
The ghost multiplication $\afg\imult\afgp$ requires disjointness of the nodes, $\afg.\setnodes\cap\afgp.\setnodes=\emptyset$. 
In this case, it removes the inflow to $\afg$ that is provided by $\afgp$, and vice versa: \[
	\afg\imult\afgp
	~~\defeq~~
	\bigl(
		\afg.\setnodes\uplus\afgp.\setnodes,\;
		\afg.\edges\uplus\afgp.\edges,\;
		\restrictto{\afg.\inflow}{(\nat\setminus\afgp.\setnodes)\times\afg.\setnodes}
		\uplus
		\restrictto{\afgp.\inflow}{(\nat\setminus\afg.\setnodes)\times\afgp.\setnodes}
	\bigr)
	\ .
\]

The ordinary multiplication $\afg\statemult\afgp$ extends the requirements of the ghost multiplication.
It is defined if
\begin{inparaenum}
	\item $\afg\imult\afgp$ is defined,
	\item the inflow expectation of one graph matches the outflow of the other, $\afg.\outflowof{\anode,\anodep}=\afgp.\inflowof{\anode,\anodep}$ and $\afgp.\outflowof{\anodep,\anode}=\afg.\inflowof{\anodep,\anode}$ for all nodes $\anode\in\afg.\setnodes,\: \anodep\in\afgp.\setnodes$, and
	\item the inflow/outflow interface between the two graphs is faithful, $\afg.\fval\uplus\afgp.\fval=(\afg\imult\afgp).\fval$.
\end{inparaenum}
If defined, the multiplication is $\afg\statemult\afgp=\afg\imult\afgp$. 

Flow graphs satisfy the unique decomposition requirement: one recomputes the flow in $\astate_1\imult\astate_2$, separates the graphs, and assigns as missing inflow the outflow of the other component. 
The unique decomposition is thus the moment in which the ghost multiplication requires computational effort (a recomputation of the flow).
Phrased differently, the ghost multiplication of flow graphs has a symbolic character in which the inflow from nodes in the same graph is hidden and only made explicit when the graph is decomposed.

\begin{lemma}[Unique Decomposition]
	\label{thm:unique-ghost-decomposition}
	Let $\astate_1\imult\astate_2=\astatep$. Then there are $\astatep_1$ and $\astatep_2$ with $\astatep=\astatep_1\mstar\astatep_2$ and $\astate_1.\setnodes=\astatep_1.\setnodes$ and $\astate_2.\setnodes=\astatep_2.\setnodes$.
	Moreover, the flow graphs $\astatep_1$ and $\astatep_2$ are unique. 
\end{lemma}

Combined with the results from~\cite{DBLP:conf/tacas/MeyerWW23}, we obtain the following.

\begin{lemma}\label{Lemma:FlowAlgebra}
	Flow graphs $(\setflowconstraints, \statemult, \imult, \set{\emptyfg})$ form a separation algebra.
\end{lemma}

If the flow graphs already compose as they are, there is nothing to do for the ghost multiplication.
The following lemma also holds for other separation algebras (see e.g. \cref{sec:ex-helping}), but we did not see a need to make it a requirement of our contextualization principle.

\begin{lemma}
	\label{Lemma:MultCoincides}
	If  $\afg\statemult\afgp$ is defined, so is $\afg\imult\afgp$ and we have $\afg\statemult\afgp=\afg\imult\afgp$.
\end{lemma}

\smartparagraph{Physical Updates}
Physical updates of flow graphs may only change the graph structure, but cannot change the nodes and the inflow:
if $\afgp\in \upof{\acom}{\afg}$, then we can rely on $\afg.\setnodes=\afgp.\setnodes$ and $\afg.\inflow=\afgp.\inflow$.

\subsection{Problem and Approach}
\label{sec:instantiation:updates-approach}

The ghost multiplication of flow graphs is difficult to approximate as it involves a full fixed point computation.
It is well-understood how to approximate fixed points with abstract interpretation.
In our setting, the challenge is that we do not want to retain precise shape information about the heap graph.
This is akin to applying abstract interpretation to statically analyze a program without having precise information about the program's source code.

We therefore develop a \emph{shape-independent} fixed-point approximation.
We start from the observation that commands change $\afg\mstar\afgc$ to $\afgp\imult\afgc$ with $\afgp\in\upof{\acom}{\afg}$.
This suggests we should understand which relations  $\afg\ctxfprel\afgp$ are stable under adding contexts $\afgpp$ in that they entail $\afg\mstar\afgpp\ctxfprel\afgp\imult\afgpp$.
The plan is then to strengthen the approximate physical update by such stable relations.
That is, we define $\absupof{\acom}{\afg}=\upof{\acom}{\afg}$ if $\afg\ctxfprel\afgp$ for all $\afgp\in \upof{\acom}{\afg}$.
Otherwise, we let $\absupof{\acom}{\afg}$ abort to obtain a sound approximation.
The approximate physical update then allows us to transfer knowledge about the fixed point $\afg\mstar\afgpp$ to $\afgp\imult\afgpp$.
This paves the way to a precise approximate ghost multiplication without the need for shape information.

\begin{example}
	\label{ex:estimator-motivation}
	To build intuition for the stable relations $\ctxfprel$ and how they relate to contextualization, we return to the example from \cref{sec:motivation} (and \cref{ex:inset-flow}).
	Let $\ctnof{\anode} = (\ite{\anode = \Root}{\emptyset}{\set{\keyof{\anode}}})$ denote the contents of a node $\anode$. 
	Recall that the core idea for the proof of the BST is to express its structural invariant in terms of a node-local predicate that may refer to the node's flow (or rather the derived quantities $\inset$ and $\keyset$).
	We denote this predicate by $\ninv{\anode,\setnodes}$ where $\setnodes$ is the set of all nodes of the tree:
	\begin{align*}
		\ninv{\anode,\setnodes} ~=~~&
			\set{\leftof{\anode},\rightof{\anode}}\subseteq\setnodes\uplus\set{\nullptr}
			~~\land~~
			\bigl(\isof{\anode}\neq\bot \implies \ctnof{\anode}\subseteq\ksof{\anode})
			~~\land~~
			\\&
			\isof{\anode} \neq \top
			~~\land~~
			\bigl(\anode = \Root \implies (-\infty,\infty] \subseteq \isof{\anode} \land \keyof{\anode}=\infty\bigr)
			\enspace .
	\end{align*}
	The node-$\anode$-local invariant requires that
	\begin{inparaenum}
		\item the overall structure $\setnodes$ is self-contained, i.e., $\anode$ cannot reach nodes outside of $\setnodes$,
		\item if $\anode$ is reachable, then its contents are contained in its keyset, i.e., the keyset invariant,
		\item $\anode$ has at most one parent node that is reachable from $\Root$, i.e., at most one node sends flow to $\anode$, and
		\item $\Root$ is a sentinel node with key $\infty$ and it is the entry point for all searches.
	\end{inparaenum}
	Note that the monoid operation $\oplus$ guarantees that the reachable nodes form a tree.

	The context $\acontext$ for an update $\acom$ on the BST will consist of a set of nodes $ \anode \in \setnodesp \subseteq \setnodes$ that satisfy $\ninv{\anode, \setnodes}$.
	To ensure that $\acontext$ is preserved, we must therefore understand 
	which relations $\ctxfprel$ preserve $\ninv{\anode, \setnodes}$.
	A candidate is the relation stating that the update may increase $\isof{\anode}$ to a larger set if $\isof{\anode} \notin \{\bot,\top\}$ holds before the update, and otherwise leaves $\isof{\anode}$ unchanged.
	This $\ctxfprel$ preserves $\ninv{\anode,\setnodes}$ because $\isof{\anode}$ only occurs on the right side of subset inclusions.
	This is exactly the kind of change that occurs in the nodes that are in the left subtree of $\anode$ in \cref{fig:bst-remove-flows}.
	Similarly, if $\anode$ is not the $\Root$, then $\isof{\anode}$ can be reduced, as long as it still contains $\keyof{\anode}$.
	On the other hand, if $\anode$ was unreachable before the update ($\isof{\anode} = \bot$) but were to become reachable and receive a proper inset, then this update would not be allowed because it may violate, e.g., $\ctnof{\anode} \subseteq \ksof{\anode}$.

	In summary, our main task is to identify relations $\ctxfprel$ that approximate $\up{\acom}$, preserve $\ninv{\anode,\setnodes}$, and are stable under the ghost multiplication.
	\qed
\end{example}

\subsection{Shape-Independent Fixed-Point Approximation}
\label{sec:instantiation:updates}

We introduce estimator relations to help us identify stable relations $\ctxfprel$.

\begin{definition}[Estimator on a Flow Monoid]
	An \emph{estimator} $\fprel\ \subseteq\amonoid \times \amonoid$ is a precongruence that is stable under joins of ascending chains and over which the edge functions are monotonic, that is: 
	\begin{compactenum}[({E}1)]
		\item
			$\fprel$ is reflexive and transitive,
		\item
			$\amonval \fprel \amonvalp$ implies $\amonval+\amonvalpp \fprel \amonvalp+\amonvalpp$ for all $\amonval,\amonvalp,\amonvalpp\in\amonoid$,
		\item
			$\bigjoin\achain\fprel\bigjoin\achainp$ for all ascending chains $\achain=\amonval_0\leq\amonval_1\leq\cdots$ and $\achainp=\amonvalp_0\leq\amonvalp_1\leq\cdots$ with $\amonval_i\fprel\amonvalp_i$ for all $i\in\nat$, 
		\item
			$\amonval \fprel \amonvalp$ implies $f(\amonval) \fprel f(\amonvalp)$ for all edge functions. 
	\end{compactenum}
\end{definition}

As we show in \Cref{app:fg-theory}, one can relax the stability under joins, but the above definition is more intuitive than the liberal one. 

\begin{example}
	\label{ex:estimator}
	Coming back to \cref{ex:estimator-motivation}, the following relation is an estimator on the inset flow monoid: \(
		\amonval\simplerel\amonvalp
		\iff
		\amonval=\amonvalp
		\,\lor\,
		\bigl(\,\setc{\amonval,\amonvalp}\cap\setc{\bot,\top}=\emptyset
		\,\land\,
		\amonval\subseteq\amonvalp\,\bigr)
	\).
	It captures updates to the BST that may increase the inset of a node if it is reachable by exactly one path from $\Root$.
	\qed
\end{example}

To lift estimator relations to flow graphs, we need a concept from~\cite{DBLP:conf/tacas/MeyerWW23}. 
They associate with a flow graph its \emph{transfer function} 
\(
	\transformerof{\afg}:((\nat\setminus\setnodes)\times\setnodes\rightarrow \amonoid)\rightarrow(\nat\setminus\setnodes\rightarrow \amonoid)
\) that makes explicit how the fixed point computation for the flow turns inflows  into outflows, \[
	\transformerof{\afg}(\inflow)(\anodep)\ =\  \sum_{\anode\in\afg.\setnodes} \afg[\inflow].\outflow(\anode,\anodep)
  \ .
\]
We set the inflow to the given function, denoted by $\afg[\inflow]$, determine the outflow for the resulting flow graph, and sum up the flow values that are sent from the flow graph to the node of interest.

\begin{definition}[Estimator on Flow Graphs]
	Let $\fprel$ be an estimator on the flow monoid. 
	It induces the estimator $\afg_1 \ctxfprel \afg_2$ on flow graphs defined by $\afg_1.\setnodes = \afg_2.\setnodes$, $\afg_1.\inflow = \afg_2.\inflow$, and
	$\transformerof{\afg_1}\fprel_{\afg_1.\inflow}\transformerof{\afg_2}$. 
	Here, $\transformerof{\afg_1} \fprel_{\afg_1.\inflow} \transformerof{\afg_2}$ means 
	$\transformerof{\afg_1}(\inflow)(\anode) \fprel \transformerof{\afg_2}(\inflow)(\anode)$, for all $\inflow \leq \afg_1.\inflow$ and all $\anode\in\afg_1.\setnodes$.
\end{definition}

The relation $\afg\ctxfprel\afgp$ guarantees the desired stability $\afg\mstar\afgpp\ctxfprel\afgp\imult\afgpp$.
For the ghost multiplication, we would also like to use $\afg\mstar\afgpp$ to give an estimate on $\afgp\imult\afgpp$.
We expect that $\afgpp$ will receive more inflow from the nodes in $\afg$, which are also the nodes of $\afgp$.
As this additional inflow may be funneled back to $\afgp$, also $\afgp$ may receive more inflow from the nodes in $\afgpp$.
This is captured by a family of relations on inflows over the same set of nodes $\setnodes$.
For $\setnodesp\subseteq\nat\setminus\setnodes$, define
\begin{align*}
	\inflow_1\fprelup{\setnodesp}\inflow_2,
	\quad\text{if}\quad
	\restrictto{\inflow_1}{\overline{\setnodesp}\times\setnodes}=\restrictto{\inflow_2}{\overline{\setnodesp}\times\setnodes}
	~~\text{ and }~~
	\sum_{\anodep\in\setnodesp} \inflow_1(\anodep,\anode) \fprel \sum_{\anodep\in\setnodesp} \inflow_2(\anodep,\anode)
	~\text{ for all }\anode\in\setnodes
	\ . 
\end{align*}

This is the fixed-point approximation we work with. 

\begin{theorem}[Shape-Independent Fixed-Point Approximation]
	\label{thm:instantiation-closure}
	Let $\afg\ctxfprel\afgp$ and $\afg\statemultdef\afgc$.
	Then, $\afg\statemult\afgc\ctxfprel\afgp\imult\afgc=\afgp[\inflow_\afgp]\statemult\afgc[\inflow_\afgc]$ with 	$\afg.\inflow\fprelup{\afgc.\setnodes}\inflow_\afgp$ and $\afgc.\inflow\fprelup{\afg.\setnodes}\inflow_\afgc$.
\end{theorem}

\subsection{Instantiation}
\label{sec:instantiation:approximates}

We are now prepared to define the approximate updates.

\smartparagraph{Approximate Physical Update}
The approximate physical update $\absup{\acom}$ strengthens the original update with an estimator relation. 
If this estimator cannot be established, it aborts:
\begin{align*}
	\absupof{\acom}{\astate}=
	\begin{cases}
		\upof{\acom}{\astate}&\qquad\text{if $\upof{\acom}{\astate}\neq\abort$ and for all $\astatep\in\upof{\acom}{\astate}$ we have $\astate\ctxfprel\astatep$}\\
		\abort&\qquad\text{otherwise}. 
	\end{cases}
\end{align*}

\begin{theorem}
	\label{thm:instantiation-physical}
	$\absup{\acom}$ is an approximate physical update:
	\begin{inparaenum}
		\item $\upof{\acom}{\afg}\predleq\absupof{\acom}{\afg}$, and
		\item if $\absupof{\acom}{\afg}\neq\abort$ then $\absupof{\acom}{\afg\mstar\afgc}=\absupof{\acom}{\afg}\imult\afgc$.
	\end{inparaenum}
\end{theorem}

\smartparagraph{Approximate Ghost Multiplication}
We can rely on the estimator relation in the definition of the approximate ghost multiplication.  

\begin{definition}\label{Definition:GhostMultiplication}
	If there is $\afg\ctxfprel\afgp$ with $\afg\statemultdef\afgpp$ or there is a $\afgppp\ctxfprel\afgpp$ with $\afgppp\statemultdef\afgp$, then we define $\ghostabsof{\afgp}{\afgpp} = \setcond{\afgpp[\inflow]}{\afgpp.\inflow\fprelup{\afgp.\setnodes}\inflow}$.
	Otherwise, we set $\ghostabsof{\afgp}{\afgpp}=\abort$.
\end{definition}

We rely on \Cref{thm:instantiation-closure} and define $\ghostabsof{\afgp}{\astatepp}$ as the set of all flow graphs that coincide with~$\astatepp$ except that they have more inflow according to $\fprelup{\afgp.\setnodes}$. 
Then $\ghostabs{\afgp}$ is a closure operator and thus idempotent.
This means the reflexive and transitive closure $\rho^*$ in the definition of $\apredpp$ will be reached after only one iteration, provided the approximate physical update is deterministic.

Recall that $\astatep\absimult \astatepp$ is a derived operation defined as $\ghostabsof{\astatepp}{\astatep}\mstar \ghostabsof{\astatep}{\astatepp}$.

\begin{theorem}
	\label{thm:instantiation-ghost}
	The family of relations $\ghostabs{\afgp}$ defined above is an approximate ghost multiplication. 
	In particular, if there is $\afg\ctxfprel\astatep$ and $\astate\statemultdef\astatepp$ then $\astatep\imult\astatepp \in \astatep\absimult\astatepp$.
\end{theorem}

\smartparagraph{Contextualization for Flow Graphs}
With the instantiation from above we are ready to use the contextualization principle from \cref{sec:contextualization} to handle commands $\semof{\acom}{\astate\mstar\astatepp}$.
To that end, assume the physical update is local to $\astate$, that is, $\absupof{\acom}{\astate}=\astatep$. 
Then, the task is to establish some estimation $\astate\ctxfprel\astatep$.
If $\astate\ctxfprel\astatep$ holds, construct the $\fprel$-closures for $\astatep$ and $\astatepp$.
They are $\apred_{\astatep}=\ghostabsof{\astatepp}{\astatep}$ and $\apred_{\astatepp}=\ghostabsof{\astatep}{\astatepp}$.
Relying on \cref{thm:contextualization}, our instantiation then guarantees that the update satisfies $\semof{\acom}{\astate\mstar\astatepp}\predleq\apred_{\astatep}\mstar\apred_{\astatepp}$.


\subsection{Example: Contextualization for a Binary Search Tree}
\label{sec:bst-example}

We demonstrate our shape-independent contextualization by applying it to a binary search tree (BST).
We highlight how our approach localizes the proof to bounded footprints in scenarios where framing cannot.

For clarity of the exposition, we stay in the sequential setting and 
discuss a sequential BST implementation.
However, our implementation is modeled after find-grained concurrent implementations, like \cite{DBLP:conf/wdag/FeldmanE0RS18,DBLP:conf/europar/CrainGR13,DBLP:conf/ppopp/BronsonCCO10,DBLP:conf/icdcn/RamachandranM15,DBLP:conf/ppopp/DrachslerVY14,DBLP:journals/topc/NatarajanRM20,DBLP:conf/ppopp/RamachandranM15,DBLP:conf/icdcn/RamachandranM15}.
Specifically, the structural updates are similar in both sequential and concurrent implementations. 
In our experiments in \cref{sec:instantiation:automation}, we verify the actual concurrent implementations, and confirm this similarity.

We consider a BST implementation that comes with operations \code{contains}, \code{insert}, and \code{delete} for looking up, inserting, and deleting keys, respectively.
The implementation of these operations is standard.
Operation \code{contains} traverses the tree using binary search.
Operation \code{insert} adds new keys as leaf nodes to the tree if the given key is not already present.
Operation \code{delete} marks nodes as deleted using a \code{del} bit.
Marking nodes is referred to as logical deletion, it changes the contents of the tree but does not modify its structure.
The physical removal, i.e., the unlinking of marked nodes, is performed by dedicated maintenance operations.

Due to space constraints, we elide the implementation of \code{contains}, \code{insert}, and \code{delete} here; they appear in \cref{app:bst-rot}.
Instead, we focus on the maintenance operations \code{removeSimple} and \code{removeComplex}.
They are the interesting part because their updates have an unbounded footprint.

To specify the operations of our BST implementation, we define the predicate $\bst{\abscontent}$ denoting a binary search tree with logical contents $\abscontent$.
With this understanding, we wish to verify the following specification of the maintenance operations:
\[
  \hoareof{\bst{\abscontent}}{~\mcode{removeSimple()}~}{\bst{\abscontent}}
  \quad~~\text{and}~~\quad
  \hoareof{\bst{\abscontent}}{~\mcode{removeComplex()}~}{\bst{\abscontent}}
  \enspace .
\]
In order to tie the logical contents of the specification to the physical state of the implementation, we define $\bst{\abscontent} \defeq \exists\setnodes.~\inv{\abscontent,\setnodes,\setnodes}$.
Predicate $\invraw$ is the structural invariant:
\begin{align*}
  \inv{\abscontent, \setnodesp, \setnodes} ~=~~&
    \Root\in\setnodes \MSTAR \nullptr\notin\setnodes \MSTAR
    \setnodesp\subseteq\setnodes \MSTAR \abscontent=\ctnof{\setnodesp} \MSTAR {\bigmstar}_{\!\!\anode\in\setnodesp~~} \nodeof{\anode} \mstar \ninv{\anode,\setnodes}
  \\
  \ninv{\anode,\setnodes} ~=~~&
    \set{\leftof{\anode},\rightof{\anode}}\subseteq\setnodes\uplus\set{\nullptr}
    ~~\land~~
    \bigl(\isof{\anode}\neq\bot \implies \ctnof{\anode}\subseteq\ksof{\anode})
    ~~\land~~
    \\&
    \isof{\anode} \neq \top
    ~~\land~~
    \bigl(\anode = \Root \implies (-\infty,\infty] \subseteq \isof{\anode} \land \keyof{\anode}=\infty\bigr)
    \enspace .
\end{align*}
The invariant has two main ingredients.
First, it ties the expected logical contents $\abscontent$ to the physical contents $\ctnof{\setnodesp}$ of the region $\setnodesp$.
The set $\ctnof{\setnodesp}=\bigcup_{\anodep\in\setnodesp} \ctnof{\anodep}$ collects the keys of all unmarked nodes, $\ctnof{\anodep}=(\ite{\delof{\anode}}{\emptyset}{\set{\keyof{\anode}}})$.
Second, it carries the resources $\nodeof{\anode}$ for all nodes $\anode\in\setnodesp$ and specifies their properties using the node-local invariant $\ninv{\anode,\setnodes}$ from \cref{ex:estimator-motivation} (repeated for convenience).
Predicate $\nodeof{\anode}$ boils down to a standard points-to predicate, we omit its definition.

We turn to the implementation of \code{removeSimple} and \code{removeComplex}, and prove them correct.

\subsubsection{Simple Removals}
\label{sec:bst:remove-simple}

Operation \code{removeSimple} physically removes (unlinks) nodes that have been marked as deleted.
It is a ``simple'' removal because it unlinks nodes only if they have at most one child.
To satisfy its specification, it must leave the logical contents of the tree unchanged.


\begin{figure}
  \begin{subfigure}[b]{.57\textwidth}
\begin{tikzpicture}[overlay,remember picture]
  \begin{scope}
    \coordinate[yshift=1.7ex] (ctx) at (pic cs:ctx-simple);
    \coordinate[yshift=-1.75ex,xshift=1ex] (end) at (pic cs:ctx-simple-end) {};
    \gettikzxy{(ctx)}{\links}{\oben}
    \gettikzxy{(end)}{\rechts}{\unter}
    \draw[context] (\links,\oben) rectangle (\rechts,\unter-.75pt);
    \draw[footprint] (\links+1ex,\oben-2.5ex) rectangle (\rechts-1ex,\unter);
  \end{scope}
\end{tikzpicture}
\begin{lstlisting}[language=SPL,belowskip=0pt]
$\annot{
  \inv{\abscontent,\setnodes,\setnodes}
}$
def ^removeSimple^() {
  $\x$, _ = find(*);
  $\annot{
    \inv{\abscontent,\setnodes,\setnodes} \MSTAR \x\in\setnodes \MSTAR \isof{\x}\neq\bot
  }$
  $\y$ = $\x$.left;
  if ($\y$ == null) return;
  if (!$\y$.del) return;
  $\annotml{
    &\inv{\abscontent,\setnodes,\setnodes} \MSTAR \x,\y\in\setnodes \MSTAR \delof{\y}\\
    &\!\mstar{} \isof{\x}\neq\bot \MSTAR \leftof{\x}=\y
  }$
  if ($\y$.right == null) {
    // symmetric @\label{code:simple-remove:symmetric}@
  } elif ($\y$.left == null) {
    $\annotml{
      &\inv{\abscontent,\setnodes,\setnodes} \MSTAR \x,\y\in\setnodes \MSTAR \isof{\x}\neq\bot \MSTAR \leftof{\x}=\y \\
      &\!\mstar{} \delof{\y} \MSTAR \leftof{\y}=\nullptr\neq\rightof{\y}
    }$ @\label{code:simple-remove:pre}@
    @\tikzmark{ctx-simple}@  $\ctx{
        \inv{\abscontent_1,\setnodes{\setminus}\setc{\x,\y},\setnodes} \MSTAR \abscontent=\abscontent_1\cup\abscontent_2
      }$  @\tikzmark{ctx-simple-mid}@ @\label{code:simple-remove:context}@
      $\annotml{
        &\inv{\abscontent_2,\setc{\x,\y},\setnodes} \MSTAR \isof{\x}\neq\bot \MSTAR \delof{\y} \\
        &\!\mstar{} \leftof{\x}=\y \MSTAR \leftof{\y}=\nullptr\neq\rightof{\y}
      }$ @\label{code:simple-remove:pre-footprint}@
      $\x$.left = $\y$.right; @\label{code:simple-remove:unlink}@
      $\annotml{
        &\inv{\abscontent_2,\setc{\x,\y},\setnodes} \MSTAR \isof{\x}\neq\bot=\isof{\y} \MSTAR \delof{\y}\\
        &\!\mstar{} \leftof{\x}=\rightof{\y} \MSTAR \leftof{\y}=\nullptr\neq\rightof{\y}
      }$  @\tikzmark{ctx-simple-end}@ @\label{code:simple-remove:post-footprint}@
    $\annot{
      \inv{\abscontent,\setnodes,\setnodes}
    }$ @\label{code:simple-remove:post}@
} }
$\annot{
  \inv{\abscontent,\setnodes,\setnodes}
}$
\end{lstlisting}
  \caption{%
    Implementation and proof outline for simple removals, highlighting the \colorbg{backgroundContext}{\textcolor{colorContext}{context}} and \colorbg{backgroundFootprint}{\textcolor{colorAnnotation}{footprint}} of the update unlinking the marked node $\y$ on \cref{code:simple-remove:unlink}.
    \label{fig:simple-removal:impl}%
    \label{fig:simple-removal:proof}%
  }
  \end{subfigure}
  \hfill
  \begin{subfigure}[b]{.41\textwidth}
  \begin{subfigure}{\textwidth}
  \center
\begin{tikzpicture}[level/.style={sibling distance = 3cm/#1, level distance = 1.1cm}]
  \draw [context] (-2.6,.35) rectangle (2.0,-3.6);
  \node [anchor=south east] at (1.95,-3.55) {Context};
  \draw [footprint] (.5,.25) rectangle (-2.5,-2.65);
  \node [anchor=north west] at (-2.45,.2) {Footprint};

  \node (x) [treenode] {$\x$}
    child [treeptr] {
      node (y) [treenode] {$\y$}
        child [treeptr,-|] {
          node (N) {}
        }
        child [treeptr] {
          node (A) [subtree] {$A$}
        }
    }
    child [treeptr] {
      node (B) [subtree] {$B$}
    }
  ;
  \path (x) -- (y.north) node[midway,xshift=-2mm,flow] {$\ais\cap[-\infty,\x)$};
  \path (x) -- (B.north) node[midway,xshift=2mm,flow] {$\ais\cap(x,\infty]$};
  \path (y) -- (N.north) node[pos=.51,flow] {$\bot$};
  \path (y) -- (A.north) node[midway,xshift=2mm,flowch] {$\ais\cap(y,x)$};
  \node (i) [flow,anchor=north west] at (.85,.2) {$\ais\neq\bot$};
  \path (i.west) edge[inflow,bend right=20] (x);
\end{tikzpicture}
  \caption{%
    Tree structure on \cref{code:simple-remove:pre-footprint}, prior to unlinking $\y$.
    We write $\x$/$\y$ instead of $\keyof{\x}$/$\keyof{\y}$.%
    \label{fig:simple-removal:prestate}%
  }
  \end{subfigure}\\[3mm]
  \begin{subfigure}{\textwidth}
  \center
\begin{tikzpicture}[level/.style={sibling distance = 3cm/#1, level distance = 1.1cm}]
  \draw [context] (-2.6,.35) rectangle (2.0,-3.6);
  \node [anchor=south east] at (1.95,-3.55) {Context};
  \draw [footprint] (.5,.25) rectangle (-2.5,-2.65);
  \node [anchor=north west] at (-2.45,.2) {Footprint};

  \node (x) [treenode] {$\x$}
    child [treeptr] {
      node (y) [treenode] {$\y$} edge from parent[draw=none]
        child [treeptr,-|] {
          node (N) {}
        }
        child [treeptr] {
          node (A) [subtree] {$A$}
        }
    }
    child [treeptr] {
      node (B) [subtree] {$B$}
    }
  ;
  \draw[treeptr] (x) -- (A.north);
  \path (x) -- (A.north) node[pos=.7975,xshift=6mm,flowch] {$\ais\cap[-\infty,\x)$};
  \path (x) -- (B.north) node[midway,flow] {$\ais\cap(\x,\infty]$};
  \path (y) -- (N.north) node[pos=.51,flow] {$\bot$};
  \path (y) -- (A.north) node[pos=.4425,flowch] {$\bot$};
  \node (i) [flow,anchor=north west] at (.85,.2) {$\ais\neq\bot$};
  \path (i.west) edge[inflow,bend right=20] (x);
\end{tikzpicture}
  \caption{%
    Tree structure after unlinking $\y$, \cref{code:simple-remove:post-footprint}.%
    \label{fig:simple-removal:poststate}%
  }
  \end{subfigure}
  \end{subfigure}
  \caption{%
    A \emph{simple} removal unlinks internal, marked nodes if they have at most one child.
    The operation does not alter the logical contents of the tree.
    Unlinking marked nodes with two children is done by \emph{complex} removals.%
    \label{fig:simple-removal}%
  }
\end{figure}


The implementation and proof outline of \code{removeSimple} are given in \cref{fig:simple-removal:impl}. 
Starting from an arbitrary node $\x$ currently linked into the tree, the left child $\y$ of $\x$ is read out.
If $\y$ is a non-null marked node with at most one child, \code{removeSimple} unlinks $\y$.
We focus on the case where $\y$ has no left child.
Then, it is unlinked by making $\y$'s right child the left child of $\x$, i.e., by updating $\selof{\x}{left}$ to $\rightof{\y}$ on \cref{code:simple-remove:unlink}.
\Cref{fig:simple-removal:prestate,fig:simple-removal:poststate} illustrate the pre- and post-state of the update for the part of the tree rooted in $\x$.
The \colorbg{backgroundFootprint}{footprint} of the physical unlinking contains just the nodes $\x$ and~$\y$.
The proof for the update of the footprint is as expected because it is readily checked that the invariant is maintained for the nodes $\x$ and $\y$.
This is the transition from \cref{code:simple-remove:pre-footprint} to \cref{code:simple-remove:post-footprint}.

The accompanying ghost update, however, is unbounded, it oozes into the right subtree $A$ of $\y$.
Coming back to \cref{fig:simple-removal:prestate,fig:simple-removal:poststate}, let $\ais=\isof{\x}\neq\bot$ be the inset of $\x$.
Before the update, $A$-bound searches follow the edges $\leftof{\x}$ and $\rightof{\y}$.
That is, the inset of $A$ is $\ais\cap(\keyof{\y},\keyof{\x})$.
After the update, $\y$ is skipped and the inset of $A$ is $\ais\cap[-\infty,\keyof{x})$.
Because the keys in $A$ are larger than $\keyof{\y}$, the additional inset $\ais\cap[-\infty,\keyof{\y}]$ after the update is forwarded to the left-most leaf in $A$.
That is, the ghost footprint is unbounded and the \ruleref{frame} rule does not apply.

The \ruleref{context} rule, however, does apply.
We have already seen that $\simplerel$ is an estimator that captures updates to the tree that increase the inset of reachable nodes, like the one we have here.
The invariant is $\simplerel$-closed because the inset/keyset occurs only on the right-hand sight of inclusions, the inset of unreachable nodes remains $\bot$, and no node's inset becomes $\top$.
That is, we can \colorbg{backgroundContext}{contextualize} everything outside the footprint $\set{\x,\y}$.
This is the annotation on \cref{code:simple-remove:context}.

Altogether, $\choareof{\set{\text{\Cref{code:simple-remove:context}}}}{\text{\Cref{code:simple-remove:pre-footprint}}}{\text{\Cref{code:simple-remove:unlink}}}{\text{\Cref{code:simple-remove:post-footprint}}}$ is a valid CASL statement.
Rule \ruleref{context} thus yields $\choareof{\emp}{\text{\Cref{code:simple-remove:pre}}}{\text{\Cref{code:simple-remove:unlink}}}{\text{\Cref{code:simple-remove:post}}}$ and by relative soundness this is a valid statement in classical separation logic.
That is, \code{removeSimple} does not alter the logical contents of the tree.

\subsubsection{Complex Removals}
\label{sec:bst:remove-complex}

Operation \code{removeComplex} unlinks marked nodes that have two children, without changing the logical contents of the tree.
\Cref{fig:complex-removal} gives the implementation and proof outline.
There are four steps in \code{removeComplex}.
First, it obtains a reachable, marked node $\x$ with two children.
Second, it uses the helper \code{findSucc} to find the left-most leaf $\y$ and its parent $\p$ in the subtree $B$ rooted at $\x$'s right child.
The helper implementation and proof are straightforward, we defer it to \cref{app:bst-rotate}.
To avoid distracting case distinctions, we assume $\p\neq\x$.
Third, the contents of $\x$ and $\y$ are swapped, \crefrange{code:complex-remove:move}{code:complex-remove:del-y}.
This is the most interesting part and we discuss it below.
Last, $\y$ is unlinked.
The unlinking is as in \code{removeSimple} because $\y$ has at most one child.


\begin{figure}
  \begin{subfigure}[b]{.6\textwidth}
\begin{tikzpicture}[overlay,remember picture]
  \begin{scope}
    \coordinate[yshift=3ex] (ctx) at (pic cs:ctx-complex-move);
    \coordinate[yshift=-2ex,xshift=6ex] (end) at (pic cs:ctx-complex-move-end) {};
    \gettikzxy{(ctx)}{\links}{\oben}
    \gettikzxy{(end)}{\rechts}{\unter}
    \draw[context] (\links,\oben) rectangle (\rechts,\unter-.75pt);
    \draw[footprint] (\links+1ex,\oben-5ex) rectangle (\rechts-1ex,\unter);
  \end{scope}
  \begin{scope}
    \coordinate[yshift=2ex] (ctx) at (pic cs:ctx-complex-unlink);
    \coordinate[yshift=-.6ex] (end) at (pic cs:ctx-complex-unlink-end) {};
    \coordinate[yshift=-3ex,xshift=6ex] (width) at (pic cs:ctx-complex-move-end) {};
    \gettikzxy{(ctx)}{\links}{\oben}
    \gettikzxy{(end)}{\irgendwas}{\unter}
    \gettikzxy{(width)}{\rechts}{\irgendwer}
    \draw[context] (\links,\oben) rectangle (\rechts,\unter-.75pt);
    \draw[footprint] (\links+1ex,\oben-2.8ex) rectangle (\rechts-1ex,\unter);
  \end{scope}
\end{tikzpicture}
\begin{lstlisting}[language=SPL]
$\color{colorAnnotation}\ctxupclosed{\setnodesp,\aks} \mkern+24.4mu~\defeq~ \forall\z\in\setnodesp.~\isof{\z}\neq\bot\implies\keyof{\z}\notin\aks$
$\color{colorAnnotation}\invp{\abscontent,\z,\setnodes} ~\defeq~ \nodeof{\z} \mstar \abscontent=\ctnof{\z} \mstar \setc{\leftof{\z},\rightof{\z}}\subseteq\setnodes\uplus\setc{\nullptr}$
\end{lstlisting}
\begin{lstlisting}[language=SPL,belowskip=0pt]
$\annot{
  \inv{\abscontent,\setnodes,\setnodes}
}$
def ^removeComplex^() {
  $\x$, _ = find(*);
  if (!$\x$->del) return;
  if ($\x$->left == NULL) return;
  if ($\x$->right == NULL) return;
  $\p$, $\y$ = findSucc($\x$); @\label{code:complex-remove:find-succ}@
  let $\akey_\y=\keyof{\y}$, $\aks=\isof{\x}\cap(\keyof{\x},\akey_\y]$;
  $\annotml{
    &\inv{\abscontent,\setnodes,\setnodes} \MSTAR \x,\y,\p\in\setnodes \MSTAR \isof{\x}\neq\bot\neq\isof{\p} \MSTAR \delof{\x} \\
    &\!\mstar{} \leftof{\x}\neq\nullptr\neq\rightof{\x} \MSTAR \leftof{\p}=\y \MSTAR \leftof{\y}=\nullptr \\
    &\!\mstar{} \keyof{\y}\prall{<}\keyof{\p} \MSTAR \keyof{\x},\akey_\y\in\isof{\x} \MSTAR \aks\subseteq\ksof{\y}
  }$ @\label{code:complex-remove:pre}@
  @\tikzmark{ctx-complex-move}@  $\ctxml{
      &\inv{\abscontent_1,\setnodes{\setminus}\setc{\x,\y},\setnodes} \MSTAR \abscontent=\abscontent_1\cup\abscontent_2  \MSTAR \ctxupclosed{\setnodes{\setminus}\setc{\x,\y},\aks} \\
      &\!\mstar{} \p\in\setnodes \MSTAR \isof{\p}\neq\bot \MSTAR \akey_\y\prall{<}\keyof{\p} \MSTAR \leftof{\p}=\y
    }$ @\label{code:complex-remove:ctx}@
    $\annotml{
      &\inv{\abscontent_2,\setc{\x,\y},\setnodes} \MSTAR \keyof{\x},\akey_\y\in\isof{\x} \MSTAR \leftof{\y}=\nullptr \\
      &\!\mstar{} \delof{\x} \MSTAR \leftof{\x}\neq\nullptr\neq\rightof{\x} \MSTAR \keyof{\y}=\akey_\y
    }$ @\label{code:complex-remove:footprint-pre}@
    $x$->key = $\y$->key; @\label{code:complex-remove:move}@
    $\annotml{
      &\inv{\emptyset,\setc{\x},\setnodes} \MSTAR \invp{\abscontent_2,\y,\setnodes} \MSTAR \akey_\y\in\isof{\x} \\
      &\!\mstar{} \delof{\x} \MSTAR \leftof{\y}=\nullptr \MSTAR \keyof{\x}=\keyof{\y}=\akey_\y
    }$ @\label{code:complex-remove:footprint-post-move}@
    $\x$->del = $\y$->del;@\label{code:complex-remove:del-x}@ $\y$->del = true;@\label{code:complex-remove:del-y}@
    $\annotml{
      &\inv{\abscontent_2,\setc{\x},\setnodes} \MSTAR \invp{\emptyset,\y,\setnodes} \MSTAR \akey_\y\in\isof{\x} \\
      &\!\mstar{} \delof{\x} \MSTAR \leftof{\y}=\nullptr \MSTAR \keyof{\x}=\keyof{\y}=\akey_\y
    }$  @\tikzmark{ctx-complex-move-end}@ @\label{code:complex-remove:footprint-post}@
  $\annotml{
      &\inv{\abscontent,\setnodes{\setminus}\setc{\y},\setnodes} \MSTAR \invp{\emptyset,\y,\setnodes} \MSTAR \isof{\p}\neq\bot \MSTAR \delof{\y} \\
      &\!\mstar{} \x,\p\in\setnodes \MSTAR \leftof{\p}=\y \MSTAR \keyof{\y}\prall{<}\keyof{\p} \MSTAR \leftof{\y}=\nullptr
  }$ @\label{code:complex-remove:intermediate}@
  @\tikzmark{ctx-complex-unlink}@  $\ctx{
      \inv{\abscontent_1,\setnodes{\setminus}\setc{\y,\p},\setnodes} \MSTAR \abscontent=\abscontent_1\cup\abscontent_2 \MSTAR \x\in\setnodes
    }$
    $\annotml{
      &\inv{\abscontent_2,\setc{\p},\setnodes} \MSTAR \invp{\emptyset,\setc{\y},\setnodes} \MSTAR \keyof{\y}\prall{<}\keyof{\p} \\
      &\!\mstar{} \isof{\p}\neq\bot \MSTAR \leftof{\p}=\y \MSTAR \delof{\y} \MSTAR \leftof{\y}=\nullptr
    }$
    $\p$->left = $\y$->right; @\label{code:complex-remove:unlink}@
    $\annot{
      \inv{\abscontent_2,\setc{\p},\setnodes} \mstar \invp{\emptyset,\setc{\y},\setnodes} \mstar \isof{\y}=\bot \mstar \delof{\y}
    }$
    $\annot{
      \inv{\abscontent_2,\setc{\y,\p},\setnodes}
    }$  @\tikzmark{ctx-complex-unlink-end}@
  $\annot{
    \inv{\abscontent,\setnodes,\setnodes} \MSTAR \x,\y,\p\in\setnodes
  }$ @\label{code:complex-remove:post}@
}
$\annot{
  \inv{\abscontent,\setnodes,\setnodes}
}$
\end{lstlisting}
  \caption{%
    Implementation and proof outline for complex removals, highlighting the \colorbg{backgroundContext}{\textcolor{colorContext}{context}} and \colorbg{backgroundFootprint}{\textcolor{colorAnnotation}{footprint}} of the restructuring updates.
    \label{fig:complex-removal:proof}%
  }
  \end{subfigure}
  \hfill
  \begin{subfigure}[b]{.35\textwidth}
    \begin{subfigure}{\textwidth}
  \centering
\begin{tikzpicture}[level/.style={sibling distance = 2.5cm/#1, level distance = 1.15cm}]
  \draw [context] (-2,.35) rectangle (2.7,-5.35);
  \node [anchor=south west] at (-1.95,-5.3) {Context};
  \path [footprint] (-.4,.25) -- (.4,.25) -- (1.1,-3.7) -- (-.4,-3.7);
  \node [anchor=north west,rotate=90] at (-.35,-3.65) {Footprint};

  \node (x) [treenode] {$\x$}
    child [treeptr] {
      node (A) [subtree] {$A$}
    }
    child [treeptr] {
      node (B) [subtree,] {$B$}
        child {
          node (y) [treenode] {$\y$} edge from parent[draw=none]
            child [treeptr,-|] {
              node (N) {}
            }
            child [treeptr] {
              node (C) [subtree] {$C$}
            }
        }
        child { node {} edge from parent[draw=none] }
    }
  ;
  \path (x) -- (A.north) node[midway,xshift=-2mm,flowch] {$\ais\cap[-\infty,\x)$};
  \path (x) -- (B.north) node[midway,xshift=2mm,flowch] {$\ais\cap(\x,\infty]$};
  \path[inflow] (B.213) -- (y) node[midway,flow] {$\aisp\neq\bot$};
  \path (y) -- (N.north) node[midway,flow] {$\bot$};
  \path (y) -- (C.north) node[midway,flow,yshift=.35pt,xshift=6mm] {$\aisp\cap(\y,\infty]$};
  \node (i) [flow,anchor=north west] at (.85,.2) {$\ais\neq\bot$};
  \path (i.west) edge[inflow,bend right=20] (x);
\end{tikzpicture}
  \caption{%
    Tree structure on \cref{code:complex-remove:move}, prior to copying $\y$'s key to $\x$.
    We write $\x$/$\y$ instead of $\keyof{\x}$/$\keyof{\y}$.
    \label{fig:complex-removal:prestate}%
  }
  \end{subfigure}\\[3mm]
  \begin{subfigure}{\textwidth}
  \centering
\begin{tikzpicture}[level/.style={sibling distance = 2.5cm/#1, level distance = 1.15cm}]
  \draw [context] (-2,.35) rectangle (2.7,-5.35);
  \node [anchor=south west] at (-1.95,-5.3) {Context};
  \path [footprint] (-.4,.25) -- (.4,.25) -- (1.1,-3.7) -- (-.4,-3.7);
  \node [anchor=north west,rotate=90] at (-.35,-3.65) {Footprint};

  \node (x) [treenode] {$\y$}
    child [treeptr] {
      node (A) [subtree] {$A$}
    }
    child [treeptr] {
      node (B) [subtree,] {$B$}
        child {
          node (y) [treenode] {$\y$} edge from parent[draw=none]
            child [treeptr,-|] {
              node (N) {}
            }
            child [treeptr] {
              node (C) [subtree] {$C$}
            }
        }
        child { node {} edge from parent[draw=none] }
    }
  ;
  \path (x) -- (A.north) node[midway,xshift=-2mm,flowch] {$\ais\cap[-\infty,\y)$};
  \path (x) -- (B.north) node[midway,xshift=2mm,flowch] {$\ais\cap(\y,\infty]$};
  \path[inflow] (B.213) -- (y) node[midway,flow] {$\aisp\neq\bot$};
  \path (y) -- (N.north) node[midway,flow] {$\bot$};
  \path (y) -- (C.north) node[midway,flow,yshift=.35pt,xshift=6mm] {$\aisp\cap(\y,\infty]$};
  \node (i) [flow,anchor=north west] at (.85,.2) {$\ais\neq\bot$};
  \path (i.west) edge[inflow,bend right=20] (x);
\end{tikzpicture}
  \caption{%
    Tree structure after copying $\y$'s key to $\x$, \cref{code:complex-remove:move}.
    \label{fig:complex-removal:poststate}%
  }
  \end{subfigure}
  \end{subfigure}
  \caption{%
    A \emph{complex} removal unlinks internal, marked nodes with two children, not altering the logical contents.%
    \label{fig:complex-removal}%
  }
\end{figure}


We turn the discussion to the third step.
To avoid confusion between the values of fields before and after the following updates, we record the values $\akey_\x=\keyof{\x}$, $\akey_\y=\keyof{\y}$, and $\ais=\isof{\x}$ from before the update, as on \cref{code:complex-remove:pre}.
Note that $\akey_\x<\akey_\y$.
Now, \cref{code:complex-remove:move} copies $\akey_\y$ into $\x$.
This is challenging due to its intricate unbounded flow footprint, which is visualized in \cref{fig:complex-removal:prestate,fig:complex-removal:poststate}.
The update increases the inset of $\x$'s left subtree $A$ from $\ais\cap[-\infty,\akey_\x)$ to $\ais\cap[-\infty,\akey_\y)$.
The inset of $\x$'s right subtree $B$, in turn, decreases from $\ais\cap(\akey_\x,\infty]$ to $\ais\cap(\akey_\y,\infty]$.
That is, the portion $\aks=\ais\cap(\akey_\x,\akey_\y]\subseteq\ksof{\y}$ is redirected from $B$ to $A$.
The estimator $\complexrel$ defined by \[
  \amonval\complexrel\amonvalp
  ~~\defifff~~
    \amonval\simplerel\amonvalp
    \,~\lor~
    \bigl(\,
      \setc{\amonval,\amonvalp}\cap\setc{\bot,\top}=\emptyset ~\land~ \akey_\x\notin\amonval ~\land~ \amonval\setminus\aks\subseteq\amonvalp
    \,\bigr)
\]
captures that insets may
\begin{inparaenum}
	\item increase arbitrarily, or
	\item decrease by up to $\aks$, if they do not contain $\akey_\x$.
\end{inparaenum}
The side condition localizes the decrease to the subtrees of $\x$.

For handling the update, we choose nodes $\x$ and $\y$ as the footprint and contextualize everything else.
Technically, the context is $\inv{\abscontent_1,\setnodes\setminus\set{\x,\y},\setnodes}$ and the footprint is $\inv{\abscontent_2,\set{\x,\y},\setnodes}$ with $\abscontent=\abscontent_1\cup\abscontent_2$.
The main challenge is to show that the context is $\complexrel$-closed, in particular, tolerates reduced insets.
To that end, consider a node $\z\in\setnodes\setminus\set{\x,\y}$ with $\isof{\z}\neq\bot$.
By the invariant, $\ctnof{\z}\subseteq\isof{\z}$.
If $\z$ is marked, then $\ctnof{\z}=\emptyset$ and the inclusion is vacuously true.
Otherwise, we are obliged to show $\keyof{\z}\notin\aks$ to preserve the inclusion.
Because $\aks\subseteq\ksof{\y}$ and $\keyof{\z}\in\ksof{\z}$, it suffices to show $\ksof{\y}\cap\ksof{\z}=\emptyset$.
This follows from a result due to \citet{DBLP:journals/tods/ShashaG88}, stating that the keyset of all nodes are pairwise disjoint if all searches are deterministic and start in a dedicated root node.\footnote{%
	We could encode this into the invariant $\invraw$ using the keyset algebra proposed in~\cite{DBLP:conf/pldi/KrishnaPSW20}, but refrain from the added complexity.
}
Since the requirements are satisfied in our setting, the result applies and discharges our proof obligation.
Overall, the context $\inv{\abscontent_1,\setnodes\setminus\set{\x,\y},\setnodes}$ on \cref{code:complex-remove:ctx} is $\complexrel$-closed.

We turn to the footprint $\inv{\abscontent_2,\set{\x,\y},\setnodes}$, \cref{code:complex-remove:footprint-pre}.
The physical update changing $\selof{\x}{key}$ from $\akey_\x$ to $\akey_\y$ is as expected.
\Cref{thm:contextualization} prescribes that the footprint be $\complexrel$-closed after the update.
Because $\akey_\x,\akey_\y\in\ais$ prior to the update, $\complexrel$ does not remove $\akey_\y$ from the inset of $\x$.
Together with $\delof{\x}$, we have $\inv{\emptyset,\set{\x},\setnodes}$ after the update.
As expected, the inset of $\y$ decreases by up to $\aks$---we do not know the exact loss, and we do not care.
Consequently, the invariant of $\y$ breaks because it may no longer receive its key $\akey_\y$ in its inset.
The subsequent actions will reestablish the invariant for $\y$.
Overall, we arrive at the postcondition on \cref{code:complex-remove:footprint-post-move}.

\Cref{code:complex-remove:del-x,code:complex-remove:del-y} finalizes the swap of $\x$ and $\y$.
The update, which does not alter the flow, is straightforward and yields $\inv{\abscontent_2,\set{\x},\setnodes}$ and $\ctnof{\y}=\emptyset$.
Finally, \cref{code:complex-remove:unlink} unlinks $\y$, as in the case for \code{removeSimple}.
Unlinking the marked $\y$ reestablishes its invariant.
We arrive at $\inv{\abscontent,\setnodes,\setnodes}$, \cref{code:complex-remove:post}, as required.


\subsection{Proof Automation}
\label{sec:instantiation:automation}

We implemented contextualization for flow graphs in the proof outline checker \nekton \cite{DBLP:conf/cav/MeyerOWW23} and applied it successfully to three challenging concurrent balanced binary trees, the FEMRS tree \cite{DBLP:conf/wdag/FeldmanE0RS18}, the contention-friendly tree (CFBST) \cite{DBLP:conf/europar/CrainGR13}, and the practical concurrent tree (PCBST) \cite{DBLP:conf/ppopp/BronsonCCO10}.
Our proofs establish
\begin{inparaenum}
 	\item functional correctness for the fixed linearization points of the algorithms, and
 	\item that the maintenance operations (rotations and removals) do not alter the logical contents of the tree.
\end{inparaenum}
This is the first formal proof of the FEMRS tree's maintenance operations and, as far as we know, the first proofs of CFBST and PCBST.
Our version of \nekton cannot deal with non-fixed linearization points of the non-blocking \code{contains} method due to imprecision in the tool's hindsight reasoning (see below).
However, we believe that this is an orthogonal concern to the contextualization presented here, because only the maintenance operations suffer from an unbounded footprint.
Detailed results are given in \Cref{table:eval}.
The blow up in runtime comes from the fact that \nekton handles the computation of strongest postconditions for disjunctions suboptimally (it computes a disjunctive normal form), resulting in a large number of SMT queries per atomic step in the proof.
We believe this can be alleviated with a more careful encoding.

\begin{table*}%
	\caption{%
		Experimental results for checking proofs of challenging concurrent tree implementations with \nekton.
		Our results include the lines of code (\#code), lines of proof annotations (\#proof), lines of invariant/flow definitions (\#def), the ratio between \#code and \#proof+\#def, and the runtime and verdict of the proof check. Runtimes are averaged across 10 runs. The evaluation was conducted on an Apple M1 Pro.}%
	\label{table:eval}%
		%
		%
		%
		%
		%
		%
		\center%
		\newcommand{\cell}[1]{\makebox[1cm][c]{\(#1\)}}
		\newcommand{\cellYes}[1]{\makebox[1cm][r]{\(#1\)}\hspace{1.5mm}\makebox[.4cm][l]{\symbolYes}}
		\setlength{\tabcolsep}{4pt}
		\begin{tabularx}{\textwidth}{Xlccccc}%
			\toprule
			Benchmark
				& Properties verified
				& \#code
				& \#proof
				& \#def
				& Ratio
				& ~Time
				\\
			\midrule
			FEMRS tree \cite{DBLP:conf/wdag/FeldmanE0RS18}
				& fixed LPs, maintenance
				& \cell{251}
				& \cell{318}
				& \cell{95}
				& \cell{1:1.65}
				& \cellYes{34m}
				\\
			\midrule
			CFBST \cite{DBLP:conf/europar/CrainGR13}
				& fixed LPs, maintenance
				& \cell{257}
				& \cell{276}
				& \cell{107}
				& \cell{1:1.49}
				& \cellYes{45m}
				\\
			\midrule
			PCBST \cite{DBLP:conf/ppopp/BronsonCCO10}
				& all LPs, maintenance
				& \cell{444}
				& \cell{657}
				& \cell{115}
				& \cell{1:1.74}
				& \cellYes{101m}
				\\
			\bottomrule
	\end{tabularx}
\end{table*}

\nekton takes as input the program under scrutiny, its proof outline (separation logic assertions in the program), and a flow domain specifying the flow monoid and the edge functions being used.
It then checks whether the proof is valid, i.e., is a valid derivation using the rules from \Cref{Figure:PL}.
This step relies on entailment checking procedures tailored towards flows.

We adapted \nekton to support contextualization.
To that end, the flow domains that \nekton accepts as input are extended with an estimator.
Then, whenever an update is not local to the footprint that \nekton constructs, we use the given estimator and check whether the flow update is compatible with the estimator (\nekton guarantees that the physical update is contained in the footprint).
In terms of \Cref{thm:instantiation-closure}, this means we check $\afg\ctxfprel\afgp$ where $\afg$ and $\afgp$ is the footprint before and after the update, respectively.
To conform with the upward-closure requirement for such updates on the entire flow graph, $\afgp\imult\afgpp$ in \Cref{thm:instantiation-closure} where $\afgpp$ is the context, we have \nekton check that all assertions in the proof are stable wrt. the given estimator.
We do this for all assertions and not just the post assertion of the update because the update could be performed by an interfering thread at \emph{any time}, it becomes part of the interference set alluded to in \Cref{sec:og}.
With these adaptions, \nekton is able to validate our proof of the FEMRS, CFBST, and PCBST trees.

Our benchmark set is relatively small and we did not validate non-fixed linearzation points.
This is not due to incompatibility with our theory, but rather due to orthogonal challenges with \nekton.
The main problem is that \nekton breaks down interferences into per-address interferences, resulting in imprecision.
On the one hand, imprecise interferences hinder hindsight reasoning and thus our ability to validate non-fixed linearzation points.
On the other hand, it forces us to apply the upward-closure to all nodes in the context, even though some locking strategies in fine-grained concurrent trees can prevent nodes from actually experiencing a change in inflow.
As a consequence, some estimators do not work with \nekton although they allow for valid proofs in principle.
We did not improve \nekton's handling of interferences for this paper as this would require a rewrite of large parts of the code base to deal with orthogonal concerns.
We note that the contextualization that goes into our proofs reflect the estimators required for numerous other concurrent trees \cite{DBLP:conf/spaa/HowleyJ12,DBLP:conf/podc/EllenFRB10,DBLP:conf/ppopp/BrownER14,DBLP:conf/podc/ArbelA14,DBLP:conf/icdcn/RamachandranM15,DBLP:conf/ppopp/RamachandranM15,DBLP:conf/ppopp/Drachsler-Cohen18,DBLP:journals/topc/NatarajanRM20}.


\section{Related Work and Conclusion}
\label{sec:related}

The work closest to ours is a separation logic for establishing worst-case space complexity bounds of higher-order programs under garbage collection semantics \cite{DBLP:journals/pacmpl/MoineCP23}.
A core aspect of this work is to capture the \emph{roots} of heap graphs, i.e., memory addresses that are referenced from stack variables, because the root-reachable memory cannot be garbage collected. To that end, they introduce so-called \emph{stackable} assertions $\mathsf{Stackable}(l, p)$ for a location $l$ and a fractional permission $p$. Then, referencing location $l$ from a stack variable consumes a fraction of the stackable assertion. Once the stack variable goes out of scope, the stackable assertion is regained so that $l$ can be garbage collected.

\citet{DBLP:journals/pacmpl/MoineCP23} then observe that stackable assertions add a certain complexity to proofs. To alleviate this and allow for more automation, they adapt their program logic to take the form $\langle R \rangle \{ \phi \} \astmt \{ \psi \}$. Here, $R$ is a so-called \emph{souvenir}, which keeps track of the addresses for which a stackable assertion has been consumed (i.e., the set of addresses referenced from the stack). The semantics of souvenirs requires that the actually consumed stackable assertions are maintained by $\astmt$ and contain \emph{at least} the locations $R$. Note that the souvenir does not capture the stackable assertions exactly, nor does it make an assumption about their available fractions. Thus, souvenirs can be seen as another concrete instance of our development: a souvenir is a context and its construction aligns with our idea of upward closures with respect to an estimator.

The other work closest to ours is on the modular verification of reachability properties~\cite{DBLP:journals/pacmpl/Ter-GabrielyanS19}. 
The goal is to understand how the change of reachability in a subheap impacts the reachability in the overall heap, a problem referred to as reachability framing.
The contribution is a recompution method for so-called relatively-convex footprints.
As reachability information can be encoded into flows, the reachability framing problem can be cast as an approximation of a ghost multiplication for the corresponding flow graphs.
We give an approximation method for general flow graphs, and so have to work around the absence of domain-specific knowledge.
Our approach is to introduce estimator relations, which may be seen as distant relatives of relative convexity.

The main difference is that our work studies the impact of local changes on a context in a general setting, which leads to the notion of context-aware predicate transformers, the \eqref{Equation:Mediation} property, and the development of a program logic that has the new \ruleref{context} to frame the (known) context.
Another difference is that, inspired by classical framing, our approach strives for invariance of the context predicate, while the mentioned work embraces change.
It is an interesting problem for future work to embed the modification of ghost information in the context into a program logic.

Also related to our work is the ramification rule for separation logic~\cite{DBLP:conf/popl/HoborV13} (of which the principle developed in~\cite{DBLP:journals/pacmpl/Ter-GabrielyanS19} can be seen as a concrete instance).
Ramifications ease local reasoning about overlaid structures.
The rule says that to prove a global specification $\hoareof{\apred'}{\astmt}{\apredp'}$, one can focus on a more local one $\hoareof{\apred}{\astmt}{\apredp}$ provided $\apred'\subseteq \apred\mstar(\apredp\sepimp\apredp')$.
Indeed, the proof for \code{remove} that we gave in Section~\ref{sec:motivation} to motivate our work is via ramification.
As discussed there, the predicate $\htreepred$ is hard to work with, and it is precisely the separating implication $\apredp\sepimp\apredp'$ involved in ramification.

While this is the technical link, our work has a different goal than ramification, namely to localize footprints in cases where they become unbounded. 
In the settings of interest to us, the triple $\hoareof{\apred}{\astmt}{\apredp}$ cannot be proven in the first place, because the computation aborts due to missing resources.
Our way out was to propose context-aware reasoning $\choareof{\apredpp}{\apred}{\astmt}{\apredp}$, which guarantees that the missing resources can be found in the context $\apredpp$.
While ramification is based on the \ruleref{frame} rule, we had to integrate context-aware reasoning deeply into the program logic, down to the semantics that we had to change to context-aware predicate transformers.
What came as a surprise was that every predicate transformer can be made context-aware via the induced semantics.

We share the goal of localizing unbounded footprints with the recent work~\cite{DBLP:journals/pacmpl/MeyerWW22}.
Their technique applies in cases where the unbounded footprint is traversed prior to the data structure's update.
It relies on the traversal to build up a predicate that captures the update's effect, very much in the style of ramifications.
Here, we consider the missing case that the unbounded footprint is not traversed, but still influenced by the propagation of updated ghost information.
We observe that the essential data structure invariants are invariant under such modifications, and develop the \ruleref{context} rule to frame out the corresponding parts (although they undergo modifications).

The context $\apredpp$ in \theLogic specifications $\choareof{\apredpp}{\apred}{\astmt}{\apredp}$ looks similar to a resource invariant in concurrent separation logic~\cite{DBLP:conf/concur/Brookes04,DBLP:conf/concur/OHearn04}.
The proof rule for atomic blocks in CSL temporarily adds the resource invariant to the specification in order to prove the atomic block.
Our \ruleref{context} rule proceeds the other way around and subtracts the context from the state to be able to reason locally.
The difference becomes clear when seeing the rules side-by-side:
$$
			\inferH{atomic}{
				\choareOf{\emp}{\apred\mstar \apredpp}{\astmt}{\apredp\mstar \apredpp}
				}{
				\choareOf{\apredpp}{\apred}{\text{\bfseries atomic}\ \astmt}{\apredp}
				}\hspace{3cm}
							\inferHlab{context-dup}{context}{
				\choareOf{\acontext}{\apred}{\astmt}{\apredp}
				}{
				\choareOf{\emp}{\apred\mstar \apredpp}{\astmt}{\apredp\mstar \apredpp}\ .
                              }$$
                            
We also address the contextualization problem, the problem of determining a predicate $\apredpp$ capturing substate that remains invariant under transitions.
This is related to the resource invariant synthesis problem addressed in~\cite{DBLP:conf/pldi/GotsmanBCS07,DBLP:conf/aplas/CalcagnoDV09}.
The concurrent setting suggests a thread-modular analaysis and a focus on locks.
Also related to contextualization is bi-abduction where, given predicates $\apred$ and $\apredp$, the task is to infer a frame $\apredp'$ as an unneeded part of the state and an anti-frame $\apred'$ as a missing premise, so that $\apred\mstar\apred'\subseteq\apredp\mstar\apredp'$ holds.
Our work is about ghost state, and our goal is to approximate the ghost multiplication.
This brings the new problem of approximating fixed points over heap graphs whose shape is not known.
Bi-abduction assumes to know the recursive predicates, and therefore the approach does not seem to carry over.

Beyond bi-abduction there is a rich literature on entailment checking and frame inference for separation logic with recursive predicates (see, e.g., \cite{DBLP:journals/jacm/CalcagnoDOY11, DBLP:conf/cade/BrotherstonDP11,DBLP:conf/pldi/PekQM14,DBLP:conf/sas/ToubhansCR14,DBLP:conf/nfm/EneaLSV17,DBLP:conf/cav/DardinierPWMS22,DBLP:journals/tocl/MathejaPZ23}). However, these works are limited to reasoning about tree-like structures without sharing and often do not extend to functional correctness properties. The flow framework~\cite{DBLP:journals/pacmpl/KrishnaSW18,DBLP:conf/esop/KrishnaSW20,DBLP:conf/tacas/MeyerWW23} aims to provide a shape-agnostic formalism for reasoning about rich inductive properties of general graphs.

We already discussed the connection to the morphism framework~\cite{DBLP:journals/pacmpl/Nanevski0DF19, DBLP:journals/pacmpl/FarkaN0DF21} which inspired our ghost multiplication.
Program logics like Iris~\cite{DBLP:journals/jfp/JungKJBBD18}, CAP~\cite{DBLP:conf/ecoop/Dinsdale-YoungDGPV10}, and TaDA~\cite{DBLP:conf/ecoop/PintoDG14} also provide mechanisms for introducing rich ghost state abstractions.
There, the ghost state exists only at the level of the logic and is coupled with the physical state via resource invariants.
Rather than letting the program semantics update the ghost state, the prover has the responsibility to update the ghost state via logical view shifts whenever the physical state changes and the invariant would be violated.
As view shifts must be frame-preserving, this implies that updates can still entail large footprints at the logical level.
Our work extends to these settings in the cases where the required logical view shifts are uniquely determined by the physical updates.

We implemented our approach in the proof checker \nekton \cite{DBLP:conf/cav/MeyerOWW23}.
We note that our improvements to the tool are orthogonal to techniques implemented in other proof checkers, like 
\atoolname{GRASShopper} \cite{DBLP:conf/tacas/PiskacWZ14},
\atoolname{CIVL} \cite{DBLP:conf/cav/HawblitzelPQT15,DBLP:conf/cav/KraglQ18},
\atoolname{Caper} \cite{DBLP:conf/esop/Dinsdale-YoungP17},
\atoolname{Starling} \cite{DBLP:conf/cav/WindsorDSP17},
\atoolname{Anchor} \cite{DBLP:journals/pacmpl/FlanaganF20},
\atoolname{Voila} \cite{DBLP:conf/fm/WolfSM21}, and
\atoolname{Diaframe} \cite{DBLP:conf/pldi/MulderKG22}.
These tools do not aim to simplify the reasoning about unbounded ghost state updates.



	\bibliography{dblp,bib}

	\appendix


\section{Contextual Reasoning about Linearizability}
\label{sec:ex-helping}

In this section, we use contextual reasoning to simplify linearizability proofs of concurrent data structure operations whose linearization points are future-dependent and potentially located in other threads.
Specifically, we will focus on linearizability proofs that use prophecy variables~\cite{DBLP:journals/pacmpl/JungLPRTDJ20,DBLP:journals/pacmpl/PatelKSW21}.
Such proofs involve \emph{helping protocols} that govern the transfer of linearizability obligations between threads.
The protocol is encoded using a \emph{registry} consisting of per-thread ghost state.
When a thread linearizes its own operation, it may at the same time linearize an unbounded number of operations executed by concurrent threads.
This leads to an unbounded ghost footprint at the linearization point due to the induced updates on the registry for all linearized operations.
We show that the registry can be moved to the context of the proof and that the invariant of the helping protocol can be obtained by approximate ghost updates on the context.

Concretely, we demonstrate the key ideas of the construction by revisiting the linearizability proofs for multicopy structures developed in~\cite{DBLP:journals/pacmpl/PatelKSW21}.
A multicopy structure is a concurrent data structure that implements a (total) map $M$ from keys $K$ to values $V$.
(We assume a dedicated \emph{tombstone} value $\tombstone \in V$ that indicates the absence of an entry for a key $k \in K$ in $M$.)
The data structure supports two types of operations: \code{search(k)} retrieves the value associated with $k$ in $M$ and \code{upsert(k,v)} updates the value of $k$ in $M$ to the new value $v$.
We represent the data structure state using an abstract predicate $\mcslstate(M)$. So the goal is to prove that the operations are linearizable subject to the following sequential specification:
\begin{align*}
  & \hoareof{\mcslstate(M)}{\mcode{search}(k)}{\mcode{$v$}.\; \mcslstate(M) * M(k)=v}\\
  & \hoareof{\mcslstate(M)}{\mcode{upsert}(k,v)}{\mcslstate(M[k \mapsto v])}
\end{align*}
Conceptually, a multicopy structure consists of an in-memory data structure and an on-disk data structure.
Upserts update the in-memory component, leaving the on-disk component unchanged. Concurrent maintenance operations periodically move entries from memory to disk.
Search operations first try to find an entry for $k$ in memory. Only if no entry is found in memory do they continue their search in the slower disk component.
The disk component can itself be organized as a linked structure consisting of log files, each of which may contain an old value for a key $k$.
The important observation is that the spatial organization of the pairs $(k,v)$ for a key $k$ in a multicopy structure is consistent with the temporal order in which these pairs have been upserted.
Consequently, we can use the history $h$ of key/value pairs that have been upserted thus far as an intermediate abstraction of the data structure's physical state. The key aspects of the linearizability proof can be carried out at this level of abstraction.

Formally, a history $h$ is given by
\[h \in \Hist \defeq (K \times V)^* \enspace.\]
We can then compute the map $M$ as a function of the current history $h \in \Hist$:
\begin{align*}
  M(\varepsilon) = \; & \lambda k.\; \tombstone\\
  M((k',v) \cdot h) = \; & \lambda k.\; \ite{k = k'}{v}{M(h)(k)}\enspace.
\end{align*}

We focus on the linearizability proof of search threads. A thread executing \code{search($k$)} may return a value $v$ that must have been either $M(h)(k)$ for the history $h$ when the search started, or some \code{upsert($k$,$v$)} operation linearized during the execution of the search. In the first case, the linearization point of \code{search($k$)} is right at the start of the operation. In the second case, its linearization point coincides with the linearization point of the interfering \code{upsert($k$,$v$)} thread. 

To enable thread modular reasoning about linearizability, the proof maintains a ghost state component in the shared state that consists of a registry $R$. The registry is a partial map from the thread IDs of all active search threads to their \emph{linearizability status}. These are values in the set:
\[\Status ::= \anobl{h,k,v} \mid \aful{h,k,v} \mid \aslt{h,k,v}\enspace.\]
Status $\anobl{h,k,v}$ indicates that (i) the thread started its search when the history was $h$, (ii) it is searching for key $k$, (iii) it will return value $v$, and (iv) it still has the obligation to linearize. Status $\aful{h,k,v}$ is similar but indicates that the thread has fulfilled its obligation to linearize. Finally, status $\aslt{h,k,v}$ only tracks properties (i) to (iii) without indicating whether the thread has linearized or not. Note that the choice of the return value $v$ is implemented using a prophecy variable. We elide these details here.

We endow $\Status$ with a separation algebra structure by letting $\aslt{h,k,v}$ be the unit of  $\aslt{h,k,v}$, $\anobl{h,k,v}$, and $\aful{h,k,v}$ for all $h$, $k$, and $v$, and leaving the multiplication undefined in all other cases. We lift this multiplication to registries $R$ pointwise in the canonical way.

The separation algebra used for the proof is then given by the product algebra:
\begin{align*}
\Sigma \defeq \; & \setcond{(h,R)}{\forall \tid \in \domof{R}. \; \Valid(h,R(\tid))}\\
\text{where} \quad
\Valid(h,s(h',k,v)) \defeq \; & s \neq \slt \Rightarrow h \geq h' \land (s=\obl \Leftrightarrow  \latest(h,k,v) < |h'|)\\
\latest(\varepsilon,k,v) \defeq \; & \ite{v=\tombstone}{0}{-1}\\
\latest((k',v') \cdot h,k,v) \defeq \; & \ite{(k=k' \land v=v')}{|h|+1}{\latest(h,k,v)}\enspace.
\end{align*}
Intuitively, $\latest(h,k,v)$ is the timestamp of the latest \code{upsert($k$,$v$)} in the history $h$ (where we assume that all keys are initialized to $\tombstone$ at time 0). The constraint $h \geq h'$ expresses that $h'$ is a suffix of $h$. As we shall see, the validity condition enforces that when an \code{upsert($k$,$v$)} linearizes, it must also linearize all pending \code{search($k$)} that will return $v$. We say that $(h,R)$ is valid if it is an element of $\Sigma$.

The multiplication on the product algebra is defined as $(h_1,R_1) \mstar (h_2,R_2) \defeq (h_1,R_1 \mstar R_2)$
if $h_1=h_2$ and $R_1 \mstardef R_2$, and undefined otherwise. Observe that multiplication preserves validity.

We write $\langle h \rangle$ for the predicate $\{(h,\emptyset)\}$ and $\tid \mapsto s(h',k,v)$ for the predicate
\[\setcond{(h,\{\tid \mapsto s(h',k,v)\})}{h \in \Hist \land \Valid(h,s(h',k,v))}\enspace.\]

There is only one kind of update of the physical state: upserting a new pair $(k,v)$ by appending it to the history $h$. Let $\acom$ be the command that performs the physical update. Following the setting in \cref{sec:contextualization:semantics}, we define the semantics $\upof{\acom}{\langle h \rangle \mstar\apredppp}$ in terms of the physical update and a ghost multiplication
\[\upof{\acom}{\langle h \rangle \mstar\apredppp}\ =\ \langle (k,v) \cdot h \rangle \imult\apredppp \enspace.\]
The ghost multiplication updates the registry to reestablish validity:
\begin{align*}
 ((k,v) \cdot h, R_1)  \imult \,(h, R_2) \; & = \\
(h, R_2)  \imult \,((k,v) \cdot h, R_1) \; & \defeq ((k,v) \cdot h, R_1 \uplus (\lambda \tid \in \domof{R_2}. \; \mathit{update}(k,v,R_2(\tid)))\\
 (h, R_1)  \imult \,(h, R_2) \; & \defeq (h, R_1 \uplus R_2)
\end{align*}
where $\mathit{update}(k,v,s(h',k',v')) \defeq \ite{(k=k' \land v=v' \land s=\obl)}{\aful{h',k,v}}{s(h',k',v')}$. In all remaining cases, $\apred \imult \apredp$ is undefined.

Observe that $\upof{\acom}{\langle h \rangle \mstar\apredppp}$ has an unbounded ghost footprint because it may update the registry entries for an unbounded number of threads. Further note that when the ghost update changes a thread's status from $\anobl{h',k,v}$ to $\aful{h',k,v}$, then we have $M((k,v) \cdot h)(k)=v$, so the sequential specification of \code{search($k$)} is satisfied at this point.

Another ghost update occurs when a new \code{search($k$)} thread is spawned. In this case, the ghost state is updated by adding a registry entry $\tid \mapsto s(h,k,v)$ for a fresh thread ID $\tid$, where $h$ is the current history and $v$ is the thread's prophesied return value. If $M(h)(k)=v$ then $s$ is chosen to be $\ful$ (the thread immediately linearizes) and otherwise $s=\obl$.

When reasoning about the actual code of the \code{search} and \code{update} operations, we can now move the entire registry into the context and carry out the proof by focusing only on the physical state. This is under the assumption that the context predicate $\apredpp$ used to describe the registry is preserved under all ghost updates.

To compute an appropriate context predicate, we define an approximate ghost multiplication $\ghostabs{\apred}$ that yields the upward closure under registering new search threads and extending the history with new key value pairs, starting from $\langle \varepsilon \rangle$:
\[\ghostabs{\langle a \rangle}(\apredp) = \exists h'.\, \langle h' \rangle \mstar \bigmstar_{\tid} \exists h,k,v. \; (\tid \mapsto \anobl{h,k,v} \lor \tid \mapsto \aful{h,k,v})\enspace.\]
The left-hand side is the desired context predicate $\apredpp$.


\newcommand{\clocsafe}[7]{\mathsf{safe}^{#3}_{#1, #2}(#4, #5, #6, #7)}
\newcommand{\clocsafedef}[4]{\clocsafe{\thePredicates}{\theInterference}{#1}{\acontext}{#2}{#3}{#4}}
\newcommand{\clocsafek}[3]{\clocsafedef{k}{#1}{#2}{#3}}
\newcommand{\ccfsafedef}[3]{\mathsf{cfsafe^{#1}(\acontext, #2, #3)}}
\newcommand{\ccfsafek}[2]{\ccfsafedef{k}{#1}{#2}}
\newcommand{\locsafe}[6]{\mathsf{safe}^{#3}_{#1, #2}(#4, #5, #6)}
\newcommand{\locsafedef}[4]{\locsafe{\thePredicates}{\theInterference}{#1}{#2}{#3}{#4}}
\newcommand{\locsafek}[3]{\locsafedef{k}{#1}{#2}{#3}}
\newcommand{\cfsafedef}[3]{\mathsf{cfsafe^{#1}(#2, #3)}}
\newcommand{\cfsafek}[2]{\cfsafedef{k}{#1}{#2}}

\newcommand{\acontextp}{\mathit{q}}

\section{Missing Details for \theLogicOG}
\label{app:og-casl}

So far we have presented context-aware reasoning in a sequential setting.
However, our proof principle also applies to the concurrent setting.
We present an Owicki-Gries-style extension for \theLogicSeq from \Cref{Section:CAReasoning}.

\smartparagraph{Concurrent semantics}

In the concurrent setting, we assume that the underlying separation algebra $(\setstates, \mstar, \emp)$ introduced in \Cref{sec:semantics} is, in fact, a product of two separation algebras $(\setshared, \sharedmult, \sharedemp)$ and $(\setlocal, \localmult, \localemp)$.
We require $\emp= \sharedemp \times \localemp \subseteq \setstates \subseteq \setshared \times \setlocal$.
In addition, $\setstates$ must be closed under decomposition: if $(\ashared_1\sharedmult\ashared_2, \alocal_1\localmult\alocal_2) \in \setstates$ then $(\ashared_1, \alocal_1) \in \setstates$, for all $\ashared_1,\ashared_2 \in \setshared$ and $\alocal_1,\alocal_2 \in \setlocal$.
For $(\ashared, \alocal) \in \setstates$ we call $\ashared$ the \emph{global state} and $\alocal$ the \emph{local state}.
States are composed component-wise, $(\ashared_1, \alocal_1)\statemult(\ashared_2, \alocal_2)=(\ashared_1\sharedmult\ashared_2, \alocal_1\localmult\alocal_2)$, and this composition is defined only if the product is again in~$\setstates$.

The semantics of concurrent programs is defined by a transition relation among configurations.
A \emph{configuration} is a pair  $\aconfig=(\ashared, \apc)$ consisting of a global state $\ashared\in \setshared$ and a program counter $\apc:\nat\to\setlocal\times\setstmt$.
The program counter assigns to every thread, modeled as a natural number,  the current local state and the statement to be executed next.
We use $\setconfig$ to denote the set of all configurations.
A configuration $(\ashared, \apc)$ is \emph{initial} for predicate~$\apred$ and library code~$\astmt$, if the program counter of every thread yields a local state $(\alocal, \astmt)$ where the code is the given one and the state satisfies $(\ashared, \alocal)\in \apred$.
The configuration is \emph{accepting} for predicate $\apredp$, if every terminated thread  $(\alocal, \cskip)$ satisfies the predicate, $(\ashared, \alocal)\in\apredp$.
We write these configuration predicates as the following sets
\begin{align*}
	\initset{\apred}{\astmt}
		\mkern-1mu&=\mkern-1mu
		\{(\ashared, \apc) \,{\mid}\,
			\forall i,\alocal,\mkern-1mu\widehat\astmt.\mkern+2mu
				\apc(i)\prall{=}(\alocal,\mkern-1mu \widehat\astmt) \prall{\Rightarrow}
				(\ashared,\mkern-1mu \alocal)\prall{\in}\apred \prall{\wedge} \widehat\astmt\prall{=}\astmt
			\}
	\\
	\acceptset{\apredp}
		\mkern-1mu&=\mkern-1mu
		\{(\ashared, \apc) \,{\mid}\,
			\forall i,\alocal.~
			\apc(i)=(\alocal, \cskip)\Rightarrow (\ashared, \alocal)\in\apredp \,
		\}
	\ .
\end{align*}

\begin{figure}
	\begin{mathpar}
		\vspace{-.5em}
		\inferrule{}{
			\pcStepOf{\acom}{\acom}{\cskip}
		}
		\and
		\inferrule{}{
			\pcStepOf{\seqof{\cskip}{\astmt}}{\cskip}{\astmt}
		}
		\and
		\inferrule{}{
			\pcStepOf{\loopof{\astmt}}{\cskip}{\choiceof{\cskip}{\seqof{\astmt}{\loopof{\astmt}}}}
		}
		\\
		\inferrule{\pcStepOf{\astmt_i}{\acom}{\astmt_i'}}{
			\pcStepOf{\choiceof{\astmt_1}{\astmt_2}}{\acom}{\astmt_i'}
		}
		\and
		\inferrule{
			\pcStepOf{\astmt_1}{\acom}{\astmt_1'}
		}{
			\pcStepOf{\seqof{\astmt_1}{\astmt_2}}{\acom}{\seqof{\astmt'_1}{\astmt_2}}
		}\and
		\inferrule{
			\pcStepOf{\astmt_1}{\acom}{\astmt_2}
			\\
			(\ashared_2, \alocal_2)\in\casemof{\acom}{\textcolor{blue}{\acontext}}{\ashared_1, \alocal_1}
		}{
			(\ashared_1, \apc[i\mapsto(\alocal_1, \astmt_1)])
			\rightarrow_{\textcolor{blue}{\acontext}} (\ashared_2, \apc[i\mapsto(\alocal_2, \astmt_2)])
		}
	\end{mathpar}
	\caption{%
		Transition relation $\progStepRel\;\subseteq\setconfig\times\textcolor{blue}{\powerset{\setstates}}\times\setconfig$ based on the control-flow relation $\rightarrow{} \subseteq \setstmt\times\setcom\times\setstmt$.
		\label{Figure:Relations}
	}
\end{figure}

The transition relation among configurations is standard, see \Cref{Figure:Relations}.
A command may change the global state and the local state of the executing thread.
It will not change the local state of other threads.
A computation of the program is a finite sequence of consecutive transitions.
A configuration is reachable if there is a computation that leads to it.
We write $\reachsetof{}{\aconfig}$ for the set of all configurations reachable from $\aconfig$.

Our goal is to show $\reachsetof{}{\initset{\apred}{\astmt}}\prall{\subseteq}\acceptset{\apredp}$ with an Owicki-Gries proof principle.
If this inclusion holds, we say that the concurrent Hoare triple $\hoareof{\apred}{\astmt}{\apredp}$ is valid.
To distinguish this from the sequential case, we write $\mmodels\hoareof{\apred}{\astmt}{\apredp}$.

\smartparagraph{The Owicki-Gries proof principle}
As common for Owicki-Gries approaches~\cite{DBLP:journals/acta/OwickiG76}, we reason in two steps.
First, we verify the program code as if it was run by an isolated thread.
The corresponding judgments $\thePredicates,\theInterference\semcalc\hoareOf{\apred}{\astmt}{\apredp}$ collect the predicates that were used during the proof in the set $\thePredicates$ and the interferences in the set~$\theInterference$~\cite[Section 7.3]{DBLP:conf/popl/Dinsdale-YoungBGPY13}.
Second, we check interference freedom to make sure the local proof still holds in the presence of other threads.

An \emph{interference} is a pair $(\acom, \apred)$ consisting of a command and a predicate.
It represents the fact that environment threads may execute command $\acom$ when the state belongs to $\apred$.
A state $(\ashared, \alocal)$ held by the isolated thread of interest will change under the interference to a state in
\[
	\semOf{(\acom, \apred)}{\ashared, \alocal}
	~\defeq~
	\setcond{(\ashared', \alocal)}{
		\exists \alocal_1, \alocal_2.\;\; (\ashared, \alocal_1)\in \apred ~\wedge~ (\ashared', \alocal_2)\in \semof{\acom}{\ashared, \alocal_1}}
	\enspace.
\]
We consider every state $(\ashared, \alocal_1)\in \apred$ that agrees with $(\ashared, \alocal)$ on the global component, compute the post, and combine the resulting global component with the local component $\alocal$.

The thread-local proof computes a set of interferences.
We consider interference sets up to the operation of joining predicates for the same command, $\setcompact{(\acom, \apred)}\cup \setcompact{(\acom, \apredp)}=\setcompact{(\acom, \apred\cup\apredp)}$.
Then $(\acom, \apred)\subseteq \theInterference$ means there is an interference $(\acom, \apredp)\in\theInterference$ with $\apred\subseteq \apredp$.
We write $\theInterference\mstar\apredp$ for the set of interferences $(\acom, \apred\mstar\apredp)$ with $(\acom, \apred)\in\theInterference$, and similarly for the set of predicates $\thePredicates$.

The \emph{interference-freedom check} takes as input a set of interferences $\theInterference$ and a set of predicates $\thePredicates$.
It checks that no interference can invalidate a predicate, $\semOf{(\acom, \apredp)}{\apred}\subseteq \apred$ for all $(\acom, \apredp)\in\theInterference$ and all $\apred\in\thePredicates$.
If this is the case, we write $\isInterferenceFreeOf[\theInterference]{\thePredicates}$ and say that $\thePredicates$ is interference-free wrt. $\theInterference$.

\smartparagraph{An Owicki-Gries program logic}
We reason about the validity of concurrent Hoare triples with the program logic from \Cref{fig:og-rules} (ignore the \makeColorLogic{blue} parts for now).
We write $\thePredicates, \theInterference\semcalc\hoareof{\apred}{\astmt}{\apredp}$ if the corresponding judgement can be derived using the rules from \Cref{fig:og-rules}.
This Owicki-Gries program logic is sound \cite{DBLP:journals/pacmpl/MeyerWW22}.

\begin{theorem}
	$\thePredicates, \theInterference\semcalc\hoareof{\apred}{\astmt}{\apredp}$
	and
	${}\isInterferenceFreeOf[\theInterference]{\thePredicates}$
	and
	$\apred\in\thePredicates$
	imply
	${}\mmodels\hoareof{\apred}{\astmt}{\apredp}$.
\end{theorem}
\begin{proof}
	See proof of \Cref{thm:soundness-OG}.
\end{proof}

\begin{figure*}
	\begin{mathpar}
		\inferH{og-com}{
			\casemof{\acom}{\highlight{\acontext}}{\apred}\subseteq \apredp
		}{
			\setcompact{\highlight{\acontext\mstar{}}\apredp}, \setcompact{(\acom, \highlight{\acontext\mstar{}}\apred)}\semCalc\choareHighOf{\highlight{\acontext}}{\apred}{\acom}{\apredp}
		}
		\and
		\inferH{og-cons}{
			\thePredicates', \theInterference' \semCalc\choareHighOf{\highlight{\acontext}}{\apred'}{\astmt}{\apredp'}
			\\\\
			\apred\subseteq \apred'
			\\
			\apredp'\subseteq \apredp
			\\
			\thePredicates'\subseteq \thePredicates
			\\
			\theInterference'\subseteq\theInterference
		}{
			\thePredicates, \theInterference\semCalc\choareHighOf{\highlight{\acontext}}{\apred}{\astmt}{\apredp}
		}
		\and
		\inferH{og-loop}{
			\thePredicates, \theInterference\semCalc\choareHighOf{\highlight{\acontext}}{\apred}{\astmt}{\apred}
		}{
			\setcompact{\highlight{\acontext\mstar{}}\apred}\cup\thePredicates, \theInterference\semCalc\choareHighOf{\highlight{\acontext}}{\apred}{\loopof{\astmt}}{\apred}
		}
		\and
		\inferH{og-seq}{
			\thePredicates_1, \theInterference_1\semCalc\choareHighOf{\highlight{\acontext}}{\apred}{\astmt_1}{\apredp}
			\\\\
			\thePredicates_2, \theInterference_2\semCalc\choareHighOf{\highlight{\acontext}}{\apredp}{\astmt_2}{\apredpp}
		}{
			\setcompact{\highlight{\acontext\mstar{}}\apredp}\cup\thePredicates_1\cup\thePredicates_2, \theInterference_1\cup\theInterference_2\semCalc\choareHighOf{\highlight{\acontext}}{\apred}{\astmt_1;\astmt_2}{\apredpp}
		}
		\and
		\inferH{og-choice}{
			\thePredicates_1, \theInterference_1\semCalc\choareHighOf{\highlight{\acontext}}{\apred}{\astmt_1}{\apredp}
			\\\\
			\thePredicates_2, \theInterference_2\semCalc\choareHighOf{\highlight{\acontext}}{\apred}{\astmt_2}{\apredp}
		}{
			\thePredicates_1\cup\thePredicates_2, \theInterference_1\cup\theInterference_2\semCalc\choareHighOf{\highlight{\acontext}}{\apred}{\choiceof{\astmt_1}{\astmt_2}}{\apredp}
		}
		\and
		\inferH{og-frame}{
			\thePredicates, \theInterference\semCalc\choareHighOf{\highlight{\acontext}}{\apred}{\astmt}{\apredp}
		}{
			\thePredicates\mstar\apredppp, \theInterference\mstar\apredppp\semCalc\choareHighOf{\highlight{\acontext}}{\apred\mstar \apredppp}{\astmt}{\apredp\mstar\apredppp}
		}
		\and
		\highlight{
		\inferH{og-context}{
			\thePredicates, \theInterference\semCalc\choareHighOf{\acontext\mstar\apredppp}{\apred}{\astmt}{\apredp}
		}{
			\thePredicates, \theInterference \semCalc\choareof{\acontext}{\apred\mstar \apredppp}{\astmt}{\apredp\mstar\apredppp}
		}
		}
		\and
		\highlight{
		\inferH{og-widen}{
			\thePredicates, \theInterference \semCalc\choareof{\acontext}{\apred\mstar \apredppp}{\astmt}{\apredp\mstar\apredppp}
		}{
			\thePredicates, \theInterference\semCalc\choareHighOf{\acontext\mstar\apredppp}{\apred}{\astmt}{\apredp}
		}
		}
	\end{mathpar}
	\caption{Proof rules for concurrent context-aware separation logic (\theLogicOG).\label{fig:og-rules}}
\end{figure*}

\smartparagraph{A context-aware extension.}
Similar to \Cref{Section:CAReasoning}, the context-aware version of concurrent Hoare triples takes the form $\choareof{\acontext}{\apred}{\astmt}{\apredp}$ with the understanding that $\acontext$ is meant to be framed to the pre- and postcondition.
That is, validity $\mmodels\choareof{\acontext}{\apred}{\astmt}{\apredp}$ holds iff $\mmodels\hoareof{\apred\mstar\acontext}{\astmt}{\apredp\mstar\acontext}$.
The extended program logic is given in \Cref{fig:og-rules} (including the \makeColorLogic{blue} parts) and involves statements of the form $\thePredicates, \theInterference\semcalc\choareof{\acontext}{\apred}{\astmt}{\apredp}$.
This extension, called \theLogicOG, is sound.

\begin{theorem}
	$\thePredicates, \theInterference\semcalc\choareof{\acontext}{\apred}{\astmt}{\apredp}$
	and
	${}\isInterferenceFreeOf[\theInterference]{\thePredicates}$
	and
	$\apred\in\thePredicates$
	imply
	${}\mmodels\choareof{\acontext}{\apred}{\astmt}{\apredp}$.
\end{theorem}
\begin{proof}
	See proof of \Cref{thm:soundness-OGCASL}.
\end{proof}

Note that we did not adapt the interference-freedom check for our extension.
It still applies the standard semantics $\sem{\acom}$.
Our results from \Cref{Section:CAReasoning} apply to this check as well in the sense that one can use contextualize some predicate $\acontext$ and then apply the induced semantics $\icasem{\acom}{\acontext}$ instead of $\sem{\acom}$.
We consider this an implementation detail of how to perform the interference-freedom check and elide the straightforward technicalities that would be required to record the context in which an interference is recorded in order to apply it for the check.


\section{A BST with Rotations}
\label{app:bst-rot}
\label{app:bst-rotate}

We present a BST with rotations and prove it correct using contextualization.

\subsection{Specification}
\label{app:bst-rot:specification}

To specify the operations of our BST implementation, we define the predicate $\bst{\abscontent}$ denoting a binary search tree with logical contents $\abscontent$. (We assume a global root pointer that is left implicit.)
The logical contents are a subset of the keys, a totally ordered set $\keyspace$ that has minimal and maximal elements $-\infty$ and $\infty$, respectively.
With this understanding, an implementation is a binary search tree if its \code{contains}, \code{insert}, and \code{delete} operations adhere to the following specification:
\begin{align*}
  \annot{\bst{\abscontent} \mstar -\infty\neq\key\neq\infty}~~\mcode{contains($\key$)}&~~\annot{\res.~\bst{\abscontent} \mstar \res\Leftrightarrow\key\in\abscontent}
  \\
  \annot{\bst{\abscontent} \mstar -\infty\neq\key\neq\infty}~~\mcode{~~insert($\key$)}&~~\annot{\res.~\bst{\abscontent\cup\set{\key}} \mstar \res\Leftrightarrow\key\notin\abscontent}
  \\
  \annot{\bst{\abscontent} \mstar -\infty\neq\key\neq\infty}~~\mcode{~~delete($\key$)}&~~\annot{\res.~\bst{\abscontent\setminus\set{\key}} \mstar \res\Leftrightarrow\key\in\abscontent}\ .
\end{align*}
The specification of \code{contains} requires that the logical contents $\abscontent$ of the tree remain unchanged.
Moreover, the Boolean return value $\res$ must indicate whether or not the given $\key$ is contained in $\abscontent$.
Operation \code{insert} adds the given $\key$ to the contents of the tree.
Its return value indicates whether the $\key$ was successfully inserted ($\res=\true$) or if it was already present ($\res=\false$).
Similarly, \code{delete} removes the given $\key$ from the tree and indicates whether an actual deletion took place.
As is usual, all operations expect $-\infty\neq\key\neq\infty$; values $-\infty,\infty$ are for internal purposes.

Besides the above user-facing API, virtually all (concurrent) binary trees come with maintenance operations that restructure the tree.
Restructuring operations do not alter the logical contents but may rotate or remove nodes in order to speed up future accesses.
We assume a single \code{maintenance} operation which performs the desired restructuring periodically.
Its specification is as expected:
\begin{align*}
  \annot{\bst{\abscontent}}~~\mcode{maintenance()}&~~\annot{\bst{\abscontent}}\ .
\end{align*}
Interestingly, concurrent tree implementations tend to have fairly simple implementations for \code{insert} and \code{delete} but much more intricate \code{maintenance} operations.
Our implementation will mimic this: \code{delete} will simply mark nodes as logically deleted but does not attempt to remove them from the tree, the removal is performed later by the \code{maintenance} operation.

\subsection{Implementation}
\label{app:bst-rot:implementation}

Our sequential BST implementation is given in \cref{app:fig:api-impl} (ignore the \textcolor{colorAnnotation}{annotations} for a moment).
The nodes of the tree are of type \code{Node}.
They contain \code{left} and \code{right} pointers for their left and right subtrees, respectively, a \code{key} from $\keyspace$ that they represent, and a \code{del} flag indicating whether the node is logically deleted.
We say that a node is marked if the \code{del} flag is raised, and unmarked otherwise.
Additionally, \code{Node}s contain a ghost field \code{dup} that we use in our proofs to distinguish temporary duplicates that are inserted by rotations.
That it is a ghost field means that the implementation does not react on its value, only the ghost state may.
The shared variable $\Root$ is the entry point to the tree.
It is an unmarked sentinel node containing key $-\infty$.


\begin{figure}
  \begin{minipage}[b]{.49\textwidth}
\begin{lstlisting}[language=SPL]
struct Node {  Node* left, right;
               bool del; $\keyspace$ key;
               `ghost dup;' };
\end{lstlisting}
\begin{lstlisting}[language=SPL]
var Root = new Node { key = -$\infty$ };
\end{lstlisting}
\begin{lstlisting}[language=SPL]
$\color{colorAnnotation}\addtolength{\jot}{-2.5pt}
  \loopinv{\abscontent,\setnodes,\x,\y} ~\defeq~ \begin{aligned}[t]
    &\inv{\abscontent,\setnodes}  \\
    &\hspace{-2.4cm}{}\mstar{} \x\in\setnodes \mstar \keyof{\x}\neq\key\in\isof{\x} \\
    &\hspace{-2.4cm}{}\mstar{} \key<\keyof{\x}\Rightarrow\y=\leftof{\x} \\
    &\hspace{-2.4cm}{}\mstar{} \key>\keyof{\x}\Rightarrow\y=\rightof{\x}
\end{aligned}$
\end{lstlisting}
\begin{lstlisting}[language=SPL]
$\annot{
  \inv{\abscontent,\setnodes} \MSTAR -\infty<\key
}$
def ^find^($\keyspace$ $\key$) {
  $\x$ = $\Root$; $\y$ = $\Root$.right;
  $\annot{ \loopinv{\abscontent,\setnodes,\x,\y} }$
  while ($\y$ != null && $\y$.key != $\key$) { @\label{app:code:find:loop-begin}@
    $\annot{
      \inv{\abscontent,\setnodes} \MSTAR \x,\y\prall{\in}\setnodes \MSTAR \keyof{\y}\prall{\neq}\key\prall{\in}\isof{\y}
    }$
    $\x$ = $\y$;
    $\y$ = $\x$.key < $\key$ ? $\x$.left : $\x$.right;
    $\annot{ \loopinv{\abscontent,\setnodes,\x,\y} }$
  } @\label{app:code:find:loop-end}@
  return $\x$, $\y$;
}
$\annot{\x,\y.~
  \loopinv{\abscontent,\setnodes,\x,\y}
  \MSTAR 
  \y\prall{\in}\setnodes\Rightarrow\keyof{\y}\prall{=}\key
}$
\end{lstlisting}
\begin{lstlisting}[language=SPL,belowskip=0pt]
$\annot{
  \inv{\abscontent,\setnodes} \MSTAR -\infty<\key<\infty
}$
def ^delete^($\keyspace$ $\key$) {
  $\x$, $\y$ = find($\key$);
  if ($\y$ == null || $\y$.del) {
    $\annot{ \inv{\abscontent,\setnodes} \MSTAR \key\notin\abscontent }$
    return false; // not present @\label{app:code:delete:false}@
  } else {
    $\annot{ \inv{\abscontent,\setnodes} \,\mstar\, \y\prall{\in}\setnodes \,\mstar\, \keyof{\y}\prall{=}\key\prall{\in}\abscontent\prall{\cap}\ksof{\y} }$
    $\y$.del = true; // mark as deleted @\label{app:code:delete:mark}@
    return true; @\label{app:code:delete:true}@
} }
$\annot{
  \res.~\inv{\abscontent\setminus\set{\key},\setnodes} \MSTAR \res\Leftrightarrow\key\in\abscontent
}$
\end{lstlisting}
  \end{minipage}
  \hfill
  \begin{minipage}[b]{.49\textwidth}
\begin{lstlisting}[language=SPL]
$\annot{
  \inv{\abscontent,\setnodes} \MSTAR -\infty<\key
}$
def ^insert^($\keyspace$ $\key$) {
  $\x$, $\y$ = find($\key$);
  if ($\y$ == null) {
    $\annot{ \loopinv{\abscontent,\setnodes,\x,\y} \MSTAR \key\notin\abscontent \MSTAR \y=\nullptr }$
    $\z$ = new Node($\key$); // insert new node
    if ($\key$ < $\x$.key) $\x$.left = $\z$; @\label{app:code:insert:less}@
    else $\x$.right = $\z$; @\label{app:code:insert:greater}@
    $\annot{ \inv{\abscontent\uplus\set{\key},\setnodes} \MSTAR \key\notin\abscontent }$
    return true;
  } else if ($\y$.del) {
    $\annot{ \inv{\abscontent,\setnodes} \MSTAR \key\notin\abscontent }$
    $\y$.del = false; // unmark @\label{app:code:insert:unmark}@
    return true;
  } else {
    $\annot{ \inv{\abscontent,\setnodes} \MSTAR \key\in\abscontent }$
    return false; // already present @\label{app:code:insert:false}@
} }
$\annot{
  \inv{\abscontent,\setnodes} \MSTAR -\infty<\key
}$
\end{lstlisting}
\begin{lstlisting}[language=SPL]
$\annot{
  \inv{\abscontent,\setnodes} \MSTAR -\infty<\key
}$
def ^contains^($\keyspace$ $\key$) {
  _, $\y$ = find($\key$);
  $\annotml{ \inv{\abscontent,\setnodes} \,\mstar\, \left(\begin{aligned}
    &\y\prall{\in}\setnodes \,\land\, \neg\delof{\y} \\\Leftrightarrow {} &\keyof{\y}=\key\in\abscontent\cap\ksof{\y} \end{aligned}\right) }$
  return $\y$ != null && !$\y$.del;
}
$\annot{
  \inv{\abscontent,\setnodes} \MSTAR -\infty<\key
}$
\end{lstlisting}
\begin{lstlisting}[language=SPL,belowskip=0pt]
$\annot{
  \inv{\abscontent,\setnodes} \MSTAR -\infty<\key
}$
def ^maintenance^() {
  if (*) removeSimple(); @~@// see @\textcolor{colorCodeComment}{\cref*{app:fig:simple-removal}}@
  if (*) removeComplex(); // see @\textcolor{colorCodeComment}{\cref*{app:fig:complex-removal}}@
  if (*) rotate(); @~~~~~~~@// see @\textcolor{colorCodeComment}{\cref*{app:fig:rotate}}@
}
$\annot{
  \inv{\abscontent,\setnodes} \MSTAR -\infty<\key
}$
\end{lstlisting}
  \end{minipage}
  \caption{%
    Implementation and proof outline of a sequential binary search tree (BST). 
    The \code{delete} operation marks nodes as logically deleted and defers the physical removal to the \code{maintenance} operation.
    This two-step deletion mimics concurrent BST implementations.%
    \label{app:fig:api-impl}%
  }
\end{figure}

All operations of the tree rely on the helper \code{find}.
It takes a $\key\in\keyspace$ and searches it in a standard BST fashion: upon reaching a node $\anode$, the search terminates at $\anode$ if $\anode$ is $\pnull$ or if $\key$ equals $\selof{\anode}{key}$, continues to $\selof{\anode}{left}$ if $\key$ is less than $\selof{\anode}{key}$, and continues to $\selof{\anode}{right}$ if $\key$ is greater than $\selof{\anode}{key}$.
The helper then returns the last two nodes $\anode$, $\anodep$ on the search path, such that $\anode$ is guaranteed to be non-null and $\anodep$ is the potentially-null left or right child of $\anode$.
If $\anodep$ is non-null, then it is guarantee to contain $\key$.
Node $\anode$ never contains $\key$ (otherwise the search path would not extend to $\anodep$).

Operation \code{delete($\key$)} uses \code{find} to obtain nodes $\x$ and $\y$ with the above properties.
If $\y$ is null or marked, then $\key$ is not logically contained in the tree (because $\x$ does not contain $\key$) and $\false$ is returned, \cref{app:code:delete:false}.
Otherwise, $\y$ is unmarked.
In this case, \cref{app:code:delete:mark} marks it to purge it from the logical contents of the tree and \cref{app:code:delete:true} subsequently returns $\true$.
Note that $\y$ remains physically present in the tree.
The physical removal is deferred to the \code{maintenance} operation.

Operation \code{insert($\key$)} proceeds similarly.
It uses \code{find} to obtain nodes $\x$ and $\y$.
If $\y$ is null, a new node containing $\key$ is created and added as a child of $\x$, \cref{app:code:insert:less,app:code:insert:greater}.
If $\y$ is non-null, it is guaranteed to contain $\key$.
If it is marked, it is simply unmarked, \cref{app:code:insert:unmark}.
In both cases, $\key$ is successfully added to the contents of the tree and $\true$ is returned.
Otherwise, $\y$ is unmarked, i.e., $\key$ is already present in the tree.
Hence, the ongoing insertion fails and returns $\false$, \cref{app:code:insert:false}.

Operation \code{contains($\key$)} simply returns whether the node $\y$ returned by \code{find} is non-null and unmarked.
The arguments for \code{delete}/\code{insert} already cover why this is an appropriate result.

The \code{maintenance} operation non-deterministically invokes the helper functions \code{removeSimple}, \code{removeComplex}, and \code{rotate}.
Operations \code{removeSimple} and \code{removeComplex} perform the physical removal (unlinking) of marked nodes.
We discuss them in detail in \cref{app:bst-rot:remove-simple,app:bst-rot:remove-complex}, respectively.
Operation \code{rotate} performs standard right rotations.
We discuss it in detail in \cref{app:bst-rot:rotation}.
We ignore left rotations, they are symmetrical.

\subsection{Proof Methodology}
\label{app:bst-rot:invariant}

Towards verifying our implementation against the BST specification, we first develop the predicate $\bst{\abscontent}$.
We define it by \[ \bst{\abscontent} \defeq \exists\setnodes.~\inv{\abscontent,\setnodes} \ .\]
Predicate $\invraw$ is the structural invariant of our implementation that ties its physical state to the logical contents $\abscontent$ of the specification.
We use the flow framework to devise $\invraw$.

\smartparagraph{Flow Domain}
As the underlying flow monoid, we choose
\[
  \bigl( \powerset{\keyspace}\uplus\set{\bot,\top},\, \oplus,\, \bot \bigr)
  \qquad\text{with}\qquad
  \amonval\oplus\amonvalp ~\defeq~ \begin{cases}
    \amonval &\text{if }~ \amonvalp=\bot \\
    \amonvalp &\text{if }~ \amonval=\bot \\
    \top &\text{otherwise}
  \ .
  \end{cases}
\]
The flow values propagated by this flow are sets of keys $\amonval,\amonvalp\subseteq\keyspace$, or dedicated sentinel values $\bot,\top$.
Intuitively, if $\akey\in\amonval$ for the flow $\amonval$ of a node $\anode$, then \code{find} will traverse $\anode$ when searching for $\akey$, as alluded to in \cref{app:bst-rot:implementation}.
Value $\bot$ denotes that a node is unreachable from $\Root$.
Note that $\bot$ and $\emptyset$ differ: $\emptyset$ means that the node is still reachable from $\Root$, but \code{find} will not traverse it.
Value $\top$ denotes that a node has multiple reachable parents, that is, the heap graph is not a tree.
To establish this intuition, $\bot$ is neutral with respect to $\oplus$ and in all other cases $\oplus$ yields $\top$.

\smartparagraph{Physical State}
The physical state is comprised of a collection of nodes (of type \code{Node}).
We capture the resources associated with such nodes $\anode$ in a predicate $\nodeof{\anode}$.
To simplify the notation, assume that proofs are existentially closed.
This way, we can use the naming convention $f(\anode)$ to refer to the logical variable that holds the value of field $\selof{\anode}{f}$.
Then define: \[
  \nodeof{\anode} ~\defeq~ \begin{aligned}[t]
    &\selof{\anode}{left}\pto\leftof{\anode} \MSTAR
    \selof{\anode}{right}\pto\rightof{\anode} \MSTAR
    \selof{\anode}{key}\pto\keyof{\anode}  \\\!\!\!\MSTAR
    &\selof{\anode}{del}\pto\delof{\anode} \MSTAR
    \selof{\anode}{in}\pto\inof{\anode} \MSTAR
    \selof{\anode}{dup}\pto\dupof{\anode}\enspace.
  \end{aligned}
\]
Fields $\leftof{\anode}$, $\rightof{\anode}$, $\keyof{\anode}$, and $\delof{\anode}$ are as expected, they capture the left child, right child, key, and deletion flag of node $\anode$, respectively.
They give rise to the physical contents of node $\anode$:
\[
  \ctnof{\anode} ~\defeq~ \begin{cases}
    \set{\keyof{\anode}} &\text{if } \neg\delof{\anode} \lor \anode\neq\Root \\
    \emptyset &\text{otherwise}
  \ .
  \end{cases}
\]
Field $\inof{\anode}$ is the inflow of $\anode$.
Most of the time we are interested in the \emph{inset} of $\anode$, which we define as \[\isof{\anode} ~\defeq~ \bigoplus_{\anodep \in \setnodes} \inof{\anode}(\anodep) \ .\]
Field $\dupof{\anode}$ is the ghost field used for rotations.
Its possible values are $\dupvalnone$ (not a duplicate), $\dupvalleft$ (duplicate due to a left rotation), and $\dupvalright$ (duplicate due to a right rotation).

\smartparagraph{Edge Functions}
We derive edge functions from the physical representation of nodes.
That is, the edge functions $\edges$ of every flow graph are induced by the physical state of its nodes.
Intuitively, the edge functions filter the incoming flow values (search paths) according to the BST search principle from \cref{app:bst-rot:implementation}.
Formally, we define:
\begin{align*}
  \edgesatof{\anode}{\anodep}{\amonval} ~=~ \begin{cases}
    \top  &\text{if }~ \leftof{\anode}=\anodep=\rightof{\anode} \\
    \amonval\cap[-\infty,\keyof{\anode})  &\text{if }~ \anodep=\leftof{\anode} ~\land~ \dupof{\anode}\neq\dupvalleft \\
    \amonval\cap(\keyof{\anode},\infty]  &\text{if }~ \anodep=\rightof{\anode} ~\land~ \dupof{\anode}\neq\dupvalright \\
    \bot  &\text{otherwise}
  \end{cases}
\end{align*}
where we use $\bot\cap\amonval=\bot$ and $\top\cap\amonval=\top$.
The first case handles edges from a node $\anode$ to its left child $\anodep$.
The edge forwards the portion of the given flow value $\amonval$ that is smaller than $\anode$'s key.
The side condition $\dupof{\anode}\neq\dupvalleft$ prevents the edge function from forwarding flow if $\anode$ is a duplicate inserted by a left rotation. This is needed because $\anode$'s left child already receives flow from the node that $\anode$ duplicates.
Similarly, the second case forwards the portion of $\amonval$ that is larger than $\anode$'s key to its right child, provided $\anode$ is not a duplicate during a right rotation.
In all other cases, the edge function produces $\bot$.

We refer to the \emph{left outset} of a node $\anode$ as the quantity produced by the edge function $\edgesat{\anode}{\leftof{\anode}}$ for the inflow of $\anode$.
Formally, this is $\oslof{\anode} = \edgesatof{\anode}{\leftof{\anode}}{\isof{\anode}}$ if $\leftof{\anode}\neq\nullptr$ and $\oslof{\anode}=\emptyset$ otherwise.
The \emph{right outset} $\osrof{\anode}$ is defined correspondingly.
Subtracting $\anode$'s outsets from its inset yields the keys for which \code{find} terminates in $\anode$.
That is, these are the keys that could be in $\anode$ while still satisfying the BST order property for the remaining graph.
We refer to this quantity as the \emph{keyset} of $\anode$: \[
  \ksof{\anode} ~=~ \begin{cases}
    \emptyset  &\text{if }~ \isof{\anode}\in\set{\bot,\top} \\
    \isof{\anode} \setminus \bigl( \oslof{\anode} \cup \osrof{\anode} \bigr)  &\text{otherwise}
  \ .
  \end{cases}
\]

\smartparagraph{Invariant}
The structural invariant of our BST implementation is given by a predicate $\inv{\abscontent,\setnodesp,\setnodes}$, which denotes a subregion $\setnodesp\subseteq\setnodes$ of the entire structure $\setnodes$ with logical contents $\abscontent$.
The invariant carries the resources for the nodes in $\setnodesp$ and describes their properties:
\begin{align*}
  \inv{\abscontent, \setnodesp, \setnodes} ~=~~&
    \Root\in\setnodes \MSTAR \nullptr\notin\setnodes \MSTAR
    \setnodesp\subseteq\setnodes \MSTAR \abscontent=\ctnof{\setnodesp} \MSTAR \bigmstar_{\anode\in\setnodesp} \nodeof{\anode} \mstar \inv{\anode,\setnodes}
  \\
  \inv{\anode,\setnodes} ~=~~&
    \set{\leftof{\anode},\rightof{\anode}}\subseteq\setnodes\uplus\set{\nullptr} ~~\land~~
    \dupof{\anode}=\dupvalnone ~~\land~~
    \\&
    \isof{\anode} \neq \top ~~\land~~
    \ctnof{\anode}\subseteq\ksof{\anode} ~~\land~~
    \bigl(\isof{\anode}\neq\bot \implies \keyof{\anode}\in\isof{\anode}\bigr) ~~\land~~
    \\&\bigl(\anode = \Root \implies \isof{\anode} = [-\infty,\infty] \land \neg\delof{\anode} \land \keyof{\anode}=-\infty\bigr)
\end{align*}
The invariant has two main ingredients.
First, it ties the expected logical contents $\abscontent$ to the physical contents $\ctnof{\setnodesp}$ of the region $\setnodesp$, as desired.
Second, it carries the resources $\nodeof{\anode}$ for all nodes $\anode\in\setnodesp$ from the region and specifies their properties using the node-local invariant $\inv{\anode,\setnodes}$.
The node-$\anode$-local invariant requires that
\begin{inparaenum}
  \item the overall structure $\setnodes$ is self-contained, i.e., $\anode$ cannot reach nodes outside of $\setnodes$,
  \item $\anode$ is not a duplicate, i.e., duplicates are inserted only temporarily during rotation,
  \item $\anode$ has at most one parent node that is reachable from $\Root$ (up to temporary duplicates), i.e., at most one node sends flow to $\anode$,
  \item the physical contents of $\anode$ are contained in its keyset, i.e., the search paths for $\anode$'s contents reach and terminate in $\anode$,
  \item if $\anode$ has inflow, then it receives at least its own key, and
  \item $\Root$ is a sentinel node with key $-\infty$, it is never marked as deleted, and it is the entry point for all searches.
\end{inparaenum}
For brevity, we write $\inv{\abscontent,\setnodes}$ to mean $\inv{\abscontent,\setnodes,\setnodes}$.

\smartparagraph{Compositionality}
When framing or contextualizing a subregion of $\setnodes$, we employ the following compositionality of the invariant:
\begin{align*}
  \inv{\abscontent,\setnodesp_1\uplus\setnodesp_2,\setnodes}
  ~~\iff~~
  \exists\abscontent_1,\abscontent_2.~
  \inv{\abscontent_1,\setnodesp_1,\setnodes} \mstar
  \inv{\abscontent_2,\setnodesp_2,\setnodes} \mstar
  \abscontent=\abscontent_1\cup\abscontent_2\enspace.
  \tag{\textsc{comp}}
  \label{app:weak-decomposition}
\end{align*}
That is, we decompose the invariant into two disjoint regions and recompose them.
However, the decomposition does not localize the reasoning because it does not localize the logical contents.
The invariant alone does not guarantee that the contents $\abscontent_1$ and $\abscontent_2$ are disjoint, despite $\setnodesp_1$ and $\setnodesp_2$ being disjoint.
To overcome this, we strengthen the decomposition by requiring $\setnodes = \setnodesp_1 \cup \setnodesp_2$:
\begin{align*}
  \begin{aligned}
    \inv{\abscontent,\setnodes,\setnodes} \mstar
    \setnodes=\setnodesp_1\uplus\setnodesp_2
    ~~\implies~~
    \exists\abscontent_1,\abscontent_2.
    \begin{aligned}[t]
      &\inv{\abscontent_1,\setnodesp_1,\setnodes} \mstar
      \inv{\abscontent_2,\setnodesp_2,\setnodes} \mstar
      \\{}\mstar{}
      &\ksof{\setnodesp_1}\cap\ksof{\setnodesp_2}=\emptyset \mstar
      \abscontent=\abscontent_1\uplus\abscontent_2\enspace.
    \end{aligned}
  \end{aligned}
  \tag{\textsc{decomp}}
  \label{app:strong-decomposition}
\end{align*}
The implication only allows one to decompose the entire tree, but in return establishes that the keysets of the decomposed regions are disjoint.
From this, we conclude that $\abscontent_1$ and $\abscontent_2$ are disjoint, because the invariant guarantees $\abscontent_1\subseteq\ksof{\setnodesp_1}$ and $\abscontent_2\subseteq\ksof{\setnodesp_2}$.
That the keysets of $\setnodesp_1$ and $\setnodesp_2$ are disjoint follows from a result due to \citet{DBLP:journals/tods/ShashaG88}.
Translated to flows, the result requires that
\begin{inparaenum}
  \item all edge functions are decreasing, i.e, $\edgesatof{\anode}{\anodep}{\amonval}\leq\amonval$,
  \item the left and right outsets are disjoint, $\oslof{\anode}\cap\osrof{\anode}=\emptyset$, and
  \item only the root node receives inflow from outside the full graph.
\end{inparaenum}
These properties are ensured by $\invraw$ and the definition of the edge functions. However, the result only applies if $\invraw$ is satisfied by the entire graph $\setnodes$.
Hence, we apply the decomposition only if we are splitting the full graph.

\subsection{Verifying the Implementation}

We now show that the implementation from \cref{app:fig:api-impl} implements a BST along \cref{app:bst-rot:specification}.
Proof outlines for \code{find} and \code{delete} are given in form of \textcolor{colorAnnotation}{colored annotations} in \cref{app:fig:api-impl}.

The proof for \code{find($\key$)} follows our intuition from \cref{app:bst-rot:implementation}: the traversal, \crefrange{app:code:find:loop-begin}{app:code:find:loop-end}, \emph{goes with the flow} to locate the search key, $\key$.
It maintains the invariant that it is still on right track, $\key\in\isof{\x}$, $\x$ does not contain the search key, $\key\neq\keyof{\x}$, and that $\y$ is the next node on the search path, $\y=\leftof{\x}$ or $\y=\rightof{\x}$ if $\key<\keyof{\x}$ or $\keyof{\x}<\key$, respectively.
The traversal terminates at the end of the search path, if $\y$ is null ($\y\notin\setnodes$) or if $\key=\keyof{\y}$.

The proof for \code{delete} relies on the above properties that \code{find} establishes for the nodes $\x$ and $\y$ it returns.
There are three cases.
\begin{asparaenum}[(1)]
  \item
    If $\y$ is null, then we have $\key\in\ksof{\x}$ because $\key\in\isof{\x}$ and there is no outset to null.
    Moreover, $\key\notin\ctnof{\x}$ because $\key\neq\keyof{\x}$.
    We now use the keyset disjointness of the strong decomposition \eqref{app:strong-decomposition} to conclude that $\x$ is the only node that has $\key$ in its keyset, $\key\notin\ksof{\setnodes\setminus\set{\x}}$.
    Recall that $\ctnof{\setnodes\setminus\set{\x}}\subseteq\ksof{\setnodes\setminus\set{\x}}$ holds by the invariant.
    Hence, $\key$ is not contained in the tree, $\key\notin\abscontent$.
    This justifies returning $\false$ on \cref{app:code:delete:false}.
  \item
    If $\y$ is non-null and marked, we know that $\key$ flows from $\x$ to $\y$, $\key\in\isof{\y}$.
    That the search terminates in $\y$ means $\key=\keyof{\y}$.
    Together, $\key\in\ksof{\y}$.
    That $\y$ is marked means $\ctnof{\y}=\emptyset$.
    With a similar argument as before, we conclude $\key\notin\abscontent$.
    Returning $\false$ on \cref{app:code:delete:false} is again correct.
  \item
    If $\y$ is non-null and unmarked, the reasoning is similar.
    We have $\key\in\ksof{\y}\cap\ctnof{\y}$ and $\key\notin\ctnof{\setnodes\setminus\set{\y}}$.
    From \eqref{app:strong-decomposition} we get that $\key$ is stored exclusively in $\y$, $\abscontent=\ctnof{\y}\uplus\ctnof{\setnodes\setminus\set{\y}}$.
    Hence, marking $\y$ on \cref{app:code:delete:mark} effectively removes $\key$ from $\abscontent$.
    That is, after \cref{app:code:delete:mark} the state of the tree is $\inv{\abscontent\setminus\set{\key},\setnodes,\setnodes}$.
    This justifies returning $\true$ on \cref{app:code:delete:true}.
    Note that the update does not change the flow, so its physical and ghost footprint is just $\y$.
\end{asparaenum}

The proofs for \code{insert} and \code{contains} follow a similar pattern.
In the remainder of this section, we focus on \code{removeSimple}, \code{removeComplex}, and \code{rotate}.

\subsection{Simple Removal}
\label{app:bst-rot:remove-simple}

The maintenance operation \code{removeSimple} unlinks nodes from the tree that have been marked as deleted.
It is a ``simple'' removal because it unlinks nodes only if they have at most one child.
We expect \code{removeSimple} to leave unchanged the logical contents of the tree.
Concretely, we establish the following specification for it, as imposed by the specification of \code{maintenance}: \[
  \annot{
    \inv{\abscontent,\setnodes}
  }
  ~~
  \mcode{removeSimple()}
  ~~
  \annot{
    \inv{\abscontent,\setnodes}
  }
  \ .
\]


\begin{figure}
  \begin{subfigure}[b]{.56\textwidth}
\begin{tikzpicture}[overlay,remember picture]
  \begin{scope}
    \coordinate[yshift=1.7ex] (ctx) at (pic cs:app-ctx-simple);
    \coordinate[yshift=-1.75ex,xshift=1ex] (end) at (pic cs:app-ctx-simple-end) {};
    \gettikzxy{(ctx)}{\links}{\oben}
    \gettikzxy{(end)}{\rechts}{\unter}
    \draw[context] (\links,\oben) rectangle (\rechts,\unter-.75pt);
    \draw[footprint] (\links+1ex,\oben-2.5ex) rectangle (\rechts-1ex,\unter);
  \end{scope}
\end{tikzpicture}
\begin{lstlisting}[language=SPL,belowskip=0pt]
$\annot{
  \inv{\abscontent,\setnodes}
}$
def ^removeSimple^() {
  $\x$, _ = find(*);
  $\annot{
    \inv{\abscontent,\setnodes} \MSTAR \x\in\setnodes \MSTAR \isof{\x}\neq\bot
  }$
  $\y$ = $\x$.left;
  if ($\y$ == null) return;
  if (!$\y$.del) return;
  $\annotml{
    &\inv{\abscontent,\setnodes} \MSTAR \x,\y\in\setnodes \MSTAR \delof{\y}\\
    &\!\mstar{} \isof{\x}\neq\bot \MSTAR \leftof{\x}=\y
  }$
  if ($\y$.right == null) {
    // ... @\label{app:code:simple-remove:symmetric}@
  } elif ($\y$.left == null) {
    $\annotml{
      &\inv{\abscontent,\setnodes} \MSTAR \x,\y\in\setnodes \MSTAR \isof{\x}\neq\bot \MSTAR \leftof{\x}=\y \\
      &\!\mstar{} \delof{\y} \MSTAR \leftof{\y}=\nullptr\neq\rightof{\y}
    }$ @\label{app:code:simple-remove:pre}@
    @\tikzmark{app-ctx-simple}@  $\ctx{
        \inv{\abscontent_1,\setnodes{\setminus}\setc{\x,\y},\setnodes} \MSTAR \abscontent=\abscontent_1\cup\abscontent_2
      }$  @\tikzmark{app-ctx-simple-mid}@ @\label{app:code:simple-remove:context}@
      $\annotml{
        &\inv{\abscontent_2,\setc{\x,\y},\setnodes} \MSTAR \isof{\x}\neq\bot \MSTAR \delof{\y} \\
        &\!\mstar{} \leftof{\x}=\y \MSTAR \leftof{\y}=\nullptr\neq\rightof{\y}
      }$ @\label{app:code:simple-remove:pre-footprint}@
      $\x$.left = $\y$.right; @\label{app:code:simple-remove:unlink}@
      $\annotml{
        &\inv{\abscontent_2,\setc{\x,\y},\setnodes} \MSTAR \isof{\x}\neq\bot=\isof{\y} \MSTAR \delof{\y}\\
        &\!\mstar{} \leftof{\x}=\rightof{\y} \MSTAR \leftof{\y}=\nullptr\neq\rightof{\y}
      }$  @\tikzmark{app-ctx-simple-end}@ @\label{app:code:simple-remove:post-footprint}@
    $\annot{
      \inv{\abscontent,\setnodes}
    }$ @\label{app:code:simple-remove:post}@
  }
}
$\annot{
  \inv{\abscontent,\setnodes}
}$
\end{lstlisting}
  \caption{%
    Implementation and proof outline for simple removals, highlighting the \colorbg{backgroundContext}{\textcolor{colorContext}{context}} and \colorbg{backgroundFootprint}{\textcolor{colorAnnotation}{footprint}} of the update unlinking the marked node $\y$ on \cref{app:code:simple-remove:unlink}.
    \label{app:fig:simple-removal:impl}%
    \label{app:fig:simple-removal:proof}%
  }
  \end{subfigure}
  \hfill
  \begin{subfigure}[b]{.40\textwidth}
  \begin{subfigure}{\textwidth}
  \center
\begin{tikzpicture}[level/.style={sibling distance = 3cm/#1, level distance = 1.1cm}]
  \draw [context] (-2.6,.35) rectangle (2.0,-3.6);
  \node [anchor=south east] at (1.95,-3.55) {Context};
  \draw [footprint] (.5,.25) rectangle (-2.5,-2.65);
  \node [anchor=north west] at (-2.45,.2) {Footprint};

  \node (x) [treenode] {$\x$}
    child [treeptr] {
      node (y) [treenode] {$\y$}
        child [treeptr,-|] {
          node (N) {}
        }
        child [treeptr] {
          node (A) [subtree] {$A$}
        }
    }
    child [treeptr] {
      node (B) [subtree] {$B$}
    }
  ;
  \path (x) -- (y.north) node[midway,xshift=-2mm,flow] {$\ais\cap[-\infty,\x)$};
  \path (x) -- (B.north) node[midway,xshift=2mm,flow] {$\ais\cap(\x,\infty]$};
  \path (y) -- (N.north) node[pos=.51,flow] {$\bot$};
  \path (y) -- (A.north) node[midway,xshift=2mm,flowch] {$\ais\cap(\y,\x)$};
  \node (i) [flow,anchor=north west] at (-.75,.85) {$\ais\neq\bot$};
  \path (i.220) edge[inflow,bend right=20] (x);
\end{tikzpicture}
  \caption{%
    Tree structure on \cref{app:code:simple-remove:pre-footprint}, prior to unlinking $\y$.
    We write $\x$/$\y$ instead of $\keyof{\x}$/$\keyof{\y}$.%
    \label{app:fig:simple-removal:prestate}%
  }
  \end{subfigure}\\[3mm]
  \begin{subfigure}{\textwidth}
  \center
\begin{tikzpicture}[level/.style={sibling distance = 3cm/#1, level distance = 1.1cm}]
  \draw [context] (-2.6,.35) rectangle (2.0,-3.6);
  \node [anchor=south east] at (1.95,-3.55) {Context};
  \draw [footprint] (.5,.25) rectangle (-2.5,-2.65);
  \node [anchor=north west] at (-2.45,.2) {Footprint};

  \node (x) [treenode] {$\x$}
    child [treeptr] {
      node (y) [treenode] {$\y$} edge from parent[draw=none]
        child [treeptr,-|] {
          node (N) {}
        }
        child [treeptr] {
          node (A) [subtree] {$A$}
        }
    }
    child [treeptr] {
      node (B) [subtree] {$B$}
    }
  ;
  \draw[treeptr] (x) -- (A.north);
  \path (x) -- (A.north) node[pos=.7975,xshift=6mm,flowch] {$\ais\cap[-\infty,\x)$};
  \path (x) -- (B.north) node[midway,flow] {$\ais\cap(\x,\infty]$};
  \path (y) -- (N.north) node[pos=.51,flow] {$\bot$};
  \path (y) -- (A.north) node[pos=.4425,flowch] {$\bot$};
  \node (i) [flow,anchor=north west] at (-.75,.85) {$\ais\neq\bot$};
  \path (i.220) edge[inflow,bend right=20] (x);
\end{tikzpicture}
  \caption{%
    Tree structure after unlinking $\y$, \cref{app:code:simple-remove:post-footprint}.%
    \label{app:fig:simple-removal:poststate}%
  }
  \end{subfigure}
  \end{subfigure}
  \caption{%
    A \emph{simple} removal unlinks internal, marked nodes if they have at most one child.
    Unlinking marked nodes does not alter the logical contents of the tree.
    Removing marked nodes with two children is done by \emph{complex} removals.%
    \label{app:fig:simple-removal}%
  }
\end{figure}

The implementation and proof outline of \code{removeSimple} are given in \cref{app:fig:simple-removal:impl}.
It proceeds as follows.
Starting from some arbitrary node $\x$ currently linked into the tree, the left child $\y$ of $\x$ is read out.
If $\y$ is a non-null marked node with at most one child, \code{removeSimple} tries to unlink $\y$.
In the case where $\y$ has no left child, $\y$ is unlinked by making $\y$'s right child the left child of $\x$, i.e. by updating $\selof{\x}{left}$ to $\rightof{\y}$.
The state of the tree before and after the unlinking is depicted in \cref{app:fig:simple-removal:prestate,app:fig:simple-removal:poststate} (ignore the annotations for a moment).
The case where $\y$ has no right child is symmetric.
Similarly, removing the marked right child of $\x$, \cref{app:code:simple-remove:symmetric}, is symmetric.
We omit the symmetric cases.

The interesting part of the proof is the moment where $\y$ is unlinked.
The precondition of the unlinking is on \cref{app:code:simple-remove:pre}.
It states that the tree currently represents the set $\abscontent$ and satisfies the invariant, $\inv{\abscontent,\setnodes}$.
Moreover, the precondition captures our above intuition, stating that $\y$ is the marked left child of $\x$, $\delof{\y} \mstar \leftof{\x}=\y$, and has no left child itself, $\leftof{\y}=\nullptr$.
The proof goal is to establish $\inv{\abscontent,\setnodes}$ on \cref{app:code:simple-remove:post}, that is, show the unlinking of $\y$ on \cref{app:code:simple-remove:unlink} maintains both the logical contents and the invariant.
The main challenge with this update is its unbounded footprint.
To see this, consider \cref{app:fig:simple-removal:prestate}.
It depicts the part of the tree rooted in $\x$ prior to the update.
Because $\x$ is linked into the tree, it has non-$\bot$ inflow, say $\ais\defeq\isof{\x}$.
Node $\y$ receives the portion of $\ais$ that is smaller than $\x$' key, $\isof{\y}=\oslof{\x}=\ais\cap[-\infty,\keyof{x})$.
From that inset, $\y$ forwards the portion that is larger than its key to is right child, subtree $A$ in the \cref{app:fig:simple-removal:prestate}.
Overall, this means that the inset of $A$ before unlinking $\y$ is $\ais\cap(\keyof{\y},\keyof{\x})$.
After unlinking $\y$, i.e., making subtree $A$ the left child of $\x$ as depicted in \cref{app:fig:simple-removal:poststate}, $A$ receives all the flow $\ais\cap[-\infty,\keyof{x})$ that $\y$ used to receive, not just the $\keyof{\y}$-larger portion of it.
That is, the unlinking results in subtree $A$ receiving the additional flow $\aks\defeq\ais\cap[-\infty,\keyof{\y}]$.
To see why this makes the footprint of unlinking $\y$ unbounded, observe that the nodes in $A$ all have keys larger than $\y$.
That is, the additional flow $\aks$ is forwarded to the left-most leaf in $A$.
Because $A$ can be an arbitrary tree, the footprint is unbounded.

To overcome the unbounded footprint, we contextualize subtree $A$ (and the right child of $\x$ which is irrelevant here).
To be precise, we decompose the tree described by $\inv{\abscontent,\setnodes}$ into the footprint $\inv{\abscontent_2,\set{\x,\y},\setnodes}$, consisting of nodes $\x$ and $\y$, and the context $\inv{\abscontent_1,\setnodes\setminus\set{\x,\y},\setnodes}$, consisting of the remainder of the tree, with $\abscontent=\abscontent_1\cup\abscontent_2$.
In \cref{app:fig:simple-removal:impl}, the application of rule \ruleref{context} is between \cref{app:code:simple-remove:context,app:code:simple-remove:post-footprint}: \cref{app:code:simple-remove:context} states the \colorbg{backgroundContext}{context} and the following lines reason about the \colorbg{backgroundFootprint}{footprint}.

Because the update changes the inflow of the context $\inv{\abscontent_1,\setnodes\setminus\set{\x,\y},\setnodes}$, we have to show that it tolerates the additional inflow it receives after the update.
Towards this, we capture the change in inflow using the estimator $\simplerel$ defined by: \[
  \amonval\simplerel\amonvalp
  ~~\defifff~~
  \amonval=\amonvalp
  \,~\lor~
  \bigl(\,\setc{\amonval,\amonvalp}\cap\setc{\bot,\top}=\emptyset
  ~\land~
  \amonval\subseteq\amonvalp\,\bigr)
  \ .
\]
This reflects our intuition that flow values may increase, but it prevents previously unreachable nodes (those that have a flow of $\bot$) to receive flow and it also prevents nodes from receiving flow from more sources than before (the flow increase does not result in a flow of $\top$).
The relation satisfies the requirements for an estimator.
Note that the actual update produces a flow increase of at most $\aks$ in the context, however, our proof does not rely on this upper bound.
It is readily checked that the context $\inv{\abscontent_1,\setnodes\setminus\set{\x,\y},\setnodes}$ is $\simplerel$-closed, because the keyset of a node only increases, unreachable nodes (flow of $\bot$) remain unreachable, and no node's inflow becomes $\top$.
That $\simplerel$ satisfies the requirements of an estimator is left as an exercise to the reader.

We turn to the actual update on \cref{app:code:simple-remove:unlink} within the footprint $\inv{\abscontent_2,\set{\x,\y},\setnodes}$.
The physical update unlinking $\y$ is standard.
Moreover, it is easy to see that the invariant for $\x$ is maintained.
To see that the invariant for $\y$ is maintained as well, observe that $\delof{\y}$ on \cref{app:code:simple-remove:pre-footprint} means $\y\neq\Root$ and $\ctnof{\y}=\emptyset$.
This is the annotation on \cref{app:code:simple-remove:post-footprint}, and it is $\simplerel$-closed following the same arguments as for the context.
Lastly, it remains to show that the footprint's outflow after the update is $\simplerel$-larger than prior to the update.
We already discussed that $\y$'s outflow to its right child is $\ais\cap(\keyof{\y},\keyof{\x})$ before the update and $\ais\cap[-\infty,\keyof{x})$ after it.
Because $\ais=\isof{\x}\neq\bot$, we have the desired $\ais\cap(\keyof{\y},\keyof{\x}) \simplerel \ais\cap[-\infty,\keyof{x})$.
(The outflow at $\x$'s right child does not change and $\simplerel$ is reflexive).
This concludes the proof of \code{removeSimple} along the lines of the proof strategy from \cref{sec:instantiation}.

\subsection{Complex Removal}
\label{app:bst-rot:remove-complex}

Operation \code{removeComplex} unlinks marked nodes from the tree that have two children and are thus not handled by \code{removeSimple}.
The specification is as expected: \[
  \annot{
    \inv{\abscontent,\setnodes}
  }
  ~~
  \mcode{removeComplex()}
  ~~
  \annot{
    \inv{\abscontent,\setnodes}
  }
  \ .
\]


\begin{figure}
  \begin{subfigure}[b]{.6\textwidth}
\begin{tikzpicture}[overlay,remember picture]
  \begin{scope}
    \coordinate[yshift=3ex] (ctx) at (pic cs:app-ctx-complex-move);
    \coordinate[yshift=-2ex,xshift=6ex] (end) at (pic cs:app-ctx-complex-move-end) {};
    \gettikzxy{(ctx)}{\links}{\oben}
    \gettikzxy{(end)}{\rechts}{\unter}
    \draw[context] (\links,\oben) rectangle (\rechts,\unter-.75pt);
    \draw[footprint] (\links+1ex,\oben-5ex) rectangle (\rechts-1ex,\unter);
  \end{scope}
  \begin{scope}
    \coordinate[yshift=2ex] (ctx) at (pic cs:app-ctx-complex-unlink);
    \coordinate[yshift=-.6ex] (end) at (pic cs:app-ctx-complex-unlink-end) {};
    \coordinate[yshift=-3ex,xshift=6ex] (width) at (pic cs:app-ctx-complex-move-end) {};
    \gettikzxy{(ctx)}{\links}{\oben}
    \gettikzxy{(end)}{\irgendwas}{\unter}
    \gettikzxy{(width)}{\rechts}{\irgendwer}
    \draw[context] (\links,\oben) rectangle (\rechts,\unter-.75pt);
    \draw[footprint] (\links+1ex,\oben-2.8ex) rectangle (\rechts-1ex,\unter);
  \end{scope}
\end{tikzpicture}
\begin{lstlisting}[language=SPL]
$\color{colorAnnotation}\ctxupclosed{\setnodesp,\aks} \mkern+24.4mu~\defeq~ \forall\z\in\setnodesp.~\isof{\z}\neq\bot\implies\keyof{\z}\notin\aks$
$\color{colorAnnotation}\invp{\abscontent,\z,\setnodes} ~\defeq~ \nodeof{\z} \mstar \abscontent=\ctnof{\z} \mstar \setc{\leftof{\z},\rightof{\z}}\subseteq\setnodes\uplus\setc{\nullptr}$
\end{lstlisting}
\begin{lstlisting}[language=SPL,belowskip=0pt]
$\annot{
  \inv{\abscontent,\setnodes}
}$
def ^removeComplex^() {
  $\x$, _ = find(*);
  $\annot{
    \inv{\abscontent,\setnodes} \MSTAR \x\in\setnodes \MSTAR \isof{\x}\neq\bot
  }$
  if (!$\x$->del) return;
  if ($\x$->left == NULL) return;
  if ($\x$->right == NULL) return;
  $\annotml{
    &\inv{\abscontent,\setnodes} \MSTAR \x\in\setnodes \MSTAR \isof{\x}\neq\bot
    \\&\!\mstar{} \delof{\x} \MSTAR \leftof{\x}\neq\nullptr\neq\rightof{\x}
  }$
  $\p$, $\y$ = findSucc($\x$);
  $\annotml{
    &\inv{\abscontent,\setnodes} \MSTAR \x,\y,\p\in\setnodes \MSTAR \isof{\x}\neq\bot\neq\isof{\p} \MSTAR \delof{\x} \\
    &\!\mstar{} \leftof{\x}\neq\nullptr\neq\rightof{\x} \MSTAR \leftof{\p}=\y \MSTAR \leftof{\y}=\nullptr \\
    &\!\mstar{} \keyof{\y}\prall{<}\keyof{\p} \MSTAR \isof{\y}\subseteq\osrof{\x}
  }$
  let $\akey_\y=\keyof{\y}$, $\aks=\isof{\x}\cap(\keyof{\x},\akey_\y]$;
  $\annotml{
    &\inv{\abscontent,\setnodes} \MSTAR \x,\y,\p\in\setnodes \MSTAR \isof{\x}\neq\bot\neq\isof{\p} \MSTAR \delof{\x} \\
    &\!\mstar{} \leftof{\x}\neq\nullptr\neq\rightof{\x} \MSTAR \leftof{\p}=\y \MSTAR \leftof{\y}=\nullptr \\
    &\!\mstar{} \keyof{\y}\prall{<}\keyof{\p} \MSTAR \keyof{\x},\akey_\y\in\isof{\x} \MSTAR \aks\subseteq\ksof{\y}
  }$ @\label{app:code:complex-remove:pre}@
  @\tikzmark{app-ctx-complex-move}@  $\ctxml{
      &\inv{\abscontent_1,\setnodes{\setminus}\setc{\x,\y},\setnodes} \MSTAR \abscontent=\abscontent_1\cup\abscontent_2  \MSTAR \ctxupclosed{\setnodes{\setminus}\setc{\x,\y},\aks} \\
      &\!\mstar{} \p\in\setnodes \MSTAR \isof{\p}\neq\bot \MSTAR \akey_\y\prall{<}\keyof{\p} \MSTAR \leftof{\p}=\y
    }$ @\label{app:code:complex-remove:ctx}@
    $\annotml{
      &\inv{\abscontent_2,\setc{\x,\y},\setnodes} \MSTAR \keyof{\x},\akey_\y\in\isof{\x} \MSTAR \leftof{\y}=\nullptr \\
      &\!\mstar{} \delof{\x} \MSTAR \leftof{\x}\neq\nullptr\neq\rightof{\x} \MSTAR \keyof{\y}=\akey_\y
    }$ @\label{app:code:complex-remove:footprint-pre}@
    $x$->key = $\y$->key; @\label{app:code:complex-remove:move}@
    $\annotml{
      &\inv{\emptyset,\setc{\x},\setnodes} \MSTAR \invp{\abscontent_2,\y,\setnodes} \MSTAR \akey_\y\in\isof{\x} \\
      &\!\mstar{} \delof{\x} \MSTAR \leftof{\y}=\nullptr \MSTAR \keyof{\x}=\keyof{\y}=\akey_\y
    }$ @\label{app:code:complex-remove:footprint-post-move}@
    $\x$->del = $\y$->del;@\label{app:code:complex-remove:del-x}@ $\y$->del = true;@\label{app:code:complex-remove:del-y}@
    $\annotml{
      &\inv{\abscontent_2,\setc{\x},\setnodes} \MSTAR \invp{\emptyset,\y,\setnodes} \MSTAR \akey_\y\in\isof{\x} \\
      &\!\mstar{} \delof{\x} \MSTAR \leftof{\y}=\nullptr \MSTAR \keyof{\x}=\keyof{\y}=\akey_\y
    }$  @\tikzmark{app-ctx-complex-move-end}@ @\label{app:code:complex-remove:footprint-post}@
  $\annotml{
      &\inv{\abscontent,\setnodes{\setminus}\setc{\y},\setnodes} \MSTAR \invp{\emptyset,\y,\setnodes} \MSTAR \isof{\p}\neq\bot \MSTAR \delof{\y} \\
      &\!\mstar{} \x,\p\in\setnodes \MSTAR \leftof{\p}=\y \MSTAR \keyof{\y}\prall{<}\keyof{\p} \MSTAR \leftof{\y}=\nullptr
  }$ @\label{app:code:complex-remove:intermediate}@
  @\tikzmark{app-ctx-complex-unlink}@  $\ctx{
      \inv{\abscontent_1,\setnodes{\setminus}\setc{\y,\p},\setnodes} \MSTAR \abscontent=\abscontent_1\cup\abscontent_2 \MSTAR \x\in\setnodes
    }$
    $\annotml{
      &\inv{\abscontent_2,\setc{\p},\setnodes} \MSTAR \invp{\emptyset,\setc{\y},\setnodes} \MSTAR \keyof{\y}\prall{<}\keyof{\p} \\
      &\!\mstar{} \isof{\p}\neq\bot \MSTAR \leftof{\p}=\y \MSTAR \delof{\y} \MSTAR \leftof{\y}=\nullptr
    }$
    $\p$->left = $\y$->right; @\label{app:code:complex-remove:unlink}@
    $\annot{
      \inv{\abscontent_2,\setc{\p},\setnodes} \mstar \invp{\emptyset,\setc{\y},\setnodes} \mstar \isof{\y}=\bot \mstar \delof{\y}
    }$
    $\annot{
      \inv{\abscontent_2,\setc{\y,\p},\setnodes}
    }$  @\tikzmark{app-ctx-complex-unlink-end}@
  $\annot{
    \inv{\abscontent,\setnodes} \MSTAR \x,\y,\p\in\setnodes
  }$ @\label{app:code:complex-remove:post}@
}
$\annot{
  \inv{\abscontent,\setnodes}
}$
\end{lstlisting}
  \caption{%
    Implementation and proof outline for complex removals, highlighting the \colorbg{backgroundContext}{\textcolor{colorContext}{context}} and \colorbg{backgroundFootprint}{\textcolor{colorAnnotation}{footprint}} of the restructuring updates.
    \label{app:fig:complex-removal:proof}%
  }
  \end{subfigure}
  \hfill
  \begin{subfigure}[b]{.35\textwidth}
    \begin{subfigure}{\textwidth}
\begin{tikzpicture}[level/.style={sibling distance = 2.5cm/#1, level distance = 1.25cm}]
  \draw [context] (-2,.35) rectangle (2.7,-5.75);
  \node [anchor=south west] at (-1.95,-5.7) {Context};
  \path[footprint] (-.4,.25) -- (.4,.25) -- (1.1,-4) -- (-.4,-4);
  \node [anchor=north west,rotate=90] at (-.35,-3.95) {Footprint};

  \node (x) [treenode] {$\x$}
    child [treeptr] {
      node (A) [subtree] {$A$}
    }
    child [treeptr] {
      node (B) [subtree,] {$B$}
        child {
          node (y) [treenode] {$\y$} edge from parent[draw=none]
            child [treeptr,-|] {
              node (N) {}
            }
            child [treeptr] {
              node (C) [subtree] {$C$}
            }
        }
        child { node {} edge from parent[draw=none] }
    }
  ;
  \path (x) -- (A.north) node[midway,xshift=-2mm,flowch] {$\ais\cap[-\infty,\x)$};
  \path (x) -- (B.north) node[midway,xshift=2mm,flowch] {$\ais\cap(\x,\infty]$};
  \path[inflow] (B.213) -- (y) node[midway,flow] {$\aisp\neq\bot$};
  \path (y) -- (N.north) node[midway,flow] {$\bot$};
  \path (y) -- (C.north) node[midway,flow,yshift=.35pt,xshift=6mm] {$\aisp\cap(\y,\infty]$};
  \node (i) [flow,anchor=north west] at (-.75,.85) {$\ais\neq\bot$};
  \path (i.220) edge[inflow,bend right=20] (x);
\end{tikzpicture}
  \caption{%
    Tree structure on \cref{app:code:complex-remove:move}, prior to copying $\y$'s key to $\x$.
    We write $\x$/$\y$ instead of $\keyof{\x}$/$\keyof{\y}$.
    \label{app:fig:complex-removal:prestate}%
  }
  \end{subfigure}\\[3mm]
  \begin{subfigure}{\textwidth}
\begin{tikzpicture}[level/.style={sibling distance = 2.5cm/#1, level distance = 1.25cm}]
  \draw [context] (-2,.35) rectangle (2.7,-5.75);
  \node [anchor=south west] at (-1.95,-5.7) {Context};
  \path[footprint] (-.4,.25) -- (.4,.25) -- (1.1,-4) -- (-.4,-4);
  \node [anchor=north west,rotate=90] at (-.35,-3.95) {Footprint};

  \node (x) [treenode] {$\y$}
    child [treeptr] {
      node (A) [subtree] {$A$}
    }
    child [treeptr] {
      node (B) [subtree,] {$B$}
        child {
          node (y) [treenode] {$\y$} edge from parent[draw=none]
            child [treeptr,-|] {
              node (N) {}
            }
            child [treeptr] {
              node (C) [subtree] {$C$}
            }
        }
        child { node {} edge from parent[draw=none] }
    }
  ;
  \path (x) -- (A.north) node[midway,xshift=-2mm,flowch] {$\ais\cap[-\infty,\y)$};
  \path (x) -- (B.north) node[midway,xshift=2mm,flowch] {$\ais\cap(\y,\infty]$};
  \path[inflow] (B.213) -- (y) node[midway,flow] {$\aisp\neq\bot$};
  \path (y) -- (N.north) node[midway,flow] {$\bot$};
  \path (y) -- (C.north) node[midway,flow,yshift=.35pt,xshift=6mm] {$\aisp\cap(\y,\infty]$};
  \node (i) [flow,anchor=north west] at (-.75,.85) {$\ais\neq\bot$};
  \path (i.220) edge[inflow,bend right=20] (x);
\end{tikzpicture}
  \caption{%
    Tree structure after copying $\y$'s key to $\x$, \cref{app:code:complex-remove:move}.
    \label{app:fig:complex-removal:poststate}%
  }
  \end{subfigure}
  \end{subfigure}
  \caption{%
    A \emph{complex} removal unlinks internal, marked nodes with two children.
    The operation does not alter the logical contents of the tree.
    \label{app:fig:complex-removal}%
  }
\end{figure}

\Cref{app:fig:complex-removal} gives the implementation and proof outline.
There are four steps in \code{removeComplex}.
First, it obtains an arbitrary marked node $\x$ that is reachable from $\Root$ and has two children.
Second, it uses the helper function \code{findSucc} from \cref{app:fig:find-succ} to obtain the left-most leaf $\y$ and its parent $\p$ in the subtree $B$ rooted at $\x$'s right child.\footnote{%
  To avoid distracting case distinctions, \code{findSucc} assumes $\p\neq\x$.
  We omit the case $\p=\x$ because it is much simpler: its flow update is not unbounded, it affect only the nodes $\x$ and $\y$.
}
That $\y$ is the left-most leaf in $B$ means that $\keyof{\y}$ is the next larger key after $\keyof{\x}$ in $B$.
Consequently, all search paths for keys from $(\keyof{\x},\keyof{\y})$ that reach $\x$ continue to $B$ and eventually reach $\y$.
Moreover, they terminate in $\y$ because $\y$ has no left child.
Formally, these search paths are for keys $\aks=\isof{\x}\cap(\keyof{\x},\keyof{\y})$ and they are part of $\y$'s keyset, $\aks\subseteq\ksof{\y}$.
This is the annotation on \cref{app:code:complex-remove:pre}, depicted in \cref{app:fig:complex-removal:prestate}.
Third, the contents of $\x$ and $\y$ are swapped, \crefrange{code:complex-remove:move}{code:complex-remove:del-y}.
This is the most interesting part of the proof and we discuss it in detail below.
Last, $\y$ is unlinked.
The procedure, involving $\y$ and its parent $\p$, is the same as the unlinking in \code{removeSimple}.
We will not reiterate it.

A detailed discussion of the third step, swapping the contents of $\x$ and $\y$, is in order.
To avoid confusion between the values of fields before and after the following updates, we record the values $\akey_\x=\keyof{\x}$, $\akey_\y=\keyof{\y}$, and $\ais=\isof{\x}$ from before the update, as on \cref{app:code:complex-remove:pre}.
Note that $\akey_\x<\akey_\y$.
Now, \cref{app:code:complex-remove:move} copies $\akey_\y$ into $\x$.
This is challenging due to its intricate flow update, which is visualized in \cref{app:fig:complex-removal:prestate,fig:complex-removal:poststate}.
The update
increases the inflow of $\x$'s left subtree $A$ from $\ais\cap[-\infty,\akey_\x)$ to $\ais\cap[-\infty,\akey_\y)$.
The inflow of $\x$'s right subtree $B$, in turn, decreases from $\ais\cap(\akey_\x,\infty]$ to $\ais\cap(\akey_\y,\infty]$.
That is, the portion $\aks=\ais\cap(\akey_\x,\akey_\y]$ is redirected from $B$ to $A$.
We capture this change of inflow with the estimator $\complexrel$ defined by: \[
  \amonval\complexrel\amonvalp
  ~~\defifff~~
    \amonval\simplerel\amonvalp
    \,~\lor~
    \bigl(\,
      \setc{\amonval,\amonvalp}\cap\setc{\bot,\top}=\emptyset ~\land~ \akey_\x\notin\amonval ~\land~ \amonval\setminus\aks\subseteq\amonvalp
    \,\bigr)
  \ .
\]
The relation allows the inflow to increase arbitrarily.
Moreover, it allows the inflow to decrease by up to $\aks$.
However, decreasing the inflow may only occur for inflows that do not contain $\akey_\x$.
The side condition localizes the decrease to the subtrees of $\x$, because $\akey_\x\in\isof{\x}$.
As we will see, it prevents us from \emph{loosing} $\akey_\y$ from the inflow of $\x$ when performing the update.

For handling the update, we choose nodes $\x$ and $\y$ as the footprint and contextualize everything else.
Technically, the context is $\inv{\abscontent_1,\setnodes\setminus\set{\x,\y},\setnodes}$ and the footprint is $\inv{\abscontent_2,\set{\x,\y},\setnodes}$ with $\abscontent=\abscontent_1\cup\abscontent_2$.
We now show that the context is $\complexrel$-closed, which is significantly more involved than showing $\simplerel$-closedness for \code{removeSimple}.
Consider a contextualized node $\z\in\setnodes\setminus\set{\x,\y}$ with $\isof{\y}\neq\bot$.
The invariant states $\keyof{\z}\in\isof{\z}$.
To preserve this inclusion, despite reducing the inflow by up to $\aks$, requires $\keyof{\z}\notin\aks$.
To see this, observe that $\keyof{\z}\in\ksof{\z}$ follows from $\keyof{\z}\in\isof{\z}$ prior to the update.
Using \eqref{app:strong-decomposition} for the above context-footprint decomposition, yields $\ksof{\z}\cap\ksof{\y}=\emptyset$.
Because we already argued for $\aks\subseteq\ksof{\y}$, we get $\keyof{\z}\notin\aks$ as desired.
This line of reasoning also implies that the $\ctnof{\z}\subseteq\ksof{\z}$ part of the invariant is preserved.
Note that $\Root$ does not loose flow due to $\akey_\x\in\isof{\Root}$.
Overall, we conclude that the context $\inv{\abscontent_1,\setnodes\setminus\set{\x,\y},\setnodes}$ on \cref{app:code:complex-remove:ctx} is indeed $\complexrel$-closed.

We turn to the footprint $\inv{\abscontent_2,\set{\x,\y},\setnodes}$, \cref{app:code:complex-remove:footprint-pre}.
The physical update changing $\selof{\x}{key}$ from $\akey_\x$ to $\akey_\y$ is as expected.
It remains to discuss how it affects the flow of $\x$ and $\y$.
\Cref{thm:contextualization} prescribes that the footprint be $\complexrel$-closed after the update.
Because $\akey_\x,\akey_\y\in\ais$ prior to the update, $\complexrel$ guarantees that $\akey_\y$ remains in the inset of $\x$ after the update.
That is, we have $\inv{\emptyset,\set{\x},\setnodes}$ after the update, because $\x$ is still marked.
The inflow of $\y$, in turn, may decrease by up to $\aks$.
This is expected because $\y$'s inflow is solely due to $\x$'s outflow.
Note that relation $\complexrel$ does not remove \emph{exactly} $\aks$---we simply do not know the exact loss in inflow, and we do not care.
As a consequence, the invariant of $\y$ breaks, because it no longer receive its key $\akey_\y$ as inflow.
The subsequent actions will re-establish the invariant for $\y$.
Overall, we arrive at the postcondition on \cref{app:code:complex-remove:footprint-post-move}.

Next, \code{removeComplex} finalizes swapping the contents of $\x$ and $\y$ by swapping their \code{del} flags, \cref{app:code:complex-remove:del-x,app:code:complex-remove:del-y}.
This results in $\ctnof{\x}$ being updated to $\abscontent_2$.
Because $\abscontent_2\subseteq\set{\akey_\y}\subseteq\ksof{\x}$, we obtain the invariant $\inv{\abscontent_2,\set{\x},\setnodes}$.
The content of $\y$, on the other hand, is deleted, $\ctnof{\y}=\emptyset$.
This is the annotation on \cref{app:code:complex-remove:footprint-post}.
Note that these updates do not change the flow.
At this point, we can recompose $\x$ with context and obtain $\inv{\abscontent,\setnodes\setminus\set{\y},\setnodes}$, \cref{app:code:complex-remove:intermediate}.
It remains to incorporate $\y$.

Finally, \cref{app:code:complex-remove:unlink} unlinks $\y$.
As stated earlier, we do not present the details of the unlinking here as it is similar to \code{removeSimple}.
Overall, the unlinking results in the inflow of $\y$ to be $\bot$.
Together with $\delof{\y}$, this reestablished the invariant, $\inv{\emptyset,\set{\y},\setnodes}$.
Recomposing it with $\inv{\abscontent,\setnodes\setminus\set{\y},\setnodes}$, we arrive at the desired $\inv{\abscontent,\setnodes}$, \cref{app:code:complex-remove:post}.

\subsection{Rotations}
\label{app:bst-rot:rotation}


\begin{figure}
  \begin{minipage}[b]{.49\textwidth}
  \begin{figure}[H]
\begin{lstlisting}[language=SPL,belowskip=0pt]
$\annot{
  \bst{\abscontent}
}$
$\annot{
  \inv{\abscontent,\setnodes}
}$
def ^rotate^() {
  $\x$ = getNode();
  $\y$ = $\x$.left;
  if ($\y$ == null) return;
  $\z$ = $\y$.left;
  if ($\z$ == null) return;
  $\annotml{
    & \inv{\abscontent,\setnodes} \MSTAR \x,\y,\z\in\setnodes \MSTAR \isof{\x}\neq\bot \\
    &\!\mstar{} \leftof{\x}=\y \MSTAR \leftof{\y}=\z
  }$
  // duplicate @{\color{colorCodeComment}$y$}@
  $\dup$ = new Node($\y$.key);
  `${\color{colorGhost}\dup}$.dup = ${\color{colorGhost}\dupvalright}$;'
  $\dup$.del = $\y$.del;
  $\dup$.righ = $\y$.right;
  // restructure
  $\dup$.left = $\z$.right;
  $\annotml{
    & \inv{\abscontent,\setnodes} \mstar \nodeof{\dup} \mstar \Phi(\setnodes,\x,\y,\z,\dup) \,\mstar\, \leftof{\x}=\y \\
    &\!\mstar{} \dupof{\dup}=\dupvalright \MSTAR \leftof{\dup}=\rightof{\z}
  }$
  $\z$.right = $\dup$; @\label{app:code:rotate:insert-dup}@
  $\annotml{
    & \inv{\abscontent,\setnodes,\setnodes\cup\setc{\dup}} \MSTAR \nodeof{\dup} \MSTAR \Phi(\setnodes,\x,\y,\z,\dup) \\
    &\!\mstar{} \leftof{\x}=\y \MSTAR \dupof{\dup}=\dupvalright \MSTAR \dup=\rightof{\z}
  }$
  $\x$.left, ${\color{colorGhost}\dup}$`.dup' = $\z$, ${\color{colorGhost}\dupvalnone}$;  @\label{app:code:rotate:unlink-original}@
  $\annotml{
    & \inv{\abscontent,\setnodes{\setminus}\setc{\y},\setnodes\cup\setc{\dup}} \mstar \nodeof{\dup} \mstar \Phi(\setnodes,\x,\y,\z,\dup) \\
    &\!\mstar{} \invp{\ctnof{\y},\y,\setnodes\cup\setc{\dup}} \MSTAR \leftof{\x}=\z \\
    &\!\mstar{} \dupof{\dup}=\dupvalnone \MSTAR \dup=\rightof{\z} \MSTAR \isof{\y}=\bot
  }$
  $\annotml{
    &\inv{\abscontent,\setnodes{\setminus}\setc{\y}\cup\setc{\dup},\setnodes\cup\setc{\dup}} \MSTAR \isof{\y}=\bot \\
    &\!\mstar{} \invp{\ctnof{\y},\y,\setnodes\cup\setc{\dup}} \MSTAR \delof{\dup}=\delof{\y}
  }$
  $\y$.del = true;
  $\annot{
    \inv{\abscontent,\setnodes\cup\setc{\dup}}
  }$
}
$\annot{
  \bst{\abscontent}
}$
\end{lstlisting}
    \caption{%
      Standard right rotation.
      It duplicates the target node, inserts the duplicate, and then removes the target node.
      See \cref{app:fig:rotate-aux} for auxiliary definition.%
      \label{app:fig:rotate}%
    }
  \end{figure}
  \end{minipage}
  \hfill
  \begin{minipage}[b]{.475\textwidth}
  \begin{figure}[H]
    \[
      \Phi(\setnodes,\x,\y,\z,\dup) ~~\defeq~ \begin{aligned}[t]
        & \x,\y,\z\in\setnodes \\
        {}\mstar{}& \isof{\x}\neq\bot \\
        {}\mstar{}& \leftof{\y}=\z \\
        {}\mstar{}& \rightof{\dup}=\rightof{\y} \\
        {}\mstar{}& \keyof{\dup}=\keyof{\y} \\
        {}\mstar{}& \delof{\dup}=\delof{\y}
      \end{aligned}
    \]%
    \vspace{-3mm}
    \caption{
      Auxiliary predicate for \code{rotate}, \cref{app:fig:rotate}.%
      \label{app:fig:rotate-aux}%
    }
  \end{figure}\vspace{-3mm}
  \begin{figure}[H]
\begin{lstlisting}[language=SPL,belowskip=0pt]
$\annot{
  \inv{\abscontent,\setnodes} \mstar \x\in\setnodes \mstar \isof{\x}\neq\bot \mstar \rightof{\x}\neq\nullptr
}$
def ^findSucc^(Node* $\x$) {
  $\p$ = $\x$.right;
  $\y$ = $\p$.left;
  assume($\y$ != null);
  $\annotml{
    &\inv{\abscontent,\setnodes} \MSTAR \x,\p,\y\in\setnodes \MSTAR \isof{\p}\neq\bot \\
    &\!\mstar{} \leftof{\p}=\y \MSTAR \isof{\y}\subseteq\osrof{\x}
  }$
  while ($\y$.left != null) {
    $\annotml{
      &\inv{\abscontent,\setnodes} \,\mstar\, \x,\p,\y,\leftof{\y}\in\setnodes \,\mstar\, \isof{\p}\neq\bot \\
      &\!\mstar{} \leftof{\p}=\y \MSTAR \isof{\y}\subseteq\osrof{\x}
    }$
    $\p$ = $\y$;
    $\y$ = $\p$.left;
    $\annotml{
      &\inv{\abscontent,\setnodes} \MSTAR \x,\p,\y\in\setnodes \MSTAR \isof{\p}\neq\bot \\
      &\!\mstar{} \leftof{\p}=\y \MSTAR \isof{\y}\subseteq\osrof{\x}
    }$
  }
  $\annotml{
    &\inv{\abscontent,\setnodes} \MSTAR \x,\p,\y\in\setnodes \MSTAR \isof{\y}\subseteq\osrof{\x} \\
    &\!\mstar{} \isof{\p}\neq\bot \MSTAR \leftof{\p}=\y \MSTAR \leftof{\y}=\nullptr
  }$
  return $\x$, $\y$;
}
$\annotml{
  \p,\y.~~
  &\inv{\abscontent,\setnodes} \MSTAR \x,\p,\y\in\setnodes \MSTAR \leftof{\p}=\y \\
  &\!\mstar{} \leftof{\y}=\nullptr \MSTAR \keyof{\y}\prall{<}\keyof{\p} \\
  &\!\mstar{} \isof{\y}\subseteq\osrof{\x}
}$
\end{lstlisting}
    \caption{%
      Helper function for finding the left-most leaf in the right subtree of a node $\x$.
      \label{app:fig:find-succ}%
    }
  \end{figure}
  \end{minipage}
\end{figure}

Operation \code{rotate} performs a standard right rotation.
We expect it not to change the logical contents of the tree.
The specification is: \[
  \annot{
    \bst{\abscontent}
  }
  ~~
  \mcode{rotate()}
  ~~
  \annot{
    \bst{\abscontent}
  }
  \ .
\]
Unlike for the previous maintenance operations, we use $\bst{\abscontent}$ instead of $\inv{\abscontent,\setnodes}$.
This is because \code{rotate} creates a new node to the heap graph that we \emph{hide} in the existential quantifier of $\bst{\abscontent}$.


\begin{figure}
  \center
\begin{tikzpicture}[baseline={(x.center)},level/.style={sibling distance = 3cm/#1, level distance = 1.5cm}]

  \node (x) [treenode] {$\x$}
    child [treeptr] {
      node (y) [treenode,xshift=-3mm] {$\y$}
        child [treeptr] {
          node (z) [treenode] {$\z$}
          child [treeptr] {
            node (A) [subtree] {$A$}
          }
          child [treeptr] {
            node (B) [subtree] {$B$}
          }
        }
        child [treeptr] {
          node (C) [subtree] {$C$}
        }
    }
  ;
  \draw[treeptr,dashed] (x) -- (1,-.8);
  \path (x) -- (y.north) node[midway,flowch] {$\ais\cap[-\infty,\x)$};
  \path (y) -- (z.north) node[midway,xshift=-5mm,flowch] {$\ais\cap[-\infty,\y)$};
  \path (y) -- (C.north) node[midway,xshift=4mm,flow] {$\ais\cap(\y,\x)$};
  \path (z) -- (A.north) node[midway,xshift=-6mm,flow] {$\ais\cap[-\infty,\z)$};
  \path (z) -- (B.north) node[midway,xshift=5mm,flow] {$\ais\cap(\z,\y)$};
  \node (i) [flow,anchor=north west] at (-.75,.85) {$\ais\neq\bot$};
  \path (i.220) edge[inflow,bend right=20] (x);
\end{tikzpicture}
\hfill$\stackrel{\text{\cref{app:code:rotate:insert-dup}}}{\scalebox{2}{$\rightsquigarrow$}}$\hfill
\begin{tikzpicture}[baseline={(x.center)},level/.style={sibling distance = 3cm/#1, level distance = 1.5cm}]

  \node (x) [treenode] {$\x$}
    child [treeptr] {
      node (y) [treenode,xshift=-3mm] {$\y$}
        child [treeptr] {
          node (z) [treenode] {$\z$}
          child [treeptr] {
            node (A) [subtree] {$A$}
          }
          child [treeptr] {
            node (c) [treenode,fill=colorGhost!60] {$\dup$}
              child {
                node (B) [subtree] {$B$}
              }
              child [treeptr] {
                node (C) [subtree] {$C$}
              }
          }
        }
        child {
          node {} edge from parent[draw=none]
        }
    }
  ;
  \draw[treeptr,dashed] (x) -- (1,-.8);
  \path (x) -- (y.north) node[midway,flowch] {$\ais\cap[-\infty,\x)$};
  \path (y) -- (z.north) node[midway,xshift=-5mm,flowch] {$\ais\cap[-\infty,\y)$};
  \draw[treeptr] (y) edge[bend left=35] (C.north);
  \path (y) -- (C.north) node[pos=.875,xshift=8mm,flow] {$\ais\cap(\y,\x)$};
  \path (z) -- (A.north) node[midway,xshift=-6mm,flow] {$\ais\cap[-\infty,\z)$};
  \path (z) -- (c.north) node[midway,xshift=5mm,flowch] {$\ais\cap(\z,\y)$};
  \path (c) -- (B.north) node[midway,xshift=-6mm,flow] {$\ais\cap(\z,\y)$};
  \path (c) -- (C.north) node[midway,xshift=0mm,flow] {$\bot$};
  \node (i) [flow,anchor=north west] at (-.75,.85) {$\ais\neq\bot$};
  \path (i.220) edge[inflow,bend right=20] (x);
\end{tikzpicture}
\hfill$\stackrel{\text{\cref{app:code:rotate:unlink-original}}}{\scalebox{2}{$\rightsquigarrow$}}$\hfill
\begin{tikzpicture}[baseline={(x.center)},level/.style={sibling distance = 3cm/#1, level distance = 1.5cm}]

  \node (x) [treenode] {$\x$}
    child [treeptr] {
      node (y) [treenode,xshift=-3mm] {$\y$} edge from parent[draw=none]
        child [treeptr] {
          node (z) [treenode] {$\z$}
          child [treeptr] {
            node (A) [subtree] {$A$}
          }
          child [treeptr] {
            node (c) [treenode] {$\dup$}
              child {
                node (B) [subtree] {$B$}
              }
              child [treeptr] {
                node (C) [subtree] {$C$}
              }
          }
        }
        child {
          node {} edge from parent[draw=none]
        }
    }
  ;
  \draw[treeptr,dashed] (x) -- (1,-.8);
  \draw[treeptr] (x) edge[bend right=25] (z.north);
  \path (x) -- (y.north) node[midway,flowch] {$\ais\cap[-\infty,\x)$};
  \path (y) -- (z.north) node[midway,flowch] {$\bot$};
  \draw[treeptr] (y) edge[bend left=35] (C.north);
  \path (y) -- (C.north) node[pos=.75,xshift=7mm,flow] {$\bot$};
  \path (z) -- (A.north) node[midway,xshift=-6mm,flow] {$\ais\cap[-\infty,\z)$};
  \path (z) -- (c.north) node[midway,xshift=5mm,flowch] {$\ais\cap(\z,\y)$};
  \path (c) -- (B.north) node[midway,xshift=-6mm,flow] {$\ais\cap(\z,\y)$};
  \path (c) -- (C.north) node[midway,xshift=5mm,flow] {$\ais\cap(\y,\x)$};
  \node (i) [flow,anchor=north west] at (-.75,.85) {$\ais\neq\bot$};
  \path (i.220) edge[inflow,bend right=20] (x);
\end{tikzpicture}
  \caption{%
    The sequence of updates for a right rotation of node $\y$.
    First, a duplicate $\dup$ of $\y$ is inserted, with $\dupof{\dup}=\dupvalright$.
    Then, $\y$ is unlinked and $\dup$ becomes the original, $\dupof{\dup}=\dupvalnone$.
    The footprint of the updates consists of the nodes $\set{\x,\y,\z,\dup}$.%
    \label{app:fig:rotate-fig}%
  }
\end{figure}

\Cref{app:fig:rotate} gives the implementation and proof outline for \code{rotate}.
Given nodes $\y$ and $\z$ with $\selof{\y}{left}=\z$, the goal is to move $\y$ into $\z$'s right subtree.
In line with concurrent BST implementations, \code{rotate} does not perform this update in-place, but inserts a duplicate $\dup$ of $\y$ in $\z$'s right subtree and subsequently unlinks $\y$.
This breaks the tree shape temporarily, until $\y$ is unlinked.
To handle this in the proofs we set $\dup$'s \code{dup} field to $\dupvalright$ before inserting it and then to $\dupvalnone$ at the moment when $\y$ is unlinked.
Node $\y$ is unlinked by replacing in $\y$'s parent $\x$ the child pointer to $\y$ by $\z$.
Consult \cref{app:fig:rotate-fig} for an illustration.
The figure shows that all updates are contained within the footprint $\set{\x,\y,\z,\dup}$---the subtrees $A,B,C$ are not aware of the changes.
In particular, there is no unbounded flow update.
Hence, \code{rotate} can be verified using standard arguments.
In particular, the \ruleref{frame} rule is applicable.



\section{Proofs of Section~\ref{Section:CAReasoning}}
\label{sec:ctx-proofs}

\begin{proof}[Proof of \Cref{thm:casl-soundness}]
	For all rules, we show that the validity of their precondition entails the validity of their postcondition.
	From this, the overall claim follows by a straightforward rule induction of the \theLogicSeq derivation tree.

	\paragraph{Rule \ruleref{com}}
	We have $\casemof{\acom}{\acontext}{\apred}\predleq\apredp$.
	By \eqref{Equation:Mediation} then, $\semof{\acom}{\apred\mstar\acontext} \predleq \casemof{\acom}{\acontext}{\apred}\mstar\acontext \predleq \apredp\mstar\acontext$.
	That is, $\models \hoareof{\apred\mstar\acontext}{\acom}{\apredp\mstar\acontext}$ is valid.
	So $\models \choareof{\acontext}{\apred}{\acom}{\apredp}$ is valid as well, by \Cref{def:casl-validity}.

	\paragraph{Rule \ruleref{consequence}}
	We have $\apred\predleq\apred'$, $\apredp\predleq\apredp'$, and $\models \choareof{\acontext}{\apred'}{\astmt}{\apredp'}$.
	By \Cref{def:casl-validity}, the latter means $\models \hoareof{\apred'\mstar\acontext}{\astmt}{\apredp'\mstar\acontext}$.
	Since separation logic is sound, \Cref{Lemma:SoundnessSL}, we obtain $\models \hoareof{\apred\mstar\acontext}{\astmt}{\apredp\mstar\acontext}$ using \ruleref{consequence}.
	Again by \Cref{def:casl-validity}, we get $\models \choareof{\acontext}{\apred}{\astmt}{\apredp}$.

	\paragraph{Rule \ruleref{seq}}
	We have $\models \choareof{\acontext}{\apred}{\astmt_1}{\apredp}$ and $\models \choareof{\acontext}{\apredp}{\astmt_2}{\apredppp}$.
	By \Cref{def:casl-validity}, this means $\models \hoareof{\apred\mstar\acontext}{\astmt_1}{\apredp\mstar\acontext}$ and $\models \hoareof{\apredp\mstar\acontext}{\astmt_2}{\apredppp\mstar\acontext}$.
	By \Cref{Lemma:SoundnessSL}, an application of rule \ruleref{seq} gives $\models \hoareof{\apred\mstar\acontext}{\seqof{\astmt_1}{\astmt_2}}{\apredppp\mstar\acontext}$.
	Then, $\models \choareof{\acontext}{\apred}{\seqof{\astmt_1}{\astmt_2}}{\apredppp}$ follows by \Cref{def:casl-validity}.

	\paragraph{Rule \ruleref{choice}}
	We have $\models \choareof{\acontext}{\apred}{\astmt_1}{\apredp}$ and $\models \choareof{\acontext}{\apred}{\astmt_2}{\apredp}$.
	By \Cref{def:casl-validity}, this means $\models \hoareof{\apred\mstar\acontext}{\astmt_1}{\apredp\mstar\acontext}$ and $\models \hoareof{\apred\mstar\acontext}{\astmt_2}{\apredp\mstar\acontext}$.
	By \Cref{Lemma:SoundnessSL}, rule \ruleref{choice} gives $\models \hoareof{\apred\mstar\acontext}{\choiceof{\astmt_1}{\astmt_2}}{\apredp\mstar\acontext}$.
	Then, $\models \choareof{\acontext}{\apred}{\choiceof{\astmt_1}{\astmt_2}}{\apredp}$ follows by \Cref{def:casl-validity}.

	\paragraph{Rule \ruleref{loop}}
	We have $\models \choareof{\acontext}{\apred}{\astmt}{\apred}$.
	By \Cref{def:casl-validity}, this means $\models \hoareof{\apred\mstar\acontext}{\astmt}{\apred\mstar\acontext}$.
	By \Cref{Lemma:SoundnessSL}, rule \ruleref{loop} gives $\models \hoareof{\apred\mstar\acontext}{\loopof{\astmt}}{\apred\mstar\acontext}$.
	Then, $\models \choareof{\acontext}{\apred}{\astmt}{\apredp}$ follows by \Cref{def:casl-validity}.

	\paragraph{Rule \ruleref{frame}}
	We have $\models \choareof{\acontext}{\apred}{\astmt}{\apredp}$.
	By \Cref{def:casl-validity}, this means $\models \hoareof{\apred\mstar\acontext}{\astmt}{\apredp\mstar\acontext}$.
	By \Cref{Lemma:SoundnessSL}, rule \ruleref{frame} gives $\models \hoareof{\apred\mstar\acontext\mstar\apredppp}{\astmt}{\apredp\mstar\acontext\mstar\apredppp}$.
	Then, $\models \choareof{\acontext}{\apred\mstar\apredppp}{\astmt}{\apredp\mstar\apredppp}$ follows by \Cref{def:casl-validity}.

	\paragraph{Rule \ruleref{context}}
	We have $\models \choareof{\acontext\mstar\apredppp}{\apred}{\astmt}{\apredp}$.
	By the definition of validity, \Cref{def:casl-validity}, this means $\models \hoareof{\apred\mstar\acontext\mstar\apredppp}{\astmt}{\apredp\mstar\acontext\mstar\apredppp}$.
	Again by \Cref{def:casl-validity}, we get $\models \choareof{\acontext}{\apred\mstar\apredppp}{\astmt}{\apredp\mstar\apredppp}$.

	\paragraph{Rule \ruleref{widen}}
	We have $\models \choareof{\acontext}{\apred\mstar\apredppp}{\astmt}{\apredp\mstar\apredppp}$.
	By the definition of validity, \Cref{def:casl-validity}, this means $\models \hoareof{\apred\mstar\acontext\mstar\apredppp}{\astmt}{\apredp\mstar\acontext\mstar\apredppp}$.
	Again by \Cref{def:casl-validity}, we get $\models \choareof{\acontext\mstar\apredppp}{\apred}{\astmt}{\apredp}$.
\end{proof}

\begin{proof}[Proof of \Cref{lem:conservative-extension}]
	Relative soundness and relative completeness follow from a rule induction over the \theLogicSeq derivation that constructs a SL derivation mimicking the \theLogicSeq derivation one-to-one with context, and vice versa.
\end{proof}

\begin{proof}[Proof of \Cref{thm:soundness-OG}]
	See \cite[Theorem 4.2]{DBLP:journals/pacmpl/MeyerWW22}.
\end{proof}

\begin{proof}[Proof of \Cref{thm:soundness-OGCASL}]
	Analogous to the proof of \Cref{thm:casl-soundness}.
\end{proof}

\begin{proof}[Proof of \Cref{thm:contextualization}]
	As $\apredpp$ is a reflexive and transitive closure of $\apredppp$, we have $\apredppp\predleq\apredpp$. 
	It remains to prove $\semof{\acom}{\apred\mstar\apredpp}\predleq\apredp\mstar\apredpp$. 
	The interesting case is $\apredp\neq\abort\neq\apredpp$. 
	Then also $\apred'\neq\abort$ and we have 
	\begin{align*}
		&~~ \semof{\acom}{\apred\mstar\apredpp}\\
		\explain{$\upof{\acom}{\apred}\neq \abort$ by soundness of $\absup{\acom}$ and $\apred'\neq\abort$}=&~~ \upof{\acom}{\apred}\imult\apredpp\\
		\explain{Soundness of $\absup{\acom}$}\predleq&~~ \absupof{\acom}{\apred}\imult\apredpp\\
		\explain{Soundness of $\absimult$}\predleq&~~ \absupof{\acom}{\apred}\absimult\apredpp\\
		\explain{Definition of $\apred'$}=&~~ \apred'\absimult\apredpp\\
		\explain{Definition of $\absimult$}=&~~ \ghostabsof{\apredpp}{\apred'}\mstar\ghostabsof{\apred'}{\apredpp}\\
		\explain{Definition of $\apredp$ and $\rho$}=&~~ \apredp\mstar\rho(\apredpp)\\
		\explain{Definition of $\apredpp$}\predleq&~~ \apredp\mstar\apredpp
		\enspace .
		\qedhere
	\end{align*}
\end{proof}

\begin{proof}[Proof of \Cref{thm:approx-induced-casl-is-conservative}]
	The claim follows from \Cref{lem:conservative-extension}.
	We show that \Cref{lem:conservative-extension} applies.
	By definition, we have $\icasem{\acom}{\emp}=\sem{\acom}$.
	So it remains to show that $\icasem{\acom}{\apredppp}$ satisfies \eqref{Equation:Mediation} for all $\acom$ and $\apredppp$.
	Consider some $\apred\in\setpreds$.
	We establish $\semof{\acom}{\apred\mstar\apredppp} \predleq \icasemof{\acom}{\apredppp}{\apred}\mstar\apredppp$.
	For $\apredppp=\emp$, the inclusion follows immediately because $\icasem{\acom}{\emp}=\sem{\acom}$.
	Assume $\apredppp\neq\emp$.
	If $\icasemof{\acom}{\apredppp}{\apred}=\top$, then the desired inclusion holds by definition.
	So assume $\icasemof{\acom}{\apredppp}{\apred}\neq\top$.
	This means $\icasemof{\acom}{\apredppp}{\apred}=\ghostabsof{\apredppp}{\apred'}$ with $\apred'=\absupof{\acom}{\apred}\neq\top$ and $\apredppp=\ghostabsof{\apred'}{\apredppp}$.
	Observe that the latter means $\apredppp=\rho^*(\apredppp)$.

	Now, apply \Cref{thm:contextualization} to $\semof{\acom}{\apred\mstar\apredppp}$.
	We obtain $\semof{\acom}{\apred\mstar\acontext}\predleq\ghostabsof{\acontext}{\apred'}\mstar\acontext$ with $\acontext=\rho^*(\apredppp)$ and $\apredppp\predleq\acontext$.
	By the above observation, $\apredppp=\acontext$ must hold.
	That is, $\semof{\acom}{\apred\mstar\apredppp}\predleq\ghostabsof{\apredppp}{\apred'}\mstar\apredppp$.
	Hence, we arrive at $\semof{\acom}{\apred\mstar\apredppp}\predleq\icasemof{\acom}{\apredppp}{\apred}\mstar\apredppp$, as required.
	Overall, this concludes that $\icasem{\acom}{\bullet}$ induces a \theLogicSeq that is a conservative extension of SL.
\end{proof}

\section{Proofs for Section \ref{sec:instantiation}}
\label{app:fg-theory}

For simplicity, we write $\aflowconstraint.\fval$ and $\aflowconstraint.\outflow$ to refer to a flow graph $\aflowconstraint$'s flow and outflow, which are derived quantities (cf. \cref{sec:instantiation}).

\subsection{Additional Meta Theory}

\begin{definition}
	Consider functions $f,g$ with the same signature.
	We write $f \fpreldot g$ iff $f(x) \fprel g(x)$ for all $x$.
	We write $f \fpreleq g$ iff $f(x) \fprel g(x) \,\vee\, f(x) = g(x)$ for all $x$.
\end{definition}

\begin{definition}
	Consider a flow graph $\aflowconstraint=(\setnodes,\edges,\aninflow)$ and $\setnodesp\subseteq\nat$.
	Define \[
	  \restrictto{\aflowconstraint}{\setnodesp}
	  ~~\defeq~~
	  (\setnodes\cap\setnodesp, \restrictto{\edges}{(\setnodes\cap\setnodesp)\times\nat}, \inflow')
	\]
	such that
	\begin{inparaenum}
	 	\item $\inflow'(\anodepp, \anodep)\defeq\inflow(\anodepp, \anodep)$ for all $\anodepp\in\nat\setminus\setnodes$, $\anodep\in\setnodes\cap\setnodesp$, and
	 	\item $\inflow(\anode, \anodep)\defeq\edges_{(\anode, \anodep)}(\aflowconstraint.\fvalof{\anode})$ for all $\anode\in\setnodes\setminus\setnodesp$,  $\anodep\in\setnodes\cap\setnodesp$.
	\end{inparaenum}
\end{definition}

\begin{definition}
	Define $\transformerof{\aflowconstraint_1} \fprel_{\inflow} \transformerof{\aflowconstraint_2}$
	iff
	$\transformerof{\aflowconstraint_1}(\inflow') \fpreldot \transformerof{\aflowconstraint_2}(\inflow')$, for all inflows $\inflow' \leq \inflow$ and all nodes $\anode$.
\end{definition}

\begin{remark}
	We have $\aflowconstraint_1 \ctxfprel \aflowconstraint_2$ iff $\aflowconstraint_1.\setnodes = \aflowconstraint_2.\setnodes$ and $\aflowconstraint_1.\inflow = \aflowconstraint_2.\inflow$ and $\transformerof{\aflowconstraint_1}\fprel_{\aflowconstraint_1.\inflow}\transformerof{\aflowconstraint_2}$.
\end{remark}

\begin{definition}
	Consider an inflow $\inflow : (\nat\setminus\setnodes)\times\setnodes \to \amonoid$ and a set of nodes $\setnodesp\subseteq\nat\setminus\setnodes$.
	Then, the $\fprel$-upward $\setnodesp$-closure of $\inflow$ is: \[
		\fpclosureof{\setnodesp}{\inflow}
		~\defeq~
		\setcond{
			\restrictto{\inflow}{(\nat\setminus\setnodesp)\times\setnodes} \,\uplus\, \inflow'
		}{
			\forall\anode\in\setnodes.~\sum_{\anodep\in\setnodesp} \inflow(\anodep,\anode) \fprel \sum_{\anodep\in\setnodesp} \inflow'(\anodep,\anode)
		}
		\ .
	\]
	We write $\fpclosureof{\setnodesp}{\aflowconstraint}$ for a flow constraint $\aflowconstraint=(\setnodes,\edges,\inflow)$ to mean $\fpclosureof{\setnodesp}{\inflow}$.
	We may also write $\fpclosureof{\aflowconstraint'}{\aflowconstraint}$ for another flow graph $\aflowconstraint'$ to mean $\fpclosureof{\aflowconstraint'\!.\setnodes}{\aflowconstraint}$.
\end{definition}

\begin{definition}
	\label{def:fpcompat}
	For a relation $\fprel {\:\subseteq\;} \amonoid \times \amonoid$ and a flow graph $\aflowconstraint$, we write $\fpcompatible[\fprel]{\aflowconstraint}$ if:
	\begin{compactenum}[({C}1)]
		\item\customlabel{def:fpcompat:refl-trans}{C\theenumi}
			$\fprel$ is transitive and $\bot\fprel\bot$,
		\item\customlabel{def:fpcompat:add}{C\theenumi}
			$\amonval \fprel \amonvalp$ implies $\amonval+\amonvalpp \fprel \amonvalp+\amonvalpp$ for all $\amonval,\amonvalp,\amonvalpp\in\amonoid$,
		\item\customlabel{def:fpcompat:edges}{C\theenumi}
			all edge functions $f$ in $\aflowconstraint.\edges$ are $\fpreldot$-monotonic, i.e., 
			$\amonval \fprel \amonvalp$ implies $f(\amonval) \fprel f(\amonvalp)$, and
		\item\customlabel{def:fpcompat:lfp}{C\theenumi}
			$\bigjoin\achain\fpreldot\bigjoin\achainp$ holds for all $\leq$-ascending chains $\achain,\achainp$ of the form $\achain=f^0(\bot)\leq f^1(\bot)\leq\cdots$ and $\achainp=g^0(\bot)\leq g^1(\bot)\leq\cdots$ that have the following properties: $\bigjoin\achain\fpreldot g(\bigjoin\achain)$, and $f^i(\bot)\fpreldot g^i(\bot)$ for all $i\in\nat$, and $f,g$ are $\leq$-continuous and $\fpreldot$-monotonic functions with the signature $f,g:(\aflowconstraint.\setnodes{\to}\amonoid)\to(\aflowconstraint.\setnodes{\to}\amonoid)$.
	\end{compactenum}
\end{definition}


\begin{lemma}
	\label{thm:our-monoid-has-bot}
	The $\omega$-cpo $(\amonoid,\leq)$ has a least element $\bot$, namely $\bot=\monunit$.
\end{lemma}

\begin{lemma}
	\label{thm:our-monoid-has-addprop}
	Consider $\amonval,\amonval',\amonvalp,\amonvalp',\amonvalpp\in\amonoid$.
	If $\amonval\leq\amonvalp$, then $\amonval+\amonvalpp\leq\amonvalp+\amonvalpp$.
	Moreover, if $\amonval\leq\amonvalp$ and $\amonval'\leq\amonvalp'$, then $\amonval+\amonval'\leq\amonvalp+\amonvalp'$.
\end{lemma}

\begin{lemma}
	\label{thm:our-monoid-sum-vs-join}
	Consider $\leq$-ascending chains $\achain_0,\dots,\achain_n$ with $\achain_i=\amonval_{i,0}\leq\amonval_{i,1}\leq\cdots$.
	Then, we have the following: $\sum_{i=0}^{n}\bigjoin\setcond{\amonval_{i,j}}{j\in\nat}=\bigjoin\setcond{\sum_{i=0}^{n}\amonval_{i,j}}{j\in\nat}$.
\end{lemma}

\begin{lemma}
	\label{thm:add-relations}
	If $\fpcompatible[\fprel]{\dontcare}$, then $\amonval\fprel\amonvalp$ and $\amonval'\fprel\amonvalp'$ implies $\amonval+\amonval'\fprel\amonvalp+\amonvalp'$, for all $\amonval,\amonval',\amonvalp,\amonvalp'\in\amonoid$.
\end{lemma}

\begin{lemma}
	\label{thm:continuous-implies-monotonic-special}
	Consider $f \in \contfunof{\amonoid\to\amonoid}$ and $g \in \contfunof{(\setnodes{\to}\amonoid)\to(\setnodes{\to}\amonoid)}$ for some $\setnodes\subseteq\nat$.
	Then both $f$ and $g$ are $\leq$-monotonic.
\end{lemma}

\begin{lemma}[Kleene]
	\label{thm:fixpoint-kleene}
	Consider a $\leq$-continuous function $f: (\setnodes{\to}\amonoid) \to (\setnodes{\to}\amonoid)$ for some $\setnodes\subseteq\nat$.
	Then, we have (i) an $\leq$-ascending Kleene chain $\achain = f^0(\bot)\leq f^1(\bot)\leq\cdots$, (ii) the join $\bigjoin\achain$ exists in $\amonoid$, (iii) and $\lfpof{f}=\bigjoin\achain$.
\end{lemma}

\begin{lemma}
	\label{thm:flow-iteration-monotonic-continuous}
	Consider a flow graph $\aflowconstraint=(\setnodes,\edges,\inflow)$ and function $f : (\setnodes\to\amonoid)\to(\setnodes\to\amonoid)$ defined by $f(\cval)(\anode)=\sum_{\anodep\in\nat\setminus\setnodes}\inflow(\anodep,\anode) + \sum_{\anodep\in\setnodes}\edgesatof{\anodep}{\anode}{\cval(\anodep)}$.
	Then, $f$ is $\leq$-monotonic and $\leq$-continuous.
	Moreover, if $\fpcompatible[\fprel]{\aflowconstraint}$, then $f$ is also $\fprel$-monotonic.
\end{lemma}

\begin{lemma}
	\label{thm:flow-as-lfp}
	Consider a flow graph $\aflowconstraint=(\setnodes,\edges,\inflow)$.
	Define the function $\fiter : (\setnodes\prall{\to}\amonoid) \to (\setnodes\prall{\to}\amonoid)$ by $\fiter(\fval)(\anode)\defeq\fval(\anode)$
	Then, the flow in $\aflowconstraint$. is given by $\aflowconstraint.\fval = \lfpof{\fiter} = \bigjoin\setcond{\fiter^i(\bot)}{i\in\nat}$.
\end{lemma}


\begin{lemma}
	\label{thm:restriction-vs-statemult}
	Consider $\aflowconstraint$ and $\setnodesp\subseteq\nat$.
	Then, we have:
	\begin{inparaenum}[(i)]
		\item $\restrictto{\aflowconstraint}{\setnodesp}.\fval = \restrictto{\aflowconstraint.\fval}{\setnodesp}$,
		\item $\restrictto{\aflowconstraint}{\setnodesp} \statemultdef \restrictto{\aflowconstraint}{\aflowconstraint.\setnodes\setminus\setnodesp}$, and
		\item $\restrictto{\aflowconstraint}{\setnodesp} \statemult \restrictto{\aflowconstraint}{\aflowconstraint.\setnodes\setminus\setnodesp} = \aflowconstraint$.
	\end{inparaenum}
\end{lemma}

\begin{lemma}
	\label{thm:outflow-vs-statemult}
	Consider $\aflowconstraint_1,\aflowconstraint_2$ with $\aflowconstraint_1\statemultdef\aflowconstraint_2$ and $\anode\in\aflowconstraint_1.\setnodes$ and $\anodep\in\nat\setminus(\aflowconstraint_1.\setnodes\cup\aflowconstraint_2.\setnodes)$.
	Then, we have the following: $(\aflowconstraint_1\statemult\aflowconstraint_2).\outflow(\anode,\anodep)=\aflowconstraint_1.\outflow(\anode,\anodep)=\aflowconstraint_1.\edgesatof{\anode}{\anodep}{\aflowconstraint_1.\fval(\anode)}$.
\end{lemma}

\begin{lemma}
	\label{thm:transformer-vs-statemult}
	Consider $\aflowconstraint_1,\aflowconstraint_2$ with $\aflowconstraint_1\statemultdef\aflowconstraint_2$ and $\anode\in\nat\setminus(\aflowconstraint_1.\setnodes\cup\aflowconstraint_2.\setnodes)$.
	Then, the transformer decomposes as follows: $\transformerofof{\aflowconstraint_1\statemult\aflowconstraint_2}{(\aflowconstraint_1\statemult\aflowconstraint_2).\inflow}(\anode)=\transformerofof{\aflowconstraint_1}{\aflowconstraint_1.\inflow}(\anode)+\transformerofof{\aflowconstraint_2}{\aflowconstraint_2.\inflow}(\anode)$.
\end{lemma}

\begin{lemma}
	\label{thm:inflow-leq}
	Consider $\aflowconstraint$ and inflows $\inflow_1 \leq \inflow_2$.
	Then,
	\begin{inparaenum}
		\item $\aflowconstraint[\inflow\mapsto\inflow_1].\fval \leq \aflowconstraint[\inflow\mapsto\inflow_2].\fval$, and
		\item $\transformerofof{\aflowconstraint}{\inflow_1} \leq \transformerofof{\aflowconstraint}{\inflow_2}$.
	\end{inparaenum}
\end{lemma}

\begin{lemma}
	\label{thm:inflow-fprel}
	Consider $\aflowconstraint$ with $\fpcompatible[\fprel]{\aflowconstraint}$ and inflows $\inflow_1 \fpreleq \inflow_2$.
	Then, we have the following:
	\begin{inparaenum}
		\item $\aflowconstraint[\inflow\mapsto\inflow_1].\fval \fpreldot \aflowconstraint[\inflow\mapsto\inflow_2].\fval$, and
		\item $\transformerofof{\aflowconstraint}{\inflow_1} \fpreldot \transformerofof{\aflowconstraint}{\inflow_2}$.
	\end{inparaenum}
\end{lemma}

\begin{lemma}
	\label{thm:partial-fixpoint}
	Consider flow graphs $\aflowconstraint_1=(\setnodes_1,\edges_1,\inflow_1)$ and $\aflowconstraint_2=(\setnodes_2,\edges_2,\inflow_2)$ with $\aflowconstraint_1\statemultdef\aflowconstraint_2$.
	Define functions $f : (\setnodes_2 \prall{\to} \amonoid) \to (\setnodes_2 \prall{\to} \amonoid)$ and $\inflow_\cval : (\nat\setminus\setnodes_1) \times \setnodes_1 \to \amonoid$ by:
	\begin{align*}
		f(\cval)(\anode) ~&\defeq~~
			\sum_{\anodepp\notin\setnodes_1\cup\setnodes_2} (\aflowconstraint_1\statemult\aflowconstraint_2).\inflowat{\anodepp,\anode}
			+
			\transformerofof{\aflowconstraint_1}{\inflow_\cval}(\anode)
			+
			\sum_{\anodep\in\setnodes_2} \aflowconstraint_2.\edgesatof{\anodep}{\anode}{\cval(\anodep)}
		\\
		\inflow_\cval(\anode,\anodep) ~&\defeq~~
			\anode\in\setnodes_2 ~~~?~~~
			\aflowconstraint_2.\edgesatof{\anodep}{\anode}{\cval(\anodep)}
			~~:~~
			(\aflowconstraint_1\statemult\aflowconstraint_2).\inflowat{\anode,\anodep}
	\end{align*}
	Then, $f$ is $\leq$-monotonic and $\leq$-continuous and $\aflowconstraint_2.\fval = \lfpof{f} = \bigjoin\setcond{f^i(\bot)}{i\in\nat}$.
	Furthermore, if $\fpcompatible[\fprel]{\aflowconstraint_1\statemult\aflowconstraint_2}$, then $f$ is also $\fpreldot$-monotonic.
\end{lemma}

\begin{lemma}
	\label{thm:upward-closed-outflow}
	Consider $\aflowconstraint_1,\aflowconstraint_2,\aflowconstraint_F$ with $\aflowconstraint_1\statemultdef \aflowconstraint_F$ and $\aflowconstraint_1 \ctxfprel \aflowconstraint_2$ and $\fpcompatible[\fprel]{\aflowconstraint_1,\aflowconstraint_2,\aflowconstraint_F}$.
	Then there is a flow graph $\aflowconstraint_{2+F} = (\aflowconstraint_2.\setnodes\uplus\aflowconstraint_F.\setnodes,\:\aflowconstraint_2.\edges\uplus\aflowconstraint_F.\edges,\:\aninflow)$ such that $\transformerofof{\aflowconstraint_1\statemult\aflowconstraint}{\aninflow}\fpreldot\transformerofof{\aflowconstraint_{2+F}}{\aninflow}$, where $\aninflow=(\aflowconstraint_1\statemult\aflowconstraint_F).\inflow$.
	Moreover, $\restrictto{\aflowconstraint_{2\statemult F}}{\aflowconstraint_2.\setnodes}\in\fpclosureof{\aflowconstraint_F}{\aflowconstraint_2}$ and $\restrictto{\aflowconstraint_{2\statemult F}}{\aflowconstraint_F.\setnodes}\in\fpclosureof{\aflowconstraint_2}{\aflowconstraint_F}$.
\end{lemma}

\begin{theorem}
	\label{thm:upward-closed-framing}
	Consider flow graphs $\aflowconstraint_1,\aflowconstraint_2,\aflowconstraint_F$ with $\aflowconstraint_1\statemultdef \aflowconstraint_F$ and $\aflowconstraint_1 \ctxfprel \aflowconstraint_2$.
	Furthermore, assume $\fpcompatible[\fprel]{\aflowconstraint_1,\aflowconstraint_2,\aflowconstraint_F}$.
	Then there are $\aflowconstraint_2'\in\fpclosureof{\aflowconstraint_F}{\aflowconstraint_2}$ and $\aflowconstraint_F'\in\fpclosureof{\aflowconstraint_2}{\aflowconstraint_F}$ such that $\aflowconstraint_2'\statemultdef\aflowconstraint_F'$ and  $\aflowconstraint_1 \statemult \aflowconstraint_F \ctxfprel \aflowconstraint_2'\statemult\aflowconstraint_F'$.
\end{theorem}

\begin{lemma}
	\label{thm:easy-fprel-lfp}
	Assume $\bigjoin\achain\fprel\bigjoin\achainp$ for all $\leq$-ascending chains $\achain,\achainp$ of the form $\achain=\amonval_0\leq\amonval_1\leq\cdots$ and $\achainp=\amonvalp_0\leq\amonvalp_1\leq\cdots$ with $\amonval_i\fprel\amonvalp_i$ for all $i\in\nat$.
	Then, \eqref{def:fpcompat:lfp} holds.
\end{lemma}

\begin{lemma}
	\label{thm:fprelcompatible-sub-omega-cpo}
	If $\fprel$ is a sub-$\omega$-cpo of $\leq$, then \eqref{def:fpcompat:lfp} holds.
\end{lemma}

\begin{lemma}
	\label{thm:fprelcompatible-acc}
	If $(\amonoid,\leq)$ satisfies the ascending chain condition, then \eqref{def:fpcompat:lfp} holds.
\end{lemma}

\begin{lemma}
	\label{thm:trivial-choices}
	For all flow graphs $\aflowconstraint$ we have $\fpcompatible[\leq]{\aflowconstraint}$ and $\fpcompatible[=]{\aflowconstraint}$.
\end{lemma}


\subsection{Proofs for Additional Meta Theory}

\begin{proof}[Proof of \Cref{thm:our-monoid-has-bot}]
	We show that $\bot=\monunit$ is the least element in $\amonoid$.
	Consider some $m\in\amonoid$.
	We show $\monunit \leq m$.
	By definition, we require some $x\in\amonoid$ such that $\monunit+x=m$.
	Choosing $x=m$ satisfies the requirement.
\end{proof}

\begin{proof}[Proof of \Cref{thm:our-monoid-has-addprop}]
	By definition, $\amonval\leq\amonvalp$ means $\amonval+\delta=\amonvalp$ for some $\delta\in\amonoid$.
	As a consequence, we have $\amonval+\delta+\amonvalpp=\amonvalp+\amonvalpp$.
	By definition again, $\amonval+\amonvalpp\leq\amonvalp+\amonvalpp$.
	This concludes the first claim.
	
	With this, we obtain $\amonval+\amonval'\leq\amonvalp+\amonval'$ and $\amonval'+\amonvalp\leq\amonvalp'+\amonvalp$.
	Together, this means $\amonval+\amonval'\leq\amonvalp+\amonvalp'$.
	This concludes the second claim.
\end{proof}

\begin{proof}[Proof of \Cref{thm:our-monoid-sum-vs-join}]
	Consider two chains $\achain=\amonval_0\leq\amonval_1\leq\cdots$ and $\achainp=\amonvalp_0\leq\amonvalp_1\leq\cdots$.
	We show $(\bigjoin\achain)+(\bigjoin\achainp)=\bigjoin(\achain+\achainp)$.
	The overall claim then follows from repeatedly applying the above argument.
	Because $\amonoid$ is continuous, we have:
	\begin{align*}
		(\bigjoin\achain)+(\bigjoin\achainp)
		~=~
		\bigjoin\setcond{\amonval_i + \bigjoin\achainp}{i\in\nat}
		~=~
		\bigjoin\setcond{\bigjoin\setcond{\amonval_i + \amonvalp_j}{j\in\nat}}{i\in\nat}
		\ .
	\end{align*}
	First, observe that $\amonval_i+\amonvalp_i\leq\amonval_i+\amonvalp_j$ for all $i\in\nat$ and all $j\geq i$ by \Cref{thm:our-monoid-has-addprop}.
	This means $\amonval_i+\amonvalp_i\leq\bigjoin\setcond{\amonval_i + \amonvalp_j}{j\in\nat}$ for all $i\in\nat$.
	Hence, we get:
	\begin{align*}
		\bigjoin\setcond{\amonval_i + \amonvalp_i}{i\in\nat}
		~~\leq~~
		\bigjoin\setcond{\bigjoin\setcond{\amonval_i + \amonvalp_j}{j\in\nat}}{i\in\nat}
		\ .
	\end{align*}
	Second, observe that we have $\amonval_i+\amonvalp_j\leq\amonval_k+\amonvalp_k$ with $k=\max(i,j)$ for all $i,j\in\nat$.
	This means that for every $i\in\nat$ there is some $k\in\nat$ such that $\bigjoin\setcond{\amonval_i+\amonvalp_j}{j\in\nat}\leq\amonval_k+\amonvalp_k$.
	Hence, $\bigjoin\setcond{\amonval_i+\amonvalp_j}{j\in\nat}\leq\bigjoin\setcond{\amonval_k+\amonvalp_k}{k\in\nat}$.
	Then we get:
	\begin{align*}
		\bigjoin\setcond{\bigjoin\setcond{\amonval_i + \amonvalp_j}{j\in\nat}}{i\in\nat}
		~~\leq~~&
		\bigjoin\setcond{\bigjoin\setcond{\amonval_k + \amonvalp_k}{k\in\nat}}{i\in\nat}
		\\~~=~~&
		\bigjoin\setcond{\amonval_k + \amonvalp_k}{k\in\nat}
		\ .
	\end{align*}
	Note here that the quality holds because, for every $\amonvalpp\in\amonoid$, the sequence $\amonvalpp\leq\amonvalpp\leq\cdots$ is an $\leq$-ascending chain the join of which exists and is $\amonvalpp$.
	Altogether, we arrive at:
	\begin{align*}
		\bigjoin\setcond{\bigjoin\setcond{\amonval_i + \amonvalp_j}{j\in\nat}}{i\in\nat}
		~~=~~
		\bigjoin\setcond{\amonval_i + \amonvalp_i}{i\in\nat}
		~~=~~
		\bigjoin(\achain+\achainp)
		\ .
	\end{align*}
	This concludes $(\bigjoin\achain)+(\bigjoin\achainp) = \bigjoin(\achain+\achainp)$, as desired.
\end{proof}

\begin{proof}[Proof of \Cref{thm:add-relations}]
	Consider $\amonval,\amonval',\amonvalp,\amonvalp'\in\amonoid$ with $\amonval\fprel\amonvalp$ and $\amonval'\fprel\amonvalp'$.
	By assumption, we have $\amonval+\amonval'\fprel\amonvalp+\amonval'$.
	Also by assumption, we have $\amonval'+\amonvalp\fprel\amonvalp'+\amonvalp$.
	Hence, $\amonval+\amonval'\fprel\amonvalp+\amonvalp'$ follows from $\fprel$ being transitive by \eqref{def:fpcompat:refl-trans} of $\fpcompatible[\fprel]{\aflowconstraint}$.
\end{proof}

\begin{proof}[Proof of \Cref{thm:continuous-implies-monotonic-special}]
	Consider some $\amonval,\amonvalp\in \amonoid$ with $\amonval \leq \amonvalp$.
	Because $(\amonoid,\leq)$ is an $\omega$-cpo by assumption, the join $\amonval \join \amonvalp$ exists.
	The join is $\amonval \join \amonvalp = \amonvalp$.
	Consequently, $f$ is defined for $\amonval \join \amonvalp$, $f(\amonval \join \amonvalp)\in\amonoid$.
	We obtain $f(\amonval \join \amonvalp)=f(\amonval) \join f(\amonvalp)$ because $f$ is $\leq$-continuous.
	This means the join $f(\amonval) \join f(\amonvalp)$ must exist as well, $(f(\amonval) \join f(\amonvalp))\in\amonoid$.
	Altogether, we conclude the first claim as follows: \(
		f(\amonvalp) = f(\amonval \join \amonvalp) = f(\amonval) \join f(\amonvalp) \geq f(\amonval)
	\).
	The second claim follows analogously.
\end{proof}

\begin{proof}[Proof of \Cref{thm:fixpoint-kleene}]
	First, we show $f^i(\bot)\leq f^{i+1}(\bot)$ for all $i\in\nat$.
	We proceed by induction.
	In the base case, $f^0(\bot)=\bot\leq f^1(\bot)$ because $\bot$ is the least element in $(\amonoid,\leq)$.
	For the induction step, we have $f^i(\bot)\leq f^{i+1}(\bot)$.
	By \Cref{thm:continuous-implies-monotonic-special}, $f$ is $\leq$-monotonic.
	Hence, we have the following by induction: $f^{i+1}(\bot)=f(f^i(\bot))\leq f(f^{i+1}(\bot))=f^{i+2}(\bot)$.
	This concludes the induction and means that $\achain = f^0(\bot)\leq f^1(\bot)\leq\cdots$ is an $\leq$-ascending chain, as desired.
	Because $(\amonoid,\leq)$ is an $\omega$-cpo, the join $\bigjoin\achain$ exists in $\amonoid$.

	\medskip
	It remains to show that $\bigjoin\achain=\lfpof{f}$ holds.
	We observe that $\bigjoin\achain$ is a fixed point of $f$:
	\begin{align*}
		&f(\bigjoin\achain)
		\\\explain{f $\leq$-continuous}=&~~
		\bigjoin\setcond{f^{i+1}(\bot)}{i\in\nat}
		\\\explain{$\bot$ least element}=&~~
		\bigjoin\setcond{f^{i}(\bot)}{i\in\nat}
		\\\explain{Def. $\achain$}=&~~
		\bigjoin\achain
	\end{align*}
	\medskip\newcommand{\dd}{\cval^\dagger}%
	We now show that $\bigjoin\achain$ is the least fixed point of $f$.
	To that end, consider another fixed point $\dd$ of $f$, i.e., $f(\dd)=\dd$.
	It suffices to show that $f^i(\bot) \leq \dd$ holds for all $i$, because this implies that the join over the $f^i(\bot)$ is at most $\dd$.
	We proceed by induction.
	In the base case, $f^0(\bot)=\bot\leq\dd$ because $\bot$ is the least element.
	For the induction step, we have $f^i(\bot) \leq \dd$.
	Because $f$ is $\leq$-monotonic as noted earlier, we obtain $f(f^i(\bot)) \leq f(\dd)$.
	Since $\dd$ is a fixed point of $f$, this means $f^{i+1}(\bot)\leq\dd$.
	This concludes the induction.
	We arrive at the desired $\lfpof{f}=\bigjoin\achain$.
\end{proof}

\begin{proof}[Proof of \Cref{thm:flow-iteration-monotonic-continuous}]
	If $f$ is $\leq$-continuous, then it is also $\leq$-monotonic by \Cref{thm:continuous-implies-monotonic-special}.
	So, it suffices to show that $f$ is $\leq$-continuous.
	Consider an $\leq$-ascending chain $\cval_0\leq\cval_1\leq\cdots$ with $\cval_i:\setnodes\to\amonoid$.
	We show $f(\bigjoin\setcond{\cval_i}{i\in\nat})=\bigjoin\setcond{f(\cval_i)}{i\in\nat}$.
	To that end, consider $\anode\in\setnodes$.
	We have:
	\begin{align*}
		&f(\bigjoin\setcond{\cval_i}{i\in\nat})(\anode)
		\\\explain{Def. $f$}=~~&
		\sum_{\anodep\in\nat\setminus\setnodes}\inflow(\anodep,\anode) + \sum_{\anodep\in\setnodes}\edgesatof{\anodep}{\anode}{(\bigjoin\setcond{\cval_i}{i\in\nat})(\anodep)}
		\\\explain{point-wise Def. $\join$}=~~&
		\sum_{\anodep\in\nat\setminus\setnodes}\inflow(\anodep,\anode) + \sum_{\anodep\in\setnodes}\edgesatof{\anodep}{\anode}{\bigjoin\setcond{\cval_i(\anodep)}{i\in\nat}}
		\\\explain{$\edgesat{\anodep}{\anode}$ continuous}=~~&
		\sum_{\anodep\in\nat\setminus\setnodes}\inflow(\anodep,\anode) + \sum_{\anodep\in\setnodes}\bigjoin\setcond{\edgesatof{\anodep}{\anode}{\cval_i(\anodep)}}{i\in\nat}
		\\\explain{\Cref{thm:our-monoid-sum-vs-join}}=~~&
		\sum_{\anodep\in\nat\setminus\setnodes}\inflow(\anodep,\anode) + \bigjoin\setcond{\sum_{\anodep\in\setnodes}\edgesatof{\anodep}{\anode}{\cval_i(\anodep)}}{i\in\nat}
		\\\explain{$\amonoid$ continuous}=~~&
		\bigjoin\setcond{\sum_{\anodep\in\nat\setminus\setnodes}\inflow(\anodep,\anode) + \sum_{\anodep\in\setnodes}\edgesatof{\anodep}{\anode}{\cval_i(\anodep)}}{i\in\nat}
		\\\explain{Def. $f$}=~~&
		\bigjoin\setcond{f(\cval_i)}{i\in\nat}
	\end{align*}
	This establishes that $f$ is $\leq$-continuous, as desired.
	
	Now, assume $\fpcompatible[\fprel]{\aflowconstraint}$.
	It remains to show that $f$ is $\fprel$-monotonic.
	To that end, consider $\cval_1,\cval_2:\setnodes\to\amonoid$ with $\cval_1\fpreldot\cval_2$.
	We show that $f(\cval_1)\fpreldot f(\cval_2)$ holds.
	For all $\anode\in\setnodes$ we have:
	\begin{align*}
		&f(\cval_1)(\anode)
		\\\explain{Def. $f$}=~~&
		\sum_{\anodep\in\nat\setminus\setnodes}\inflow(\anodep,\anode) + \sum_{\anodep\in\setnodes}\edgesatof{\anodep}{\anode}{\cval_1(\anodep)}
		\\\explain{see below}\fprel~~&
		\sum_{\anodep\in\nat\setminus\setnodes}\inflow(\anodep,\anode) + \sum_{\anodep\in\setnodes}\edgesatof{\anodep}{\anode}{\cval_2(\anodep)}
		\\\explain{Def. $f$}=~~&
		f(\cval_2)(\anode)
	\end{align*}
	where the approximation holds because we have $\edgesatof{\anodep}{\anode}{\cval_1(\anodep)}\fprel\edgesatof{\anodep}{\anode}{\cval_2(\anodep)}$ by \eqref{def:fpcompat:edges} of $\fpcompatible[\fprel]{\aflowconstraint}$, for all $\anodep\in\setnodes$, and thus the approximation is preserved under addition by \eqref{def:fpcompat:add} of $\fpcompatible[\fprel]{\aflowconstraint}$ together with \Cref{thm:add-relations}.
	In case the sums are empty, the approximation follows from $\monunit=\bot$ by \Cref{thm:our-monoid-has-bot} and $\bot\fprel\bot$ by \eqref{def:fpcompat:refl-trans} of $\fpcompatible[\fprel]{\aflowconstraint}$.
\end{proof}

\begin{proof}[Proof of \Cref{thm:flow-as-lfp}]
	Follows from \Cref{thm:flow-iteration-monotonic-continuous} together with \Cref{thm:fixpoint-kleene}.
\end{proof}

\begin{proof}[Proof of \Cref{thm:restriction-vs-statemult}]
	Proven in \cite[Proof of Lemma~4]{DBLP:conf/tacas/MeyerWW23,DBLP:journals/corr/abs-2304-04886}.
\end{proof}

\begin{proof}[Proof of \Cref{thm:outflow-vs-statemult}]
	By choice of $\anode,\anodep$ and the definition of the composition $\statemult$, we obtain $(\aflowconstraint_1\statemult\aflowconstraint_2).\outflow(\anode,\anodep)=\aflowconstraint_1.\outflow(\anode,\anodep)$.
	We immediately obtain $\aflowconstraint_1.\outflow(\anode,\anodep)=\aflowconstraint_1.\edgesatof{\anode}{\anodep}{\aflowconstraint_1.\fval(\anode)}$ from the definition of the outflow.
	This concludes the proof.
\end{proof}

\begin{proof}[Proof of \Cref{thm:transformer-vs-statemult}]
	Consider some node $\anodepp\in\nat\setminus(\aflowconstraint_1.\setnodes\cup\aflowconstraint_2.\setnodes)$.
	The we have:
	\begin{align*}
		&\transformerofof{\aflowconstraint_1\statemult\aflowconstraint_2}{(\aflowconstraint_1\statemult\aflowconstraint_2).\inflow}(\anodepp)
		\\\explain{Def. $\transformerof{\dontcare}$}=~~&
		\sum_{\anode\in\aflowconstraint_1.\setnodes\cup\aflowconstraint_2.\setnodes} (\aflowconstraint_1\statemult\aflowconstraint_2).\outflow(\anode,\anodepp)
		\\\explain{by set theory}=~~&
		\sum_{\anode\in\aflowconstraint_1.\setnodes} (\aflowconstraint_1\statemult\aflowconstraint_2).\outflow(\anode,\anodepp)
		+\sum_{\anodep\in\aflowconstraint_2.\setnodes} (\aflowconstraint_1\statemult\aflowconstraint_2).\outflow(\anodep,\anodepp)
		\\\explain{by \Cref{thm:outflow-vs-statemult}}=~~&
		\sum_{\anode\in\aflowconstraint_1.\setnodes} \aflowconstraint_1.\outflow(\anode,\anodepp)
		+\sum_{\anodep\in\aflowconstraint_2.\setnodes} \aflowconstraint_2.\outflow(\anodep,\anodepp)
		\\\explain{Def. $\transformerof{\dontcare}$}=~~&
		\transformerofof{\aflowconstraint_1}{\aflowconstraint_1.\inflow}(\anodepp)
		+\transformerofof{\aflowconstraint_2}{\aflowconstraint_2.\inflow}(\anodepp)
	\end{align*}
	This concludes the proof.
\end{proof}

\begin{proof}[Proof of \Cref{thm:inflow-leq}(i)]
	Let $\aflowconstraint=(\setnodes,\edges,\dontcare)$.
	Define $\aflowconstraint_1=\aflowconstraint[\inflow\mapsto\inflow_1]$ and $\aflowconstraint_2=\aflowconstraint[\inflow\mapsto\inflow_2]$.
	By \Cref{thm:flow-as-lfp}, the flow in $\aflowconstraint_i$ is $\aflowconstraint_i.\fval=\bigjoin\setcond{f_i^j(\bot)}{j\in\nat}$ with $f_i:(\setnodes{\to}\amonoid)\to(\setnodes{\to}\amonoid)$ defined by: \[
		f_i(\cval)(\anode) ~~=~~ \sum_{\anodep\in\nat\setminus\setnodes}\inflow_i(\anodep,\anode) + \sum_{\anodep\in\setnodes}\edgesatof{\anodep}{\anode}{\cval(\anodep)}
		\ .
	\]
	To conclude, it suffices to show that $f_1^j(\bot) \leq f_2^j(\bot)$ holds for all $j\in\nat$.
	We proceed by induction.
	In the base case, $\inflow_1\leq\inflow_2$ together with \Cref{thm:our-monoid-has-addprop} gives $f_1^0(\bot) \leq f_2^0(\bot)$ by the definition of $f_1,f_2$.
	For the induction step, we show $f_1(f_1^j(\bot)) \leq f_2(f_2^j(\bot)$.
	By induction together with $\leq$-monotonicity of $f_1$ by \Cref{thm:flow-iteration-monotonic-continuous}, we have $f_1(f_1^j(\bot)) \leq f_1(f_2^j(\bot))$.
	It remains to show $f_1(f_2^j(\bot)) \leq f_2(f_2^j(\bot))$.
	This immediately follows from the fact that $\inflow_1\leq\inflow_2$ together with \Cref{thm:our-monoid-has-addprop}, similarly to the base case.
\end{proof}

\begin{proof}[Proof of \Cref{thm:inflow-leq}(ii)]
	Let $\aflowconstraint=(\setnodes,\edges,\dontcare)$.
	Define $\aflowconstraint_1=\aflowconstraint[\inflow\mapsto\inflow_1]$ and $\aflowconstraint_2=\aflowconstraint[\inflow\mapsto\inflow_2]$.
	Let $\anodep\in\nat\setminus\setnodes$.
	We have $\transformerofof{\aflowconstraint}{\inflow_i}(\anodep)=\sum_{\anode\in\setnodes}\edgesatof{\anode}{\anodep}{\aflowconstraint_i.\fval(\anode)}$ by the definition of $\transformerof{\dontcare}$ and the outflow.
	Let $\anode\in\setnodes$.
	Part (i) of this \namecref{thm:inflow-leq} gives $\aflowconstraint_1.\fval(\anode)\leq\aflowconstraint_2.\fval(\anode)$.
	Then, we get $\edgesatof{\anode}{\anodep}{\aflowconstraint_1.\fval(\anode)}\leq\edgesatof{\anode}{\anodep}{\aflowconstraint_2.\fval(\anode)}$, because $\edges$ is $\leq$-continuous and thus $\leq$-monotonic by \Cref{thm:continuous-implies-monotonic-special}.
	Hence, \Cref{thm:our-monoid-has-addprop} yields the desired $\transformerofof{\aflowconstraint}{\inflow_1}(\anodep)\leq\transformerofof{\aflowconstraint}{\inflow_2}(\anodep)$.
\end{proof}

\begin{proof}[Proof of \Cref{thm:inflow-fprel}(i)]
	Let $\aflowconstraint=(\setnodes,\edges,\dontcare)$.
	Define $\aflowconstraint_1=\aflowconstraint[\inflow\mapsto\inflow_1]$ and $\aflowconstraint_2=\aflowconstraint[\inflow\mapsto\inflow_2]$.
	By \Cref{thm:flow-as-lfp}, the flow in $\aflowconstraint_i$ is $\aflowconstraint_i.\fval=\bigjoin\setcond{f_i^j(\bot)}{j\in\nat}$ with $f_i:(\setnodes{\to}\amonoid)\to(\setnodes{\to}\amonoid)$ defined by: \[
		f_i(\cval)(\anode) ~~=~~ \sum_{\anodep\in\nat\setminus\setnodes}\inflow_i(\anodep,\anode) + \sum_{\anodep\in\setnodes}\edgesatof{\anodep}{\anode}{\cval(\anodep)}
		\ .
	\]
	To conclude, it suffices to show that $f_1^j(\bot) \fpreldot f_2^j(\bot)$ holds for all $j\in\nat$.
	We proceed by induction.
	In the base case, $f_1^0(\bot)=\bot\fpreldot\bot=f_2^0$ because $\bot\fprel\bot$ by \eqref{def:fpcompat:refl-trans} of $\fpcompatible[\fprel]{\aflowconstraint}$.
	For the induction step, we show $f_1(f_1^j(\bot)) \fpreldot f_2(f_2^j(\bot)$.
	By induction together with $\fpreldot$-monotonicity of $f_1$ by \Cref{thm:flow-iteration-monotonic-continuous}, we have $f_1(f_1^j(\bot)) \fpreldot f_1(f_2^j(\bot))$.
	Because $\fprel$ is transitive by \eqref{def:fpcompat:refl-trans} of $\fpcompatible[\fprel]{\aflowconstraint}$, it remains to show $f_1(\cval) \fpreldot f_2(\cval)$ with $\cval=f_2^j(\bot)$.
	Assume for a moment that we have $\sum_{\anodep\in\nat\setminus\setnodes}\inflow_1(\anodep,\anode)\fprel\sum_{\anodep\in\nat\setminus\setnodes}\inflow_2(\anodep,\anode)$, for all $\anode\in\setnodes$.
	Then, we immediately get $f_1(\cval) \fpreldot f_2(\cval)$ by \eqref{def:fpcompat:add} of $\fpcompatible[\fprel]{\aflowconstraint}$ and \Cref{thm:add-relations}.
	To see the correspondence, choose $\setnodesp_1,\setnodesp_2$ such that $\setnodesp_1\uplus\setnodesp_2=\nat\setminus\setnodes$ and $\inflow_1(\anodep,\anode)=\inflow_1(\anodep,\anode)$ and $\inflow_1(\anodepp,\anode)\neq\inflow_1(\anodepp,\anode)$ for all $\anode\in\setnodes,\anodep\in\setnodesp_1,\anodepp\in\setnodesp_2$.
	We have $\sum_{\anodep\in\nat\setnodesp_1}\inflow_1(\anodep,\anode)=\sum_{\anodep\in\setnodes_1}\inflow_2(\anodep,\anode)$.
	Because $\inflow_1\fpreleq\inflow_2$, we must have $\inflow_1(\anodepp,\anode)\fprel\inflow_1(\anodepp,\anode)$.
	Hence, $\sum_{\anodepp\in\nat\setnodesp_2}\inflow_1(\anodep,\anode)=\sum_{\anodep\in\setnodes_2}\inflow_2(\anodepp,\anode)$ by \eqref{def:fpcompat:add} of $\fpcompatible[\fprel]{\aflowconstraint}$ and \Cref{thm:add-relations}.
	Again by \eqref{def:fpcompat:add}, adding the two sums maintains $\fprel$.
	That is, $\sum_{\anodep\in\nat\setminus\setnodes}\inflow_1(\anodep,\anode)\fprel\sum_{\anodep\in\nat\setminus\setnodes}\inflow_2(\anodep,\anode)$ holds, as required.
	This concludes the induction.
\end{proof}

\begin{proof}[Proof of \Cref{thm:inflow-fprel}(ii)]
	Let $\aflowconstraint=(\setnodes,\edges,\dontcare)$.
	Define $\aflowconstraint_1=\aflowconstraint[\inflow\mapsto\inflow_1]$ and $\aflowconstraint_2=\aflowconstraint[\inflow\mapsto\inflow_2]$.
	Let $\anodep\in\nat\setminus\setnodes$.
	We have $\transformerofof{\aflowconstraint}{\inflow_i}(\anode)=\sum_{\anode\in\setnodes}\edgesatof{\anode}{\anodep}{\aflowconstraint_i.\fval(\anode)}$ by the definition of $\transformerof{\dontcare}$ and the outflow.
	Let $\anode\in\setnodes$.
	Part (i) of this \namecref{thm:inflow-fprel} gives $\aflowconstraint_1.\fval(\anode)\fprel\aflowconstraint_2.\fval(\anode)$.
	Then, we get $\edgesatof{\anode}{\anodep}{\aflowconstraint_1.\fval(\anode)}\fprel\edgesatof{\anode}{\anodep}{\aflowconstraint_2.\fval(\anode)}$, because $\edges$ is $\fpreldot$-monotonic by \eqref{def:fpcompat:edges} of $\fpcompatible[\fprel]{\aflowconstraint}$ from the assumption.
	Hence, \Cref{thm:add-relations} yields the desired $\transformerofof{\aflowconstraint}{\inflow_1}(\anodep)\fprel\transformerofof{\aflowconstraint}{\inflow_2}(\anodep)$
\end{proof}

\begin{proof}[Proof of \Cref{thm:partial-fixpoint}]
	Follows from \cite[Proof of Theorem~1]{DBLP:conf/tacas/MeyerWW23,DBLP:journals/corr/abs-2304-04886}.
	We repeat the argument and adapt the proof to our use case.

	\newcommand{\F}{\alpha}
	\newcommand{\FF}{\alpha^{\dagger}}
	\newcommand{\G}{\beta}
	\newcommand{\dval}{\mathit{dval}}
	\newcommand{\theinflow}{\overline{\inflow}}
	Let $\setnodes_1\defeq\aflowconstraint_1.\setnodes$, $\setnodes_2\defeq\aflowconstraint_2.\setnodes$, and $\overline{\setnodes}\defeq\nat\setminus(\setnodes_1\cup\setnodes_2)$.
	Further, let $\theinflow\defeq(\aflowconstraint_1\statemult\aflowconstraint_2).\inflow$.
	To apply \Bekic~\cite{DBLP:conf/ibm/Bekic84e}, define the target pairing of two functions
	\begin{align*}
		\F~:~A\times B\to A
		\qquad\text{and}\qquad
		\G~:~A\times B\to B
	\end{align*}
	over the same domain $A \times B$ as the function
	\begin{align*}
	 	\pairingof{\F}{\G}~:~A\times B \to A\times B
	 	\qquad\text{with}\qquad
	 	\pairingof{\F}{\G}(a) ~\defeq~ (f(a),\, g(a))
	 	\ .
	\end{align*}
	We compute the flow of $\aflowconstraint_1\uplus\aflowconstraint_2$ as the least fixed point of a target pairing $\pairingof{f}{g}$ with
	\begin{align*}
		\F ~:&~\, ((\setnodes_1\uplus\setnodes_2) \to \amonoid) \to (\setnodes_1 \to \amonoid)
		\\\text{and}\qquad
		\G ~:&~\, ((\setnodes_1\uplus\setnodes_2) \to \amonoid) \to (\setnodes_2 \to \amonoid)
		\ .
	\end{align*}
	Function $\F$ updates the flow of the nodes $\setnodes_1$ in $\aflowconstraint_1$ depending on the flow in/inflow from $\aflowconstraint_2$.
	Function $\G$ is responsible for the flow of the nodes $\setnodes_2$ in $\aflowconstraint_2$.
	The inflow from the nodes outside $\aflowconstraint_1\uplus\aflowconstraint_2$ is constant, $\theinflow$.
	Concretely, we define $\F$ and $\G$ along the lines of the flow equation \eqref{def:flow-equation}:
	\begin{align*}
		\F(\dval)(\anode) ~\defeq~
			\sum_{\anodepp\in\overline{\setnodes}} \theinflow(\anodepp,\anode)
			+
			\sum_{\anodep\in\setnodes_1\cup\setnodes_2} (\aflowconstraint_1\uplus\aflowconstraint_2).\edgesatof{\anodep}{\anode}{\dval(\anodep)}
		\\\text{and}\qquad
		\G(\dval)(\anodep) ~\defeq~
			\sum_{\anodepp\in\overline{\setnodes}} \theinflow(\anodepp,\anodep)
			+
			\sum_{\anode\in\setnodes_1\cup\setnodes_2} (\aflowconstraint_1\uplus\aflowconstraint_2).\edgesatof{\anode}{\anodep}{\dval(\anode)}
	\end{align*}
	From \Cref{thm:flow-iteration-monotonic-continuous} for $\aflowconstraint_1\statemult\aflowconstraint_2$ we get that $\F$ and $\G$ are $\leq$-monotonic and $\leq$-continuous.
	Moreover, the \namecref{thm:flow-iteration-monotonic-continuous} also gives that $\F$ and $\G$ are $\fpreldot$-monotonic, provided $\fpcompatible[\fprel]{\aflowconstraint_1\statemult\aflowconstraint_2}$ holds.
	Observe that the above definitions guarantee:
	\begin{align*}
		(\aflowconstraint_1\statemult\aflowconstraint_2).\fval ~=~ \lfpof{\pairingof{\F}{\G}}
		\ .
	\end{align*}

	We curry function $\F$ and obtain:
	\begin{align*}
		\F ~:~ (\setnodes_2 \to \amonoid) \to (\setnodes_1 \to \amonoid) \to (\setnodes_1 \to \amonoid)
		\ .
	\end{align*}
	This gives rise to the following function:
	\begin{align*}
		\F(\cval) ~:\,&~~ (\setnodes_1 \to \amonoid) \to (\setnodes_1 \to \amonoid)
		\qquad\text{for all}\qquad
		\cval ~:~ \setnodes_2 \to \amonoid
		\\\text{with}\qquad
		\F(\cval)(\dval)(\anode) ~=&~
			\sum_{\anodepp\in\overline{\setnodes}} \theinflow(\anodepp,\anode)
			+
			\sum_{\anodep\in\setnodes_1\cup\setnodes_2} (\aflowconstraint_1\uplus\aflowconstraint_2).\edgesatof{\anodep}{\anode}{(\cval\uplus\dval)(\anodep)}
		\ .
	\end{align*}
	This function is still $\leq$-monotonic and $\leq$-continuous as well as $\fpreldot$-monotonic if $\fpcompatible[\fprel]{\aflowconstraint_1\statemult\aflowconstraint_2}$.
	Therefore, it has a least fixed point by \Cref{thm:fixpoint-kleene} so that this is well-defined:
	\begin{align*}
		\FF ~:~ (\setnodes_2 \to \amonoid) \to (\setnodes_1 \to \amonoid)
		\qquad\text{with}\qquad
		\FF(\cval) ~\defeq~ \lfpof{\F(\cval)}
		\ .
	\end{align*}

	Towards an application of \Bekic, we define
	\begin{align*}
		\pairingof{\FF}{\myid} ~:~ (\setnodes_2\to\amonoid) \to ((\setnodes_1\uplus\setnodes_2)\to\amonoid)
		\qquad\text{with}\qquad
		\myid ~:~ (\setnodes_2\to\amonoid) \to (\setnodes_2\to\amonoid)
	\end{align*}
	Now, compose this function with $\G$.
	This yields:
	\begin{align*}
		\G \circ \pairingof{\FF}{\myid} ~:~ (\setnodes_2\to\amonoid) \to (\setnodes_2\to\amonoid)
		\ .
	\end{align*}
	Now, \Bekic guarantees the correctness of the following least fixed point:
	\begin{align}
		\lfpof{\pairingof{\F}{\G}} ~=~ (\FF(\cval), \cval)
		\qquad\text{with}\qquad
		\cval ~=~ \lfpof{\G \circ \pairingof{\FF}{\myid}}
		\label{proof:partial-fixpoint:bekic-fixed-point}
		\ .
	\end{align}

	\medskip
	By $\aflowconstraint_1\statemultdef\aflowconstraint_2$ we have $(\aflowconstraint_1\statemult\aflowconstraint_2).\fval=\aflowconstraint_1.\fval\uplus\aflowconstraint_2.\fval$.
	Combined with \eqref{proof:partial-fixpoint:bekic-fixed-point} this yields:
	\begin{align*}
		\aflowconstraint_2.\fval ~=~ \lfpof{\G \circ \pairingof{\FF}{\myid}}
		\ .
	\end{align*}
	We show that $\G \circ \pairingof{\FF}{\myid}$ is equivalent to $f$.
	To that end, rewrite the curried version of $\F$:
	\begin{align*}
		&\F(\cval)(\dval)(\anode)
		\\\explain{Def. $\alpha(\cval)$}=~&
		\sum_{\anodepp\in\overline{\setnodes}} \theinflow(\anodepp,\anode)
		+
		\sum_{\anodep\in\setnodes_1\cup\setnodes_2} (\aflowconstraint_1\uplus\aflowconstraint_2).\edgesatof{\anodep}{\anode}{(\cval\uplus\dval)(\anodep)}
		\\\explain{Def. $\aflowconstraint_1\statemult\aflowconstraint_2$}=~&
		\sum_{\anodepp\in\overline{\setnodes}} \theinflow(\anodepp,\anode)
		+
		\sum_{\anodep\in\setnodes_2} \aflowconstraint_2.\edgesatof{\anodep}{\anode}{\cval(\anodep)}
		+
		\sum_{\anodep\in\setnodes_1} \aflowconstraint_1.\edgesatof{\anodep}{\anode}{\dval(\anodep)}
		\\\explain{Def. $\inflow_\cval$}=~&
		\sum_{\anodep\in\setnodes_2\cup\overline{\setnodes}} \inflow_\cval(\anodep,\anode)
		+
		\sum_{\anodep\in\setnodes_1} \aflowconstraint_1.\edgesatof{\anodep}{\anode}{\dval(\anodep)}
	\end{align*}
	Note that the last sum is the flow equation for $\aflowconstraint_1$ with its inflow updated to $\inflow_\cval$.
	By definition, this means:
	\begin{align}
		\FF(\cval)
		~=~
		\lfpof{\F(\cval)}
		~=~
		\aflowconstraint_1[\inflow\mapsto\inflow_\cval].\fval
		\label{proof:partial-fixpoint:flow-in-h1}
		\ .
	\end{align}
	With this, we conclude:
	\begin{align*}
		&(\G \circ \pairingof{\FF}{\myid})(\cval)(\anodep)
		\\\explain{Def. $\circ$}\!\!=~&
		\G(\pairingof{\FF}{\myid}(\cval))(\anodep)
		\\\explain{Def. $\G$}\!\!=~&
		\sum_{\anodepp\in\overline{\setnodes}} \theinflow(\anodepp,\anodep)
		+
		\sum_{\anode\in\setnodes_1\cup\setnodes_2} (\aflowconstraint_1\uplus\aflowconstraint_2).\edgesatof{\anode}{\anodep}{\pairingof{\FF}{\myid}(\cval)(\anode)}
		\\\explain{by $\aflowconstraint_1\statemultdef\aflowconstraint_2$}\!\!=~&
		\sum_{\anodepp\in\overline{\setnodes}} \theinflow(\anodepp,\anodep)
		+
		\sum_{\anode\in\setnodes_1} \aflowconstraint_1.\edgesatof{\anode}{\anodep}{\pairingof{\FF}{\myid}(\cval)(\anode)}
		+
		\sum_{\anode\in\setnodes_2} \aflowconstraint_2.\edgesatof{\anode}{\anodep}{\pairingof{\FF}{\myid}(\cval)(\anode)}
		\\\explain{Def. $\pairingof{\FF}{\myid}$}\!\!=~&
		\sum_{\anodepp\in\overline{\setnodes}} \theinflow(\anodepp,\anodep)
		+
		\sum_{\anode\in\setnodes_1} \aflowconstraint_1.\edgesatof{\anode}{\anodep}{\FF(\cval)(\anode)}
		+
		\sum_{\anode\in\setnodes_2} \aflowconstraint_2.\edgesatof{\anode}{\anodep}{\myid(\cval)(\anode)}
		\\\explain{by \eqref{proof:partial-fixpoint:flow-in-h1}}\!\!=~&
		\sum_{\anodepp\in\overline{\setnodes}} \theinflow(\anodepp,\anodep)
		+
		\sum_{\anode\in\setnodes_1} \aflowconstraint_1.\edgesatof{\anode}{\anodep}{\aflowconstraint_1[\inflow\mapsto\inflow_\cval].\fval(\anode)}
		+
		\sum_{\anode\in\setnodes_2} \aflowconstraint_2.\edgesatof{\anode}{\anodep}{\cval(\anode)}
		\\\explain{Def. $\outflow$}\!\!=~&
		\sum_{\anodepp\in\overline{\setnodes}} \theinflow(\anodepp,\anodep)
		+
		\sum_{\anode\in\setnodes_1} \aflowconstraint_1[\inflow\mapsto\inflow_\cval].\outflow(\anode,\anodep)
		+
		\sum_{\anode\in\setnodes_2} \aflowconstraint_2.\edgesatof{\anode}{\anodep}{\cval(\anode)}
		\\\explain{Def. $\transformerof{\dontcare}$}\!\!=~&
		\sum_{\anodepp\in\overline{\setnodes}} \theinflow(\anodepp,\anodep)
		+
		\transformerofof{\aflowconstraint_1}{\inflow_\cval}(\anodep)
		+
		\sum_{\anode\in\setnodes_2} \aflowconstraint_2.\edgesatof{\anode}{\anodep}{\cval(\anode)}
		\\\explain{Def. $f$}\!\!=~&
		f(\cval)(\anodep)
	\end{align*}
	Overall, we arrive at:
	\begin{align*}
		\aflowconstraint_2.\fval
		~=~
		\lfpof{\G \circ \pairingof{\FF}{\myid}}
		~=~
		\lfpof{f}
	\end{align*}
	Finally, $f$ is $\leq$-monotonic and $\leq$-continuous because $\FF,\G,\myid$ are.
	Moreover, $f$ is $\fpreldot$-monotonic because $\FF,\G,\myid$ are, provided $\fpcompatible[\fprel]{\aflowconstraint_1\statemult\aflowconstraint_2}$.
	This concludes the proof.
\end{proof}

\begin{proof}[Proof of \Cref{thm:upward-closed-outflow}]
	\newcommand{\theinflow}{\overline{\inflow}}
	\newcommand{\thefpin}[1]{\inflow^{\smash{\setnodes_1}}_{#1}}
	We unroll the premise for flow graphs $\aflowconstraint_1,\aflowconstraint_2,\aflowconstraint_F$:
	\begin{compactenum}[(A)]
		\item\label{proof:upward-closed-outflow:multdef}
			$\aflowconstraint_1\statemultdef \aflowconstraint_F$
		\item\label{proof:upward-closed-outflow:ctxfprel}
			$\aflowconstraint_1 \ctxfprel \aflowconstraint_2$
		\item\label{proof:upward-closed-outflow:fprel-compat}
			$\fpcompatible[\fprel]{\aflowconstraint_1,\aflowconstraint_2,\aflowconstraint_F}$, i.e., \eqref{def:fpcompat:refl-trans}, \eqref{def:fpcompat:add}, \eqref{def:fpcompat:edges}, and \eqref{def:fpcompat:lfp} from \Cref{def:fpcompat}
	\end{compactenum}

	\medskip
	Let $\setnodes_1 \defeq \aflowconstraint_1.\setnodes$ and $\setnodes_F \defeq \aflowconstraint_F.\setnodes$ and $\overline{\setnodes} \defeq \nat\setminus(\setnodes_1 \cup \setnodes_F)$ and $\theinflow \defeq (\aflowconstraint_1 \statemult \aflowconstraint_F).\inflow$.
	By \eqref{proof:upward-closed-outflow:ctxfprel}, $\setnodes_1 = \aflowconstraint_2.\setnodes$.
	Define the flow graph $\aflowconstraint_{2+F}$ by \[
		\aflowconstraint_{2+F} ~\defeq~ \bigl(~
			\setnodes_1 \uplus \setnodes_F,~~
			\aflowconstraint_2.\edges \uplus \aflowconstraint_F.\edges,~~
			\theinflow
		~\bigr)
		\ .
	\]
	Note that $\aflowconstraint_{2+F}$ is well defined because $\setnodes_1 \cap \setnodes_F = \emptyset$ by \eqref{proof:upward-closed-outflow:ctxfprel}.
	Now, choose \[
		\aflowconstraint_2'=\restrictto{\aflowconstraint_{2+F}}{\setnodes_2}
		\qquad\text{and}\qquad
		\aflowconstraint_F'=\restrictto{\aflowconstraint_{2+F}}{\setnodes_F}
		\ .
	\]
	By definition, $\aflowconstraint_2'.\setnodes=\setnodes_1$ and $\aflowconstraint_F'.\setnodes=\setnodes_F$.
	Moreover, $\aflowconstraint_2'.\edges=\aflowconstraint_2.\edges$ and $\aflowconstraint_F'.\edges=\aflowconstraint_F.\edges$.
	\Cref{thm:restriction-vs-statemult} gives both $\aflowconstraint_2'\statemultdef\aflowconstraint_F'$ and $\aflowconstraint_2'\statemult\aflowconstraint_F'=\aflowconstraint_{2+F}$.
	It is easy to see $\fpcompatible[\fprel]{\aflowconstraint_2',\aflowconstraint_F',\aflowconstraint_{2+F}}$.
	
	\medskip
	We derive the flow in $\aflowconstraint_F$ and $\aflowconstraint_F'$ as a fixed point relative to the inflow provided by the transformers of $\aflowconstraint_1$ and $\aflowconstraint_2'$, respectively.
	Concretely, invoke \Cref{thm:partial-fixpoint} for $\aflowconstraint_F$ in $\aflowconstraint_1\statemult\aflowconstraint_F$ and $\aflowconstraint_F'$ in $\aflowconstraint_2'\statemult\aflowconstraint_F'$, yielding:
	\begin{align}
		\aflowconstraint_F.\fval ~=~ \lfpof{f_1} ~=&~~ \bigjoin\setcond{f_1^i(\bot)}{i\in\nat}
		\label{proof:upward-closed-outflow:flow-fp-pre}
		\\\text{and}\quad
		\aflowconstraint_F'.\fval ~=~ \lfpof{f_2} ~=&~~ \bigjoin\setcond{f_2^i(\bot)}{i\in\nat}
		\label{proof:upward-closed-outflow:flow-fp-post}
	\end{align}
	with $f_j : (\setnodes_F \to \amonoid) \to (\setnodes_F \to \amonoid)$ and $\thefpin{\cval} : (\nat\setminus\setnodes_1) \times \setnodes_1 \to \amonoid$ defined by:
	\begin{align}
		f_j(\cval)(\anode) ~&\defeq~
			\sum_{\anodepp\in\overline{\setnodes}} \theinflow(\anodepp,\anode)
			+
			\transformerofof{\aflowconstraint_j}{\thefpin{\cval}}(\anode)
			+
			\sum_{\anodep\in\setnodes_F} \aflowconstraint_F.\edgesatof{\anodep}{\anode}{\cval(\anodep)}
			\label{proof:upward-closed-outflow:frame-flow-step}
		\\
		\thefpin{\cval}(\anode,\anodep) ~&\defeq~
			\anode\in\setnodes_F ~~~?~~~
			\aflowconstraint_F.\edgesatof{\anodep}{\anode}{\cval(\anodep)}
			~~:~~
			\theinflow(\anode,\anodep)
			\label{proof:upward-closed-outflow:fp-inflow-def}
	\end{align}
	Intuitively, $\thefpin{\cval}$ is the inflow at nodes $\setnodes_1$ given the flow values $\cval$ for the nodes in $\setnodes_F$.
	That is, it is the sum of $\theinflow$ plus the flow received from $\setnodes_F$.
	Hence, by \eqref{proof:upward-closed-outflow:multdef} and choice $\aflowconstraint_2'$ we have:
	\begin{align}
		\thefpin{\aflowconstraint_F.\fval} ~=~ \aflowconstraint_1.\inflow
		\qquad\text{and}\qquad
		\thefpin{\aflowconstraint_F'.\fval} ~=~ \aflowconstraint_2'.\inflow
		\label{proof:upward-closed-outflow:inflow-frame}
		\ .
	\end{align}
	\Cref{thm:partial-fixpoint} also provides the following properties for $f_1$ and $f_2$:
	\begin{compactenum}[(A)]
		\setcounter{enumi}{3}
		\item both $f_1$ and $f_2$ are $\leq$-monotonic and $\leq$-continuous, and
			\label{proof:upward-closed-outflow:f-leq-monotonic}
		\item
			both $f_1$ and $f_2$ are $\fpreldot$-monotonic.
			\label{proof:upward-closed-outflow:f-fprel-monotonic}
	\end{compactenum}

	\medskip
	We now show $f_1^i(\cval) \fpreldot f_2^i(\cval)$ for all $\cval \leq \aflowconstraint_F.\fval$ and all $i\in\nat$.
	The claim is true for $i=0$ by \eqref{def:fpcompat:refl-trans}.
	For $i\geq 1$, we proceed by induction.
	For the base case, $i=1$, observe $\thefpin{\cval} \leq \thefpin{\aflowconstraint_F.\fval}$ by \eqref{proof:upward-closed-outflow:fp-inflow-def} together with the fact that all edge functions in $\aflowconstraint_F.\edges$ are $\leq$-continuous and thus $\leq$-monotonic by \Cref{thm:continuous-implies-monotonic-special}.
	Then, \eqref{proof:upward-closed-outflow:inflow-frame} gives $\thefpin{\cval} \leq \aflowconstraint_1.\inflow$.
	By \eqref{proof:upward-closed-outflow:ctxfprel}, this yields $\transformerofof{\aflowconstraint_1}{\thefpin{\cval}} \fpreldot \transformerofof{\aflowconstraint_2}{\thefpin{\cval}}$.
	Since this is the only part that differs in $f_1$ and $f_2$ according to \eqref{proof:upward-closed-outflow:frame-flow-step}, we obtain the desired $f_1(\cval) \fpreldot f_2(\cval)$ by \eqref{def:fpcompat:add} and \Cref{thm:add-relations}.
	(Note: should the sum be empty, we obtain the desired approximation by $\bot\fprel\bot$ from \eqref{def:fpcompat:refl-trans}.)
	For the induction step, we have $f_1^i(\cval') \fpreldot f_2^i(\cval')$ for all $\cval' \leq \aflowconstraint_F.\fval$.
	We show that $f_1^{i+1}(\cval) \fpreldot f_2^{i+1}(\cval)$ holds for all $\cval \leq \aflowconstraint_F.\fval$.
	By \eqref{proof:upward-closed-outflow:f-leq-monotonic} combined with \eqref{proof:upward-closed-outflow:flow-fp-pre}, we have $f_1(\cval) \leq f_1(\aflowconstraint_F.\fval) = \aflowconstraint_F.\fval$.
	Then, by induction, we obtain $f_1^i(f_1(\cval)) \fpreldot f_2^i(f_1(\cval))$.
	We already showed (for the base case), that $f_1(\cval) \fpreldot f_2(\cval)$ holds.
	Hence, $f_2^i(f_1(\cval)) \fpreldot f_2^i(f_2(\cval))$ because $f_2$ is $\fpreldot$-monotonic by \eqref{proof:upward-closed-outflow:f-fprel-monotonic}.
	By transitivity of $\fprel$ from \eqref{def:fpcompat:refl-trans}, we get $f_1^i(f_1(\cval)) \fpreldot f_2^i(f_2(\cval))$.
	Altogether, this concludes the induction and proves:
	\begin{align}
		\forall\, i\in\nat
		~~
		\forall\, \cval \leq \aflowconstraint_F.\fval
		. \quad
		f_1^i(\cval) ~\fpreldot~ f_2^i(\cval)
	 	\label{proof:upward-closed-outflow:flow-iteration-approx}
		\ .
	\end{align}
	Now, we invoke \Cref{thm:fixpoint-kleene} for $f_1$ and $f_2$.
	This yields the fixed points of $f_1$ and $f_2$ as the joins $\lfpof{f_1}=\bigjoin\achain$ and $\lfpof{f_2}=\bigjoin\achainp$ over the $\leq$-ascending chains $\achain \defeq f_1^0(\bot)\leq f_1^2(\bot)\leq\cdots$ and $\achainp \defeq f_2^0(\bot)\leq f_2^2(\bot)\leq\cdots$, respectively.
	We now obtain $\bigjoin\achain = f_1(\bigjoin\achain) \fpreldot f_2(\bigjoin\achain)$ from combining \eqref{proof:upward-closed-outflow:flow-fp-pre} and \eqref{proof:upward-closed-outflow:flow-iteration-approx}.
	This together with \eqref{proof:upward-closed-outflow:f-leq-monotonic}, \eqref{proof:upward-closed-outflow:f-fprel-monotonic}, and \eqref{proof:upward-closed-outflow:flow-iteration-approx}, establishes the premise of \eqref{def:fpcompat:lfp}.
	Applying the property for $\bigjoin\achain$ and $\bigjoin\achainp$ yields:
	\begin{align}
		\aflowconstraint_F.\fval
		~~\explains{\eqref{proof:upward-closed-outflow:flow-fp-pre}}{=}~~
		\lfpof{f_1}
		~~=~~
		\bigjoin\achain
		~~\explains{\eqref{def:fpcompat:lfp}}{\fpreldot}~~
		\bigjoin\achainp
		~~=~~
		\lfpof{f_2}
		~~\explains{\eqref{proof:upward-closed-outflow:flow-fp-post}}{=}~~
		\aflowconstraint_F'.\fval
		\label{proof:upward-closed-outflow:frame-flow-approx}
		\ .
	\end{align}

	\medskip
	Next, we show $\aflowconstraint_2.\inflow \fpreldot \aflowconstraint_2'.\inflow$.
	To that end, consider some $\anode\in\setnodes_2$, $\anodep\in\setnodes_F$, and $\anodepp\in\overline{\setnodes}$.
	By choice, $\aflowconstraint_1.\inflow(\anodepp,\anode)=\theinflow(\anodepp,\anode)=\aflowconstraint_2'.\inflow(\anodepp,\anode)$.
	Hence, $\aflowconstraint_2.\inflow(\anodepp,\anode)=\aflowconstraint_2'.\inflow(\anodepp,\anode)$ because \eqref{proof:upward-closed-outflow:ctxfprel} gives $\aflowconstraint_1.\inflow=\aflowconstraint_2.\inflow$.
	It remains to consider the inflow at $\anode$ from $\anodep$:
	\begin{align*}
		&\aflowconstraint_2.\inflow(\anodep,\anode)
		\\\explain{by definition}=~~&
		\aflowconstraint_F.\edgesatof{\anodep}{\anode}{\aflowconstraint_F.\fval(\anodep)}
		\\\explain{by choice of $\aflowconstraint_F'$}=~~&
		\aflowconstraint_F'.\edgesatof{\anodep}{\anode}{\aflowconstraint_F.\fval(\anodep)}
		\\\explain{by \eqref{proof:upward-closed-outflow:frame-flow-approx} and \eqref{def:fpcompat:edges}}\fprel~~&
		\aflowconstraint_F'.\edgesatof{\anodep}{\anode}{\aflowconstraint_F'.\fval(\anodep)}
		\\\explain{by definition}=~~&
		\aflowconstraint_2'.\inflow(\anodep,\anode)
		\ .
	\end{align*}
	Combining the above, we obtain:
	\begin{align}
		\aflowconstraint_2.\inflow ~\fpreleq~ \aflowconstraint_2'.\inflow
		\qquad\text{and}\qquad
		\restrictto{\aflowconstraint_2.\inflow}{\overline{\setnodes}\times\setnodes_1} ~=~ \restrictto{\aflowconstraint_2'.\inflow}{\overline{\setnodes}\times\setnodes_1}
		\label{proof:upward-closed-outflow:footprint-inflow-approx}
		\ .
	\end{align}
	Recall that $\aflowconstraint_2$ and $\aflowconstraint_2'$ differ only in the inflow.
	Hence, \eqref{proof:upward-closed-outflow:footprint-inflow-approx} together \Cref{thm:inflow-fprel} yields
	\begin{align}
		\aflowconstraint_2.\fval ~\fpreldot~ \aflowconstraint_2'.\fval
		\label{proof:upward-closed-outflow:footprint-flow-approx}
		\ .
	\end{align}

	\medskip
	Now, we are ready to show that $(\aflowconstraint_1\statemult\aflowconstraint_F).\outflow \fpreldot \aflowconstraint_{2+F}.\outflow$ holds.
	As a first step, we conclude the following relation among the transformers of $\aflowconstraint_1$ and $\aflowconstraint_2'$:
	\begin{align*}
		&\transformerofof{\aflowconstraint_1}{\aflowconstraint_1.\inflow}
		\\\explain{by \eqref{proof:upward-closed-outflow:ctxfprel}}\fpreldot~~&
		\transformerofof{\aflowconstraint_2}{\aflowconstraint_1.\inflow}
		\\\explain{by \eqref{proof:upward-closed-outflow:inflow-frame}}=~~&
		\transformerofof{\aflowconstraint_2}{\thefpin{\aflowconstraint_F.\fval}}
		\\\explain{see below}\fpreldot~~&
		\transformerofof{\aflowconstraint_2}{\thefpin{\aflowconstraint_F'.\fval}}
		\\\explain{by \eqref{proof:upward-closed-outflow:inflow-frame}}=~~&
		\transformerofof{\aflowconstraint_2}{\aflowconstraint_2'.\inflow}
	\end{align*}
	where $\transformerofof{\aflowconstraint_2}{\thefpin{\aflowconstraint_F.\fval}} \fpreldot \transformerofof{\aflowconstraint_2}{\thefpin{\aflowconstraint_F'.\fval}}$ holds because \eqref{proof:upward-closed-outflow:frame-flow-approx} gives $\aflowconstraint_F.\fval\fpreldot\aflowconstraint_F'.\fval$ which means $\thefpin{\aflowconstraint_F.\fval}\fpreldot\thefpin{\aflowconstraint_F'.\fval}$ because edges functions are $\fpreldot$-monotonic by \eqref{def:fpcompat:edges} and thus an application of \Cref{thm:outflow-vs-statemult} yields the desired property.
	By the choice of $\aflowconstraint_2'$, we obtain:
	\begin{align}
		\transformerofof{\aflowconstraint_1}{\aflowconstraint_1.\inflow}
		~\fpreldot~
		\transformerofof{\aflowconstraint_2'}{\aflowconstraint_2'.\inflow}
		\label{proof:upward-closed-outflow:footprint-inflow-transformer-pre-vs-post}
		\ .
	\end{align}
	As a second step, we conclude the following relation among the transformers of $\aflowconstraint_F$ and $\aflowconstraint_F'$, for some node $\anodepp\notin\setnodes_F$:
	\begin{align*}
		&\transformerofof{\aflowconstraint_F}{\aflowconstraint_F.\inflow}(\anodepp)
		\\\explain{Def. $\transformerof{\dontcare}$}=~~&
		\sum_{\anodep\in\setnodes_F} \aflowconstraint_F.\outflow(\anodep,\anodepp)
		\\\explain{by definition}=~~&
		\sum_{\anodep\in\setnodes_F} \aflowconstraint_F.\edgesatof{\anodep}{\anodepp}{\aflowconstraint_F.\fval(\anodep)}
		\\\explain{by \eqref{proof:upward-closed-outflow:frame-flow-approx} and \eqref{def:fpcompat:edges} and \Cref{thm:add-relations}; if empyt sum, by \eqref{def:fpcompat:refl-trans}}\fprel~~&
		\sum_{\anodep\in\setnodes_F} \aflowconstraint_F.\edgesatof{\anodep}{\anodepp}{\aflowconstraint_F'.\fval(\anodep)}
		\\\explain{by choice of $\aflowconstraint_F'$}=~~&
		\sum_{\anodep\in\setnodes_F'} \aflowconstraint_F'.\edgesatof{\anodep}{\anodepp}{\aflowconstraint_F'.\fval(\anodep)}
		\\\explain{by definition}=~~&
		\sum_{\anodep\in\setnodes_F'} \aflowconstraint_F'.\outflow(\anodep,\anodepp)
		\\\explain{Def. $\transformerof{\dontcare}$}=~~&
		\transformerofof{\aflowconstraint_F'}{\aflowconstraint_F'.\inflow}(\anodepp)
	\end{align*}
	That is,
	\begin{align}
		\transformerofof{\aflowconstraint_F}{\aflowconstraint_F.\inflow}
		~\fpreldot~
		\transformerofof{\aflowconstraint_F'}{\aflowconstraint_F'.\inflow}
		\label{proof:upward-closed-outflow:frame-inflow-transformer-pre-vs-post}
		\ .
	\end{align}
	Using the above, we arrive at the following, for some node $\anodepp\in\overline{\setnodes}$:
	\begin{align*}
		&\transformerofof{\aflowconstraint_1\statemult\aflowconstraint_F}{\theinflow}(\anodepp)
		\\\explain{by choice $\theinflow$}=~~&
		\transformerofof{\aflowconstraint_1\statemult\aflowconstraint_F}{(\aflowconstraint_1\statemult\aflowconstraint_F).\inflow}(\anodepp)
		\\\explain{by \Cref{thm:transformer-vs-statemult}}=~~&
		\transformerofof{\aflowconstraint_1}{\aflowconstraint_1.\inflow}(\anodepp)
		+\transformerofof{\aflowconstraint_F}{\aflowconstraint_F.\inflow}(\anodepp)
		\\\explain{by \eqref{proof:upward-closed-outflow:footprint-inflow-transformer-pre-vs-post} and \eqref{proof:upward-closed-outflow:frame-inflow-transformer-pre-vs-post} and \Cref{thm:add-relations}}\fpreldot~~&
		\transformerofof{\aflowconstraint_2'}{\aflowconstraint_2'.\inflow}(\anodepp)
		+\transformerofof{\aflowconstraint_F'}{\aflowconstraint_F'.\inflow}(\anodepp)
		\\\explain{by \Cref{thm:transformer-vs-statemult}}=~~&
		\transformerofof{\aflowconstraint_2'\statemult\aflowconstraint_F'}{(\aflowconstraint_2'\statemult\aflowconstraint_F').\inflow}(\anodepp)
		\\\explain{by choice $\aflowconstraint_{2+F},\aflowconstraint_2',\aflowconstraint_F$}=~~&
		\transformerofof{\aflowconstraint_{2+F}}{\theinflow}(\anodepp)
	\end{align*}
	We arrive at the desired:
	\begin{align}
		\transformerofof{\aflowconstraint_1\statemult\aflowconstraint_F}{\theinflow}
		~\fpreldot~
		\transformerofof{\aflowconstraint_2'\statemult\aflowconstraint_F'}{\theinflow}
		\label{proof:upward-closed-outflow:conclusion-transformer}
	\end{align}

	\medskip
	Observe that \eqref{proof:upward-closed-outflow:footprint-inflow-approx} immediately gives
	\begin{align}
		\aflowconstraint_2' ~\in~ \fpclosureof{\setnodes_F}{\aflowconstraint_2}
		\label{proof:upward-closed-outflow:conclusion-closure-footprint}
		\ .
	\end{align}
	It remains to argue for $\aflowconstraint_F'$.
	Recall that we have $\aflowconstraint_F.\inflow = \theinflow \uplus \restrictto{\aflowconstraint_1.\outflow}{\setnodes_1\times\setnodes_F}$ by \eqref{proof:upward-closed-outflow:multdef}.
	Moreover, we have $\aflowconstraint_F'.\inflow = \theinflow \uplus \restrictto{\aflowconstraint_2'.\outflow}{\setnodes_1\times\setnodes_F}$ by the definition of $\aflowconstraint_F'$.
	Hence, it suffices to show that the sum of inflow $\aflowconstraint_F$ receives from $\aflowconstraint_1$ is $\fprel$-related to the sum of inflow $\aflowconstraint_F'$ receives from $\aflowconstraint_2'$.
	To that end, consider some node $\anodep\in\setnodes_F$.
	Then, we conclude as follows:
	\begin{align*}
		&\sum_{\anode\in\nat\setminus\setnodes_F} \aflowconstraint_F.\inflow(\anode,\anodep)
		\\\explain{by definition}=~~&
		\sum_{\anode\in\setnodes_1} \aflowconstraint_1.\outflow(\anode,\anodep)
		~~~+\sum_{\anodepp\in\overline{\setnodes}} \aflowconstraint_F.\inflow(\anode,\anodep)
		\\\explain{Def. $\transformerof{\dontcare}$}=~~&
		\transformerofof{\aflowconstraint_1}{\aflowconstraint_1.\inflow}(\anodep)
		~~+\sum_{\anodepp\in\overline{\setnodes}} \aflowconstraint_F.\inflow(\anode,\anodep)
		\\\explain{by \eqref{proof:upward-closed-outflow:footprint-inflow-transformer-pre-vs-post}}\fprel~~&
		\transformerofof{\aflowconstraint_2'}{\aflowconstraint_2'.\inflow}(\anodep)
		~~+\sum_{\anodepp\in\overline{\setnodes}} \aflowconstraint_F.\inflow(\anode,\anodep)
		\\\explain{Def. $\transformerof{\dontcare}$}=~~&
		\sum_{\anode\in\setnodes_2'} \aflowconstraint_2'.\outflow(\anode,\anodep)
		~~~+\sum_{\anodepp\in\overline{\setnodes}} \aflowconstraint_F.\inflow(\anode,\anodep)
		\\\explain{by definition}=~~&
		\sum_{\anode\in\nat\setminus\setnodes_F'} \aflowconstraint_F'.\inflow(\anode,\anodep)
	\end{align*}
	By definition, this means:
	\begin{align}
		\aflowconstraint_F' ~\in~ \fpclosureof{\setnodes_2}{\aflowconstraint_F}
		\label{proof:upward-closed-outflow:conclusion-closure-frame}
		\ .
	\end{align}

	\medskip
	This concludes the proof, as \eqref{proof:upward-closed-outflow:conclusion-transformer}, \eqref{proof:upward-closed-outflow:conclusion-closure-footprint}, and \eqref{proof:upward-closed-outflow:conclusion-closure-frame} show the desired properties.
\end{proof}

\begin{proof}[Proof of \Cref{thm:upward-closed-framing}]
	\newcommand{\theinflow}{\inflow}
	\newcommand{\theinflowp}{\overline{\inflow}}
	\newcommand{\aflowconstraintp}{\hat{\aflowconstraint}}
	Consider flow graphs $\aflowconstraint_1,\aflowconstraint_2,\aflowconstraint_F$ with $\aflowconstraint_1\statemultdef \aflowconstraint_F$ and assume $\aflowconstraint_1 \ctxfprel \aflowconstraint_2$ and $\fpcompatible[\fprel]{\aflowconstraint_1,\aflowconstraint_2,\aflowconstraint_F}$.
	Let $\setnodes_1 \defeq \aflowconstraint_1.\setnodes$ and $\setnodes_F \defeq \aflowconstraint_F.\setnodes$ and $\overline{\setnodes} \defeq \nat\setminus(\setnodes_1 \cup \setnodes_F)$ and $\theinflow \defeq (\aflowconstraint_1 \statemult \aflowconstraint_F).\inflow$.
	Observe $\aflowconstraint_2.\setnodes = \setnodes_1$ due to $\aflowconstraint_1 \ctxfprel \aflowconstraint_2$.
	Let $\edges_1\defeq\aflowconstraint_1.\edges$, $\edges_F\defeq\aflowconstraint_F.\edges$, and $\edges_2\defeq\aflowconstraint_2.\edges$.
	Invoke \Cref{thm:upward-closed-outflow} for $\aflowconstraint_1,\aflowconstraint_2,\aflowconstraint_F$ to obtain a flow graph $\aflowconstraint_{2+F}$ with:
	\begin{align}
		\aflowconstraint_{2+F} ~&=~ (\setnodes_1\uplus\setnodes_F,~ \edges_2\uplus\edges_F,~ \theinflow)
		\label{proof:upward-closed-framing:fg-2f-def}
		\\
		\restrictto{\aflowconstraint_{2\statemult F}}{\setnodes_1}\in\fpclosureof{\aflowconstraint_F}{\aflowconstraint_2}
		\quad&\text{and}\quad
		\restrictto{\aflowconstraint_{2\statemult F}}{\setnodes_F}\in\fpclosureof{\aflowconstraint_2}{\aflowconstraint_F}
		\label{proof:upward-closed-framing:fg-2f-closure}
	\end{align}
	Now, choose $\aflowconstraint_2' \defeq \restrictto{\aflowconstraint_{2+F}}{\setnodes_1}$ and $\aflowconstraint_F' \defeq \restrictto{\aflowconstraint_{2+F}}{\setnodes_F}$.
	By \Cref{thm:restriction-vs-statemult}, we have $\aflowconstraint_2'\statemultdef\aflowconstraint_F'$ and $\aflowconstraint_2'\statemult\aflowconstraint_F'=\aflowconstraint_{2+F}$.
	Furthermore, \eqref{proof:upward-closed-framing:fg-2f-closure} immediately gives $\aflowconstraint_2'\in\fpclosureof{\aflowconstraint_F}{\aflowconstraint_2}$ and $\aflowconstraint_F'\in\fpclosureof{\aflowconstraint_2}{\aflowconstraint_F}$.

	\medskip
	It remains to show $\transformerofof{\aflowconstraint_1\statemult\aflowconstraint_F}{\theinflowp} \fpreldot \transformerofof{\aflowconstraint_{2+F}}{\theinflowp}$, for all $\theinflowp\leq\theinflow$.
	Fix some $\theinflowp\leq\theinflow$.
	Define $\aflowconstraintp_1 = \restrictto{(\aflowconstraint_1\statemult\aflowconstraint_F)[\inflow\mapsto\theinflowp]}{\setnodes_1}$ and $\aflowconstraintp_F = \restrictto{(\aflowconstraint_1\statemult\aflowconstraint_F)[\inflow\mapsto\theinflowp]}{\setnodes_F}$.
	By definition:
	\begin{align}
		\aflowconstraintp_1 ~=~ (\setnodes_1,~ \edges_1,~ \theinflowp \uplus \restrictto{\aflowconstraintp_F.out}{\setnodes_F\times\setnodes_1})
		\label{proof:upward-closed-framing:fg-hh1-def}
		\quad\text{and}\quad
		\aflowconstraintp_F ~=~ (\setnodes_F,~ \edges_F,~ \theinflowp \uplus \restrictto{\aflowconstraintp_1.out}{\setnodes_1\times\setnodes_F})
		\ .
	\end{align}
	By \Cref{thm:restriction-vs-statemult} we have $\aflowconstraintp_1\statemultdef\aflowconstraintp_F$ and $\aflowconstraintp_1\statemult\aflowconstraintp_F=(\aflowconstraint_1\statemult\aflowconstraint_F)[\inflow\mapsto\theinflowp]$.
	Towards our proof goal, we apply \Cref{thm:upward-closed-outflow} for $\aflowconstraintp_1,\aflowconstraintp_F,\aflowconstraint_2$.
	Before we can do so, however, we have to show that the \namecref{thm:upward-closed-outflow} is applicable, i.e., that $\aflowconstraintp_1 \ctxfprel \aflowconstraint_2$ holds.

	From \Cref{thm:inflow-leq} and the choice of $\theinflowp$ we know $(\aflowconstraintp_1\statemult\aflowconstraintp_F).\fval \leq (\aflowconstraint_1\statemult\aflowconstraint_F).\fval$.
	Consequently, $\aflowconstraintp_F.\outflow \leq \aflowconstraint_F.\outflow$ by definition.
	This, in turn, means $\aflowconstraintp_1.\inflow \leq \aflowconstraint_1.\inflow$ by \eqref{proof:upward-closed-framing:fg-hh1-def}.
	Hence, for all $\inflow'\leq\aflowconstraintp_1.\inflow$, we have $\transformerofof{\aflowconstraint_1}{\inflow'}\fpreldot\transformerofof{\aflowconstraint_2}{\inflow'}$ by $\aflowconstraint_1\ctxfprel\aflowconstraint_2$ from the premise.
	By definition of $\transformerof{\dontcare}$, we obtain $\transformerofof{\aflowconstraintp_1}{\inflow'}\fpreldot\transformerofof{\aflowconstraint_2}{\inflow'}$ for all $\inflow'\leq\aflowconstraintp_1.\inflow$.
	That is, $\aflowconstraintp_1\ctxfprel\aflowconstraint_2$.

	Now, we are ready to apply \Cref{thm:upward-closed-outflow} to $\aflowconstraintp_1,\aflowconstraintp_F,\aflowconstraint_2$.
	We get $\aflowconstraintp_{2+F}$ with
	\begin{align}
		\aflowconstraintp_{2+F} ~&=~ (\setnodes_1\uplus\setnodes_F,~ \edges_2\uplus\edges_F,~ \theinflowp)
		\label{proof:upward-closed-framing:fg-h2f-def}
		\\
		\transformerofof{\aflowconstraintp_1\statemult\aflowconstraintp_F}{\theinflowp} ~&\fpreldot~ \transformerofof{\aflowconstraintp_{2+F}}{\theinflowp}
		\label{proof:upward-closed-framing:fg-h2f-transformer}
	\end{align}
	As noted earlier, we have $(\aflowconstraintp_1\statemult\aflowconstraintp_F)=(\aflowconstraint_1\statemult\aflowconstraint_F)[\inflow\mapsto\theinflowp]$.
	Moreover, \eqref{proof:upward-closed-framing:fg-2f-def} combined with \eqref{proof:upward-closed-framing:fg-h2f-def} gives $\aflowconstraintp_{2+F}=\aflowconstraint_{2+F}[\inflow\mapsto\theinflowp]$.
	Hence, \eqref{proof:upward-closed-framing:fg-h2f-transformer} yields $\transformerofof{\aflowconstraint_1\statemult\aflowconstraint_F}{\theinflowp}\fpreldot\transformerofof{\aflowconstraint_{2+F}}{\theinflowp}$, as required.
	This concludes $\aflowconstraint_1\statemult\aflowconstraint_F\ctxfprel\aflowconstraint_2'\statemult\aflowconstraint_F'$.
\end{proof}

\begin{proof}[Proof of \Cref{thm:easy-fprel-lfp}]
	Follows immediately by choosing $\amonval_i=f^i(\bot)$ and $\amonvalp_j=g^j(\bot)$.
\end{proof}

\begin{proof}[Proof of \Cref{thm:fprelcompatible-sub-omega-cpo}]
	Recall that $\fprel$ is a sub-$\omega$-cpo if $\fprel$
	\begin{inparaenum}
		\item is an $\omega$-cpo such that
		\item $\fprel \subseteq \leq$ and
		\item $\sup_\fprel(\achain)=\sup_\leq(\achain)$ for all $\fprel$-ascending chains $\achain=\amonval_0\fprel\amonval_1\fprel\cdots$.
	\end{inparaenum}
	(Note $\sup_\leq(M)=\bigjoin M$.)

	Now, consider functions $f,g:(\aflowconstraint.\setnodes{\to}\amonoid)\to(\aflowconstraint.\setnodes{\to}\amonoid)$ that are $\leq$-continuous and $\fpreldot$-monotonic such that we have $f^i(\bot)\leq f^{i+1}(\bot)$ and $g^i(\bot)\leq g^{i+1}(\bot)$ and $f^i(\bot)\fpreldot g^i(\bot)$ for all $i\in\nat$, as well as $\lfpof{f}\fpreldot g(\lfpof{f})$.
	Let $\achain=f^0(\bot)\leq f^1(\bot)\leq\cdots$ and $\achainp=g^0(\bot)\leq g^1(\bot)\leq\cdots$.
	By \Cref{thm:fixpoint-kleene}, $\lfpof{f}=\bigjoin\achain$ and $\lfpof{g}=\bigjoin\achainp$ exist.
	So we have:
	\begin{gather}
		\lfpof{f} ~=~ \bigjoin\achain ~\leq~ \bigjoin\achainp ~=~ \lfpof{g}
		\label{proof:fp-approx}
		\\
		\bigjoin\achain ~\leq~ g(\bigjoin\achain)
		\quad\text{and}\quad
		\bigjoin\achain ~\fpreldot~ g(\bigjoin\achain)
		\label{proof:g-iter-init}
	\end{gather}
	Moreover, we get $g^i(\bigjoin\achain) \fpreldot g^{i+1}(\bigjoin\achain)$, for all $i\in\nat$, because $g$ is $\fpreldot$-monotonic.
	As a consequence, $\achainp'=g^0(\bigjoin\achain)\fpreldot g^1(\bigjoin\achain)\fpreldot\cdots$ is a $\fpreldot$-ascending chain.
	Because $\fprel$ is an $\omega$-cpo, this means $\bigjoin\achainp'$ exists.
	By definition, we obtain $\bigjoin\achain=g^0(\bigjoin\achain)\fpreldot\bigjoin\achainp'$.
	To conclude the overall claim, it now suffices to show that $\bigjoin\achainp'=\bigjoin\achainp$ holds.
	This, in turn, holds if $\lfpof{g}=\bigjoin\achainp'$.

	\medskip
	We first show that $\bigjoin\achainp'$ is a fixed point of $g$.
	\begin{align*}
		&~g(\bigjoin\achainp')
		\\\explain{Def. $\achainp'$}=&~
		g(\bigjoin\setcond{g^i(\join\achain)}{i\in\nat})
		\\\explain{$g$ continuous}=&~
		\bigjoin\setcond{g^{i+1}(\join\achain)}{i\in\nat}
		\\\explain{by \eqref{proof:g-iter-init}}=&~
		\bigjoin\setcond{g^{i+1}(\join\achain)}{i\in\nat} \join \bigjoin\achain
		\\\explain{$g^0(\join\achain)=\join\achain$}=&~
		\bigjoin\setcond{g^i(\join\achain)}{i\in\nat}
		\\\explain{Def. $\achainp'$}=&~
		\bigjoin\achainp'
		\ .
	\end{align*}

	\medskip\newcommand{\dd}{\cval^\dagger}
	We now show that $\bigjoin\achainp'$ is the least fixed point of $g$.
	To that end, consider another fixed point $\dd$ of $g$, i.e., $g(\dd)=\dd$.
	It suffices to show that $g^i(\bigjoin\achain) \leq \dd$ holds for all $i$, because this implies that the join over the $g^i(\bigjoin\achainp')$ is at most $\dd$.
	We proceed by induction.
	In the base case, $i=0$, we have \[
		g^0(\join\achain)
		~=~
		\bigjoin\achain
		~\explains{\eqref{proof:fp-approx}}{=}~
		\lfpof{g}
		~\leq~
		\dd
	\]
	where the last approximation holds by the definition of $\lfpof{g}$ together with the fact that $\dd$ is a fixed point of $g$.
	For the induction step, assume $g^i(\bigjoin\achain)\leq\dd$.
	We have \[
		g^{i+1}(\join\achain) ~=~ g(g^{i}(\join\achain)) ~\leq~ g(\dd) ~=~ \dd
	\]
	where the approximation is by induction together with $g$ begin $\leq$-monotonic by \Cref{thm:continuous-implies-monotonic-special} and the last equality is by the fact that $\dd$ is a fixed point of $g$.
	Altogether, we conclude the desired equality: $\lfpof{g}=\bigjoin\setcond{g^i(\bigjoin\achain)}{i\in\nat}$.

	Overall, we conclude the desired $\bigjoin\achain\fpreldot\bigjoin\achainp'=\lfpof{g}=\bigjoin\achainp$.
\end{proof}

\begin{proof}[Proof of \Cref{thm:fprelcompatible-acc}]
	Recall that $(\amonoid,\leq)$ satisfies the ascending chain condition if every chain $\leq$-ascending chains $\achain=\amonval_0\leq\amonval_1\leq\cdots$ become stationary, that is, there is some $i\in\nat$ such that $\amonval_{i}=\amonval_{i+j}$ holds for all $j\in\nat$.

	Now, consider $\leq$-ascending chains $\achain=\amonval_0\leq\amonval_1\leq\cdots$ and $\achainp=\amonvalp_0\leq\amonvalp_1\leq\cdots$ with $\amonval_k\fprel\amonvalp_k$ for all $k\in\nat$.
	Let $i\in\nat$ such that $\forall i'\in\nat.~ \amonval_i=\amonval_{i+i'}$.
	Let $j\in\nat$ such that $\forall j'\in\nat.~ \amonvalp_j=\amonvalp_{j+j'}$.
	By the ascending chain condition, $i$ and $j$ exist.
	Choose $k=\max(i,j)$.
	Then, $\forall i'\in\nat.~ \amonval_k=\amonval_{k+i'}$ and $\forall j'\in\nat.~ \amonval_k=\amonvalp_{k+j'}$.
	This means that $\bigjoin\achain=\amonval_k$ and $\bigjoin\achainp=\amonvalp_k$.
	By assumption, $\amonval_k\fprel\amonvalp_k$.
	Hence, $\bigjoin\achain\fprel\bigjoin\achainp$.
	Then, \Cref{thm:easy-fprel-lfp} establishes \eqref{def:fpcompat:lfp}.
\end{proof}

\begin{proof}[Proof of \Cref{thm:trivial-choices}]
	Property $\fpcompatible[=]{\aflowconstraint}$ is trivially true.
	We show $\fpcompatible[\leq]{\aflowconstraint}$.
	Because $\leq$ is the natural order, we immediately have \eqref{def:fpcompat:refl-trans}.
	\Cref{thm:our-monoid-has-addprop} gives \eqref{def:fpcompat:add}.
	Because edge functions are $\leq$-continuous, they are also $\leq$-monotonic by \Cref{thm:continuous-implies-monotonic-special}.
	This is \eqref{def:fpcompat:edges}.
	Finally, it is easy to see that $\leq$ is a sub-$\omega$-cpo of itself, so \eqref{def:fpcompat:lfp} follows from \Cref{thm:fprelcompatible-sub-omega-cpo}.
\end{proof}


\subsection{Proofs for the Instantiation}

Let $\astate_1\imult\astate_2=\astatep$. Then there are $\astatep_1$ and $\astatep_2$ with $\astatep=\astatep_1\mstar\astatep_2$ and $\astate_1.\setnodes=\astatep_1.\setnodes$ and $\astate_2.\setnodes=\astatep_2.\setnodes$. Moreover, the flow graphs $\astatep_1$ and $\astatep_2$ are unique. 

\begin{proof}[Proof of \Cref{thm:unique-ghost-decomposition}]
	\label{proof:unique-ghost-decomposition}
	Consider $\afg_1\imult\afg_2=\afgc$ with $\afg_i=(\setnodes_i,\edges_i,\inflow_i)$.
	Choose $\afgp_1=\restrictto{\afgc}{\setnodes_1}$ and $\afgp_2=\restrictto{\afgc}{\setnodes_2}$.
	By definition, $\afg_1.\setnodes=\afgp_1.\setnodes$ and $\afg_1.\edges=\afgp_1.\edges$.
	Similarly, $\afg_2.\setnodes=\afgp_2.\setnodes$ and $\afg_2.\edges=\afgp_2.\edges$.
	Moreover, \cref{thm:restriction-vs-statemult} gives $\afgp_1\statemultdef\afgp_2$ and $\afgp_1\statemult\afgp_2=\afgc$.
	This establishes the first claim.

	It remains that the decomposition of $\afgc$ into $\afgp_1$ and $\afgp_2$ is unique.
	Towards a contradiction, assume there are $\afgp_1',\afgp_2'$ such that $\afgp_1\neq\afgp_1'$ and $\afgp_2\neq\afgp_2'$ and $\afgp_1'\statemult\afgp_2'=\afgc$.
	By the definition of the multiplication, we have $\afgp_i.\setnodes=\afgp_i'.\setnodes$ and $\afgp_i.\edges=\afgp_i'.\edges$, for $i\in\set{1,2}$.
	That is, $\afgp_1.\inflow\neq\afgp_1'.\inflow$ or $\afgp_2.\inflow\neq\afgp_2'.\inflow$.
	Wlog. assume $\afgp_1.\inflow\neq\afgp_1'.\inflow$.
	By definition, $(\afgp_1\statemult\afgp_2).\inflow=\afgc.\inflow=(\afgp_1'\statemult\afgp_2').\inflow$.
	Hence, $\afgp_1.\inflow\neq\afgp_1'.\inflow$ means that there is a pair of nodes $\anode\in\afgp_1.\setnodes,\,\anodep\in\afgp_2.\setnodes$ that witnesses the inequality, $\afgp_1.\inflow(\anodep,\anode)\neq\afgp_1'.\inflow(\anodep,\anode)$.
	From the choice of $\afgp_1$, we know that $\afgp_1.\inflow(\anodep,\anode)=\afgp_2.\outflow(\anodep,\anode)=\afgp_2.\edgesatof{\anodep}{\anode}{\afgp_2.\fvalof{\anodep}}$.
	By definition, $\afgp_2.\fvalof{\anodep}=(\afgp_1.\fval\uplus\afgp_2.\fval)(\anodep)$.
	By $\afgp_1\statemultdef\afgp_2$ then, $\afgp_2.\fvalof{\anodep}=(\afgp_1\imult\afgp_2).\fvalof{\anodep}=\afgc.\fvalof{\anodep}$.
	Combined, $\afgp_1.\inflow(\anodep,\anode)=\afgp_2.\edgesatof{\anodep}{\anode}{\afgc.\fvalof{\anodep}}$.
	Similarly, we obtain $\afgp_1'.\inflow(\anodep,\anode)=\afgp_2'.\edgesatof{\anodep}{\anode}{\afgc.\fvalof{\anodep}}$.
	Because $\afgp_2.\edges=\afgp_2'.\edges$ by assumption, we conclude $\afgp_1.\inflow(\anodep,\anode)=\afgp_1'.\inflow(\anodep,\anode)$.
	This contradicts the earlier $\afgp_1.\inflow(\anodep,\anode)\neq\afgp_1'.\inflow(\anodep,\anode)$.
\end{proof}

\begin{proof}[Proof of \Cref{Lemma:FlowAlgebra}]
	See \cite[Lemma 2]{DBLP:conf/tacas/MeyerWW23}
\end{proof}

\begin{proof}[Proof of \Cref{Lemma:MultCoincides}]
	Follows immediately from the definition of the multiplication $\mstar$.
\end{proof}

\begin{proof}[Proof of \Cref{thm:instantiation-closure}]
	Consider flow graphs $\afg,\afgp,\afgc$ with $\afg\ctxfprel\afgp$ and, $\afg\statemultdef\afgc$.
	Because $\fprel$ is an estimator along the lines of \cref{sec:instantiation}, we have $\fpcompatible[\fprel]{\afg}$, $\fpcompatible[\fprel]{\afgp}$, and $\fpcompatible[\fprel]{\afgc}$.
	Then, \cref{thm:upward-closed-framing} for $\afg,\afgp,\afgc$ yields $\afgp'\in\fpclosureof[\fprel]{\afgc}{\afgp}$ and $\afgc'\in\fpclosureof[\fprel]{\afgp}{\afgc}$ such that $\afgp'\statemultdef\afgc'$ and  $\afg \statemult \afgc \ctxfprel \afgp'\statemult\afgc'$.
	By \cref{Lemma:MultCoincides}, we have $\afgp'\statemult\afgc'=\afgp'\imult\afgc'$.
	By the definition of $\fpclosureof[\fprel]{\afgc}{\afgp}$, we have $\afgp'=(\afgp.\setnodes,\afgp.\edges,\inflow_{\afgp'})$ with $\afgp.\inflow\fprelup{\afgc.\setnodes}\inflow_{\afgp'}$.
	Together with $\afg.\inflow=\afgp.\inflow$ from $\afg\ctxfprel\afgp$, we get $\afg.\inflow\fprelup{\afgc.\setnodes}\inflow_{\afgp'}$.
	Similarly, $\afgc'=(\afgc.\setnodes,\afgc.\edges,\inflow_{\afgc'})$ with $\afgc.\inflow\fprelup{\afgp.\setnodes}\inflow_{\afgc'}$.
	This means that $\afgp$ and $\afgp'$ agree on their inflow except for the portion from $\afgc$, and similarly $\afgc$ and $\afgc'$ agree on their inflow except for the portion from $\afgp$.
	Since the ghost multiplication removes this, we obtain $\afgp\imult\afgc=\afgp'\imult\afgc'$.
	Altogether, we arrive at the desired correspondence: \[
		\afg\statemult\afgc
		~\ctxfprel~
		\afgp'\statemult\afgc'
		~=~
		\afgp'\imult\afgc'
		~=~
		\afgp\imult\afgc
		\ .
		\qedhere
	\]
\end{proof}

\begin{proof}[Proof of \Cref{thm:instantiation-physical}]
	The first claim holds by definition.
	For the second claim, assume $\absupof{\acom}{\afg}\neq\abort$ and $\afg\statemultdef\afgc$.
	We show $\absupof{\acom}{\afg\statemult\afgc}=\absupof{\acom}{\afg}\imult\afgc$.
	To that end, it suffices to show that $\upof{\acom}{\afg\statemult\afgc}$ does not abort and its states satisfy the estimator requirement.
	Indeed, then
	\begin{align*}
		&\absupof{\acom}{\afg\statemult\afgc}
		\\
		\explain{Definition}=~~&
		\upof{\acom}{\afg\statemult\afgc}
		\\
		\explain{$\absupof{\acom}{\afg}\neq\abort$ implies $\upof{\acom}{\afg}\neq\abort$}=~~&
		\upof{\acom}{\afg}\imult\afgc
		\\
		\explain{$\absupof{\acom}{\afg}\neq\abort$}=~~&
		\absupof{\acom}{\afg}\imult\afgc. 
	\end{align*}
	That $\upof{\acom}{\afg\statemult\afgc}$ does not abort follows from $\upof{\acom}{\afg}\neq\abort$.
	For the estimator requirement, let $\afgp\imult\afgc\in \upof{\acom}{\afg\statemult\afgc}=\upof{\acom}{\afg}\imult\afgc$.
	We have to show $\afg\statemult\afgc\ctxfprel\afgp\imult\afgc$.
	Since $\absupof{\acom}{\afg}\neq\abort$, we can rely on $\afg\ctxfprel\astatep$.
	Then \cref{thm:instantiation-closure} concludes the argument.
\end{proof}

\begin{proof}[Proof of \Cref{thm:instantiation-ghost}]
	\label{proof:instantiation-ghost}
	Consider some flow graphs $\afgp,\afgc$.
	If $\afgp\absimult\afgc = \top$, there is nothing to show.
	So assume $\afgp\absimult\afgc\neq\top$.
	This means $\ghostabsof{\afgc}{\afgp}\neq\top$ and $\ghostabsof{\afgp}{\afgc}\neq\top$.
	Wlog. this means that there is some flow graph $\afg$ with $\afg\ctxfprel\afgp$ and $\afg\statemultdef\afgc$, by the definition of $\ghostabs{\afgc}$.
	Then, or \cref{thm:instantiation-closure} gives $\afgp\imult\afgc=\afgp[\inflow\mapsto\inflow_\afgp]\statemult\afgc[\inflow\mapsto\inflow_\afgc]$ with $\afgp.\inflow\fprelup{\afgc.\setnodes}\inflow_\afgp$ and $\afgc.\inflow\fprelup{\afgp.\setnodes}\inflow_\afgc$.
	By definition, we conclude: \[
		\afgp\imult\afgc
		~=~
		\afgp[\inflow\mapsto\inflow_\afgp]\statemult\afgc[\inflow\mapsto\inflow_\afgc]
		~\in~
		\ghostabsof{\afgc}{\afgp}\statemult\ghostabsof{\afgp}{\afgc}
		~=~
		\astatep\absimult \astatepp
		\ .
		\qedhere
	\]
\end{proof}

\end{document}